\documentclass[pra,twocolumn,superscriptaddress,nofootinbib]{revtex4}
\usepackage[letterpaper, left=0.8in, right=0.8in, top=0.9in, bottom=0.9in]{geometry}
\usepackage{titlesec}

\usepackage{cmap} 
\usepackage[utf8]{inputenc}
\usepackage[english]{babel}
\usepackage[T1]{fontenc}
\usepackage{amsmath}
\usepackage{amsfonts}
\usepackage{bm}
\usepackage{xcolor}
\definecolor{blueviolet}{rgb}{0.2, 0.2, 0.6}
\definecolor{webgreen}{rgb}{0,.5,0}
\definecolor{webbrown}{rgb}{.6,0,0}
\usepackage[pdftex,
	bookmarks=false,
	colorlinks=true, 
	urlcolor=webbrown, 
	linkcolor=blueviolet, 
	citecolor=webgreen,
	pdfstartpage=1,
	pdfstartview={FitH},  
	bookmarksopen=false
	]{hyperref}
\usepackage{tikz}
\usepackage{natbib}
\usepackage{mathtools}
\usepackage{graphicx}
\usepackage{csquotes}
\usepackage{braket}
\usepackage{dsfont}
\usepackage{listings}
\allowdisplaybreaks

\usepackage{multirow}
\usepackage{amssymb}

\usepackage{amsthm}

\usepackage{color}
\usepackage{nicefrac}

\DeclareFixedFont{\ttb}{T1}{txtt}{bx}{n}{9} 
\DeclareFixedFont{\ttm}{T1}{txtt}{m}{n}{9}  

\usepackage{color}
\definecolor{deepblue}{rgb}{0,0,0.5}
\definecolor{deepred}{rgb}{0.6,0,0}
\definecolor{deepgreen}{rgb}{0,0.5,0}

\newcommand\pythonstyle{\lstset{
language=Python,
basicstyle=\ttm,
morekeywords={self},              
keywordstyle=\ttb\color{deepblue},
emph={MyClass,__init__},          
emphstyle=\ttb\color{deepred},    
stringstyle=\color{deepgreen},
frame=tb,                         
showstringspaces=false
}}

\lstnewenvironment{python}[1][]
{
\pythonstyle
\lstset{#1}
}
{}

\newcommand\pythoninline[1]{{\pythonstyle\lstinline!#1!}}

\definecolor{orange}{RGB}{255,127,0}

\def\bra#1{\ensuremath{\mathinner{\langle{#1}|}}}
\def\ket#1{\ensuremath{\mathinner{|{#1}\rangle}}}

\newcommand{\ketbra}[2]{\lvert #1 \rangle \! \langle #2 \rvert}
\newcommand{\norm}[1]{\left\lVert#1\right\rVert}

\newcommand{\cA}{{\mathcal{A}}}

\newcommand{\cB}{{\mathcal{B}}}
\newcommand{\cT}{{\mathcal{T}}}

\newcommand{\cX}{{\mathcal{X}}}

\newcommand{\cD}{{\mathcal{D}}}
\newcommand{\cE}{{{\mathcal{E}}}}

\newcommand{\cN}{{{\mathcal{N}}}}

\newcommand{\cS}{{\mathcal{S}}}

\newcommand{\cU}{{{\mathcal{U}}}}

\newcommand{\euler}{\mathrm{e}}
\newcommand{\rmi}{\mathrm{i}}

\newcommand{\wt}{\widetilde}

\DeclareMathOperator{\Tr}{tr}
\DeclareMathOperator*{\E}{{\mathbb{E}}}

\newtheorem{theorem}{Theorem}
\newtheorem{corollary}{Corollary}
\newtheorem{definition}{Definition}
\newtheorem*{definition21}{Definition 6.1 of~\cite{chen2021exponential}}
\newtheorem{lemma}{Lemma}

\newtheorem{task}{Task}

\newcommand{\Id}{I}

\usepackage[noend]{algpseudocode}
\usepackage{algorithm,algorithmicx}

\algrenewcommand\alglinenumber[1]{\sf\scriptsize\color{blue}{#1}}
\algrenewcommand\algorithmicrequire{\textbf{Input:}}
\algrenewcommand\algorithmicensure{\textbf{Output:}}

\begin{document}

\title{Quantum advantage in learning from experiments}

\author{Hsin-Yuan Huang}
    \email{hsinyuan@caltech.edu}
	\affiliation{Institute for Quantum Information and Matter, Caltech, Pasadena, CA, USA}
	\affiliation{Department of Computing and Mathematical Sciences, Caltech, Pasadena, CA, USA}
\author{Michael Broughton}
	\affiliation{Google Quantum AI, 340 Main Street, Venice, CA 90291, USA}
\author{Jordan Cotler}
	\affiliation{Harvard Society of Fellows, Cambridge, MA 02138 USA}
	\affiliation{Black Hole Initiative, Cambridge, MA 02138 USA}
\author{Sitan Chen}
    \affiliation{Department of Electrical Engineering and Computer Science, University of California, Berkeley, Berkeley, CA, USA}
    \affiliation{Simons Institute for the Theory of Computing, Berkeley, CA, USA}
\author{Jerry Li}
	\affiliation{Microsoft Research AI, Redmond, WA 98052, USA}
\author{Masoud Mohseni}
    \affiliation{Google Quantum AI, 340 Main Street, Venice, CA 90291, USA}    
\author{Hartmut Neven}
    \affiliation{Google Quantum AI, 340 Main Street, Venice, CA 90291, USA}    
\author{Ryan Babbush}
    \affiliation{Google Quantum AI, 340 Main Street, Venice, CA 90291, USA}
\author{Richard Kueng}
	\affiliation{Institute for Integrated Circuits, Johannes Kepler University Linz, Austria}
\author{John Preskill}
	\affiliation{Institute for Quantum Information and Matter, Caltech, Pasadena, CA, USA}
	\affiliation{Department of Computing and Mathematical Sciences, Caltech, Pasadena, CA, USA}
	\affiliation{AWS Center for Quantum Computing, Pasadena, CA 91125, USA}
\author{Jarrod R.~McClean}
    \email{jmcclean@google.com}
	\affiliation{Google Quantum AI, 340 Main Street, Venice, CA 90291, USA}
\date{\today}

\begin{abstract}
Quantum technology has the potential to revolutionize how we acquire and process experimental data to learn about the physical world.
An experimental setup that transduces data from a physical system to a stable quantum memory, and processes that data using a quantum computer, could have significant advantages over conventional experiments in which the physical system is measured and the outcomes are processed using a classical computer.
We prove that, in various tasks, quantum machines can learn from exponentially fewer experiments than those required in conventional experiments.
The exponential advantage holds in predicting properties of physical systems, performing quantum principal component analysis on noisy states, and learning approximate models of physical dynamics.
In some tasks, the quantum processing needed to achieve the exponential advantage can be modest; for example, one can simultaneously learn about many noncommuting observables by processing only two copies of the system.
Conducting experiments with up to 40 superconducting qubits and 1300 quantum gates, we demonstrate that a substantial quantum advantage can be realized using today’s relatively noisy quantum processors.
Our results highlight how quantum technology can enable powerful new strategies to learn about nature. 
\end{abstract}

\maketitle

\section{Introduction}

Humans learn about nature by doing experiments, but up until now our ability to acquire knowledge has been hampered by viewing the quantum world through a classical lens. The rapid advance of quantum technology portends an opportunity to observe the world in a fundamentally different and more powerful way.
Instead of measuring physical systems and then processing the classical measurement outcomes to infer properties of the physical systems, quantum sensors \cite{degen2017quantum} will eventually be able to transduce~\cite{lauk2020perspectives} quantum information in physical systems directly to a quantum memory \cite{lvovsky2009optical,dennis2002topological}, where it can be processed by a quantum computer.
Fig.~\ref{fig:Cartoon}(a) illustrates the distinction between \emph{conventional} and \emph{quantum-enhanced} experiments. For example, in a quantum-enhanced experiment, multiple photons might be captured and stored coherently at each node of a quantum network, and then processed coherently to extract an informative signal ~\cite{PhysRevLett.109.070503,Bland-Hawthorn:21,giovannetti2011advances}.

Recently, mathematical analyses done by some of the authors show that there exist properties of an $n$-qubit system that a quantum machine can learn efficiently, while the required number of conventional experiments to achieve the same task is \emph{exponential} in $n$~\cite{huang2021information, aharonov2021quantum}. 
This exponential advantage contrasts sharply with the quadratic advantage achieved in many previously proposed strategies for improving sensing using quantum technology \cite{degen2017quantum}.
In this article, we propose and analyze three classes of learning tasks with exponential quantum advantage,
and report on proof-of-principle experiments using up to $40$ qubits on a Google Sycamore processor~\cite{arute2019quantum}.
These experiments confirm that a substantial quantum advantage can be realized even when the quantum memory and processor are both noisy.

\begin{figure*}[t]
\centering
\includegraphics[width=0.99\textwidth]{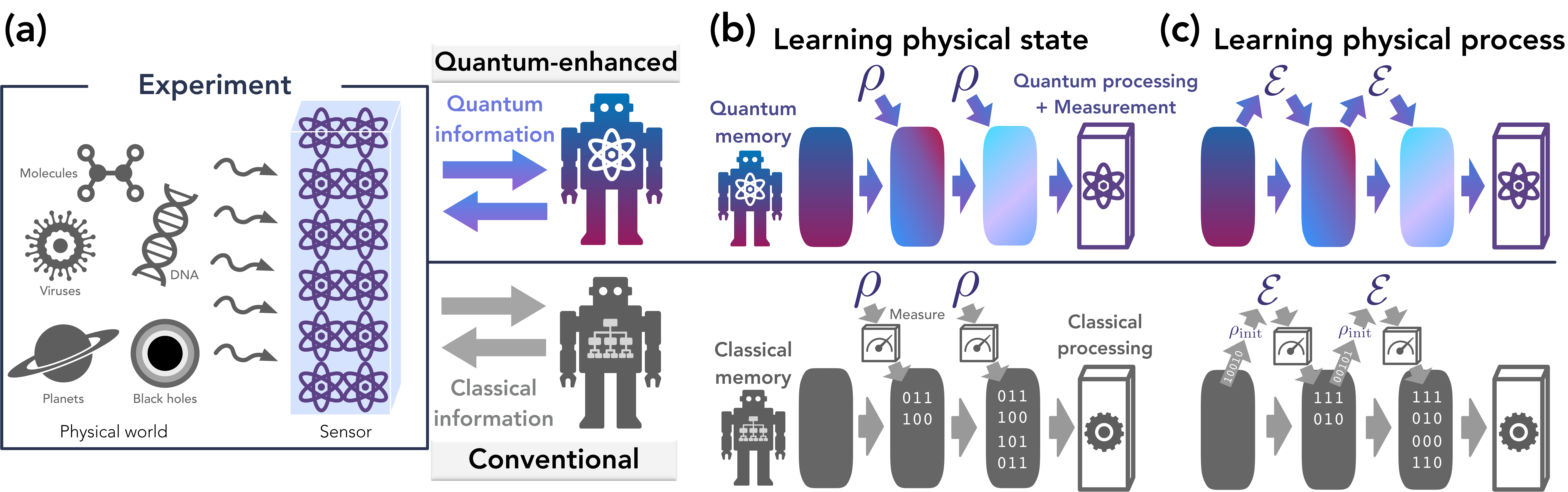}
    \caption{
    \emph{(a) Quantum-enhanced experiments versus conventional experiments.} Quantum-enhanced/conventional experiments interface with a quantum/classical machine running a quantum/classical learning algorithm that can store and process quantum/classical information.
    \emph{(b) Learning physical state $\rho$.}
    Each experiment produces a physical state $\rho$.
    In the conventional setting, we measure each $\rho$ to obtain classical data (the measurement could depend on prior measurement outcomes) and store the data in a classical memory.
    In the quantum-enhanced setting, $\rho$ can coherently alter the quantum information stored in the memory of the quantum machine (illustrated by the change in color).
    With large enough quantum memory, the quantum machine can simply store each copy of $\rho$.
    After multiple rounds of experiments, quantum processing followed by a measurement is performed on the quantum memory.
    \emph{(c) Learning physical process $\cE$. } Each experiment is an evolution under $\cE$.
    In the conventional setting, the classical machine specifies the input state to $\cE$ using a classical bitstring, and obtains classical measurement data \cite{mohseni_qpt_2008}.
    In the quantum-enhanced setting, the evolution $\cE$ coherently alters the memory of the quantum machine: the input state to $\cE$ is entangled with the quantum memory in the quantum machine and the output state is retrieved coherently by the quantum machine.
    \label{fig:Cartoon}
    }
\end{figure*}

To be more concrete, suppose that each experiment generates an $n$-qubit state $\rho$, and our goal is to learn some property of $\rho$; see Fig.~\ref{fig:Cartoon}.
We depict \textit{conventional} and \textit{quantum-enhanced} experiments for this scenario in Fig.~\ref{fig:Cartoon}(b).
In conventional experiments, each copy of $\rho$ is measured separately, the measurement data is stored in a classical memory, and a classical computer outputs a prediction for the property after processing the classical data.
In quantum-enhanced experiments, each copy of $\rho$ is stored in a quantum memory, and then the quantum machine outputs the prediction after processing the quantum data in the quantum memory.
We prove that for some tasks, the number of experiments needed to learn a desired property is exponential in $n$ using the conventional strategy, but only polynomial in $n$ using the quantum-enhanced strategy.
For suitably defined tasks, we can achieve exponential quantum advantage using a protocol as simple as storing two copies of $\rho$ in quantum memory and performing an entangling measurement.
We also show that quantum-enhanced experiments have a similar exponential advantage in a related scenario shown in Fig.~\ref{fig:Cartoon}(c), in which the goal is to learn about a quantum process $\cE$ rather than a quantum state $\rho$.

Building on observations in \cite{huang2021information, chen2021exponential} we prove that for a task that entails acquiring information about a large number of non-commuting observables, quantum-enhanced experiments can have an exponential advantage even when the measured quantum state is unentangled. By performing experiments with up to 40 superconducting qubits, we show that this quantum advantage persists even when using currently available quantum processors.
We also demonstrate quantum advantage in learning the symmetry class of a physical evolution operator, inspired by recent theoretical advances \cite{aharonov2021quantum, chen2021exponential}.
Finally, in a theoretical contribution, we rigorously prove that quantum-enhanced experiments have an exponential advantage in learning about the principal component of a noisy state, as previously indicated in \cite{lloyd2014quantum}.

In our proof-of-principle experiments, we directly execute the state preparation or process to be learned within the quantum processor. In an actual application, the quantum data analyzed by the learning algorithm might be produced by an analog quantum simulator or a gate-based quantum computer.
We also envision future applications in which quantum sensors equipped with quantum processors interact coherently with the physical world.
The robustness of quantum advantage with respect to noise, validated by our experiments using a noisy superconducting device, boosts our confidence that the quantum-enhanced strategies described here can be exploited someday to achieve a substantial advantage in realistic applications. 

\section{Provable quantum advantage}
Here we present three classes of learning tasks and the associated quantum-enhanced experiments, each yielding a provable exponential advantage over conventional experiments.
Each result is encapsulated by a theorem which we state informally. Precise statements and proofs are presented in the Supplemental Material.
Our experimental demonstrations are discussed in Sec.~\ref{sec:experimental}.
The proofs proceed by representing a classical algorithm with a decision tree depicted at the center of the gray robot in Fig.~\ref{fig:Cartoon}. The tree representation encodes how the classical memory changes as we obtain more experimental data. 
We then analyze how the transitions on the tree differ for distinct measured physical systems to provide rigorous information-theoretic lower bounds.
A general mathematical framework building on~\cite{chen2021exponential} is given in Appendix~\ref{sec:separation} .

The first task concerns learning about a physical system described by an $n$-qubit state $\rho$.
We suppose that each experiment generates one copy of $\rho$.
In the conventional setting, we measure each copy of $\rho$ to obtain classical data.
The procedure can be adaptive, that is, each measurement can depend on the data obtained in earlier measurements.
In the quantum-enhanced setting, a quantum computer can store each copy of $\rho$ in a quantum memory, and act jointly on multiple copies of $\rho$.
In both scenarios, we require all quantum data to be measured at the end of the learning phase of the procedure, so that only classical data survives.
After the learning is completed, the learner is asked to provide an accurate prediction for the expectation value of one observable drawn from a set $\{O_1, O_2, \dots \}$ where the number of observables in the set is exponentially large in $n$.
The observables in the set can be highly incompatible, that is, each observable may fail to commute with many others in the set.

In prior work by some of the authors~\cite{huang2021information,chen2021exponential}, we required the learner to predict exponentially many observables, which is not possible in practice if the system size is large.
In order to demonstrate the advantage in an actual device, we prove that predicting just the absolute value of one observable requires exponentially many copies in the \emph{conventional} scenario. In contrast, predicting the entire set of observables can be achieved with a polynomial number of copies in the \emph{quantum-enhanced} scenario.
We thereby establish the following constant versus exponential separation.
The proof is given in Appendix~\ref{sec:shadow}.

\begin{theorem}[Predicting observables]\label{thm:observables}
There exists a distribution over $n$-qubit states 
and a set of observables such that in the conventional scenario, at least order $2^n$ experiments are needed to predict the absolute value of one observable selected from the set, while a constant number of experiments suffice in the quantum-enhanced scenario.
\end{theorem}

The exponential quantum advantage can occur even if the state $\rho$ is unentangled. For example, in our experiments we consider $\rho\propto (I + \alpha P)$ where $P$ is an $n$-qubit Pauli operator and $\alpha \in (-1, 1)$. This state
can be realized as a probabilistic ensemble of product states, each of which is an eigenstate of $P$ with eigenvalue $\alpha$. Even if the state is known to be of this form, but $P, \alpha$ are unknown, the exponential separation between conventional and quantum-enhanced experiments persists.
Moreover, the quantum advantage can be achieved by performing simple entangling measurements on pairs of copies of $\rho$.
That the quantum advantage applies even when correlations among the $n$ qubits are classical leads us to believe that the quantum-enhanced strategy will be beneficial in a broad class of sensing applications.
In Appendix~\ref{sec:hard-bdd-mem}, we extend this theorem, showing that a sufficiently large quantum memory is needed to achieve this task in the quantum-enhanced scenario.

Our second machine learning task with a quantum advantage is quantum principal component analysis (PCA) \cite{lloyd2014quantum}.
In this task, each experiment produces one copy of $\rho$, and our goal is to predict properties of the (first) principal component of $\rho$, namely the eigenstate $\ket{\psi}$ of $\rho$ with the largest eigenvalue.
For example, we may want to predict the expectation values of a few observables in the state $\ket{\psi}$.
This task may become a valuable ingredient in future quantum-sensing applications. If an imperfect quantum sensor transduces a detected quantum state into quantum memory, the state is likely to be corrupted by noise. But it is reasonable to expect that properties of the principal component are relatively robust with respect to noise~\cite{koczor2021dominant}, and therefore highly informative about the uncorrupted state.
To perform quantum PCA, a learning algorithm was introduced in \cite{lloyd2014quantum} based on phase estimation which requires fault-tolerant quantum computers.
One can also obtain information about the principal component of $\rho$ using more near-term algorithms, such as virtual cooling \cite{cotler2019quantum} and virtual distillation \cite{huggins2020virtual}.

While the quantum PCA algorithm in \cite{lloyd2014quantum} is exponentially faster than known algorithms based on conventional experiments,
this advantage was not proven against \emph{all} possible algorithms in the conventional scenario.
Here, we rigorously establish the exponential quantum advantage for performing quantum PCA.
The exponential quantum advantage also holds in the near-term proposals \cite{cotler2019quantum, huggins2020virtual}.
The proofs are in Appendix~\ref{sec:qpca}.

\begin{theorem}[Performing quantum PCA]\label{thm:PCA}
In the conventional scenario, at least order $2^{n/2}$ experiments are needed to learn a fixed property of the principal component of an unknown $n$-qubit quantum state, while a constant number of experiments suffice in the quantum-enhanced scenario.
\end{theorem}
\noindent It is worth commenting on recent results in Refs.~\cite{tang2021quantum, chia2020sampling}~showing that quantum PCA can be achieved by polynomial-time classical algorithms, which may seem to contradict Theorem \ref{thm:PCA}.
Those works assume the ability to access any entry of the exponentially large matrix $\rho$ to exponentially high precision in \emph{polynomial time}.
Achieving such a high precision requires measuring exponentially many copies of $\rho$, which takes an exponential number of experiments and \emph{exponential time}.
Hence, the assumptions of \cite{tang2021quantum, chia2020sampling} do not hold here.  See~\cite{cotler2021revisiting} which provides a detailed exposition of these matters.

\begin{figure*}[t]
\centering
\includegraphics[width=1.0\textwidth]{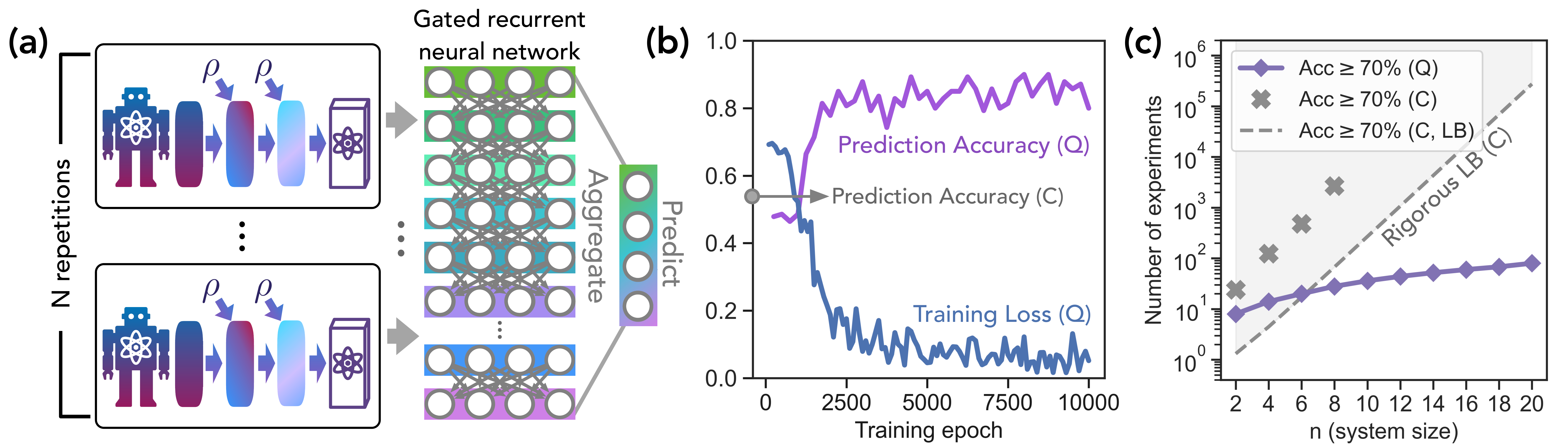}
    \caption{
    Quantum advantage in learning physical states.
    \emph{(a) Supervised machine learning (ML) model based on quantum-enhanced experiments.} $N$ repetitions of quantum-enhanced experiments are performed and the data is fed into a gated recurrent neural network (GRU) \cite{chung2014empirical, tang2015document}. The neurons in the GRU are aggregated to predict an output.
    \emph{(b) Training process of the supervised ML model.}
    We train the supervised ML model to determine which of two $n$-qubit Pauli operators has a larger magnitude for the expectation value in an unknown state $\rho$ using noiseless simulation for small system sizes $(n < 8)$.
    We consider the cross entropy \cite{murphy2012machine} as the training loss.
    Then we use the supervised ML model to make predictions using data from noisy quantum-enhanced experiments running on the Sycamore processor \cite{arute2019quantum} for larger system sizes $(8 \leq n \leq 20)$.
    We consider the probability to predict correctly as the prediction accuracy.
    Random guessing yields a prediction accuracy of $0.5$.
    \emph{(c) Quantum advantage in the number of experiments needed to achieve $\geq 70\%$ accuracy.}
    Here, (Q) corresponds to results running the supervised ML model based on quantum-enhanced experiments and (C) corresponds to results running the best known conventional strategy.
    The dotted line is a lower bound for any conventional strategy  (C, LB) as proven in Appendix~\ref{sec:compare-abs-value}.
    Even running on a noisy quantum processor, quantum-enhanced experiments are seen to vastly outperform the best theoretically achievable conventional results (C, LB).
    \label{fig:AdvantageExpt1}
    }
\end{figure*}

Another core task in quantum mechanics is understanding physical processes rather than states.
Here, each experiment implements a  physical process $\cE$, and we can interface with $\cE$ through a quantum/classical machine in the quantum-enhanced/conventional setting; see Fig.~\ref{fig:Cartoon}(c).
We show that a quantum machine can learn an approximate model of any polynomial-time quantum process $\cE$ from only a polynomial number of experiments. Given a distribution on input states, the approximate model can predict the output state from $\cE$ accurately on average.
In contrast, we would need an exponential number of experiments to achieve the same task in the conventional setting. The proof for general quantum processes is given in Appendix~\ref{sec:evolution}.

\begin{theorem}[Learning quantum processes]\label{thm:processes}
Suppose we are given a polynomial-time physical process $\cE$ acting on $n$ qubits and a probability distribution over $n$-qubit input states. In the conventional scenario, at least order $2^n$ experiments are needed to learn an approximate model of $\cE$ that predicts output states accurately on average, while a polynomial number of experiments suffice in the quantum-enhanced scenario.
\end{theorem}

\section{Demonstrations of quantum advantage}
\label{sec:experimental}
The exponential quantum advantage captured by Theorems \ref{thm:observables}, \ref{thm:PCA}, and \ref{thm:processes} applies no matter how much classical processing power is leveraged in the conventional experiments. The conventional strategy fails because there is just no way to access enough classical data to perform the specified tasks, if the number of experiments is subexponential in $n$. But these exponential separations apply in an idealized setting where quantum states are stored and processed perfectly. Will access to quantum memory unlock a substantial quantum advantage under more realistic conditions?

For two different tasks, we have investigated the robustness of the quantum advantage by conducting experiments using a superconducting quantum processor. The first task we studied pertains to Theorem~\ref{thm:observables}.
The task is to approximately estimate the magnitude for the expectation value of Pauli observables.
The unknown state is an unentangled $n$-qubit state $\rho= 2^{-n}\left(I+ \alpha P\right)$, where $\alpha = \pm 0.95$, $P$ is a Pauli operator, and both $\alpha, P$ are unknown.
After all measurements are completed and learning is terminated, two distinct Pauli operators $Q_1$ and $Q_2$ are announced, one of which is $P$ and the other of which is not equal to $P$.
We ask the machine to determine which of $|\textrm{tr}\left(Q_1\rho\right)|$ and $|\textrm{tr}\left(Q_2\rho\right)|$ is larger.

\begin{figure*}[t]
\centering
\includegraphics[width=0.96\textwidth]{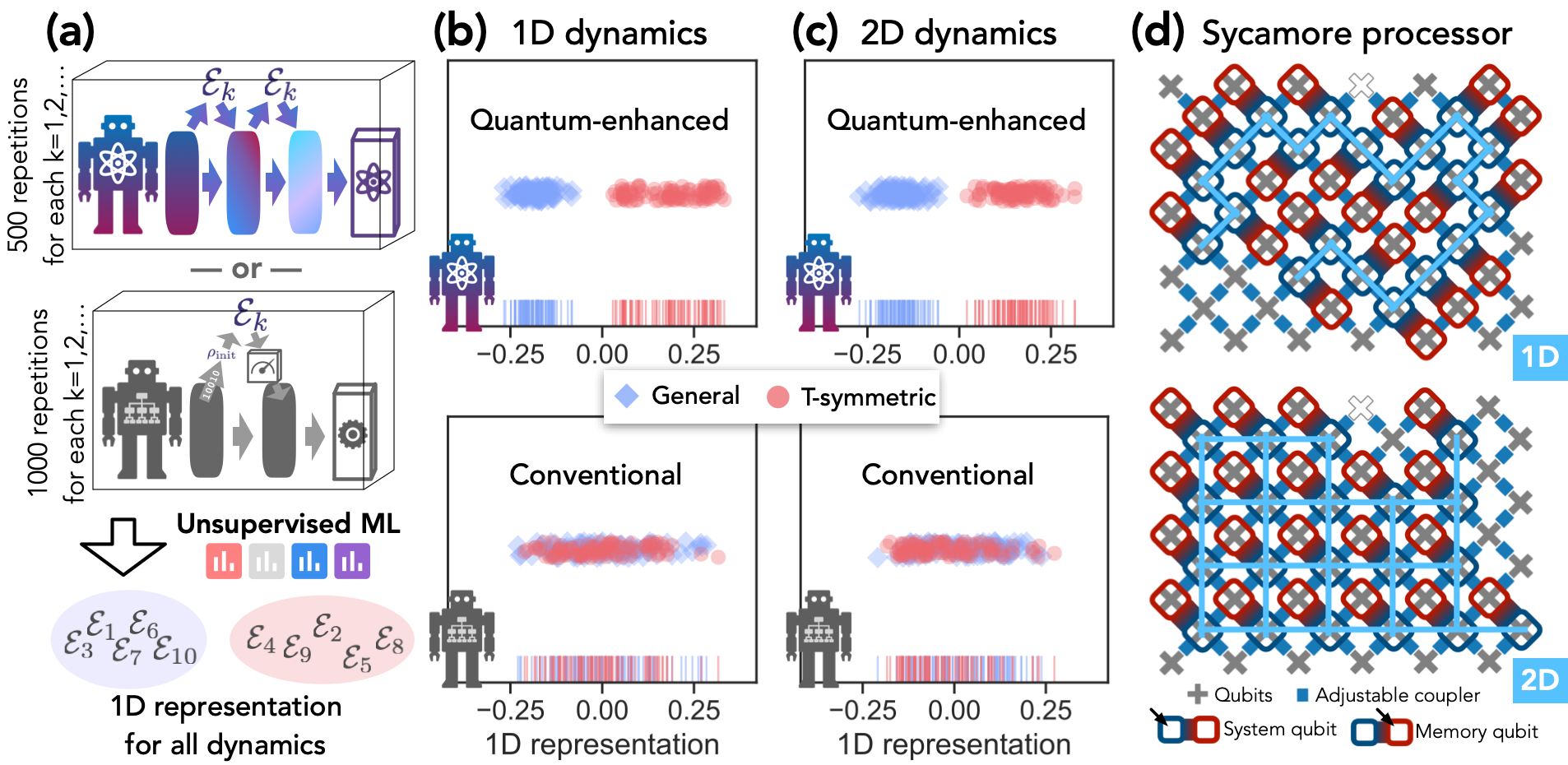}
    \caption{
    Quantum advantage in learning physical dynamics.
    \emph{(a) Unsupervised machine learning (ML) model.}
    We perform $500$ repetitions of quantum-enhanced experiments (each accessing $\cE_k$ twice) for every physical process $\cE_k$, and feed the data into an unsupervised ML model (Gaussian kernel PCA \cite{scholkopf1998nonlinear}) to learn a one-dimensional representation for describing distinct physical dynamics $\cE_{1}, \cE_{2}, \ldots$.
    Similarly, we also consider applying unsupervised ML to data obtained from $1000$ repetitions of the best known conventional experiments (each accessing $\cE_k$ once) for every physical process $\cE_k$.
    \emph{(b) Representation learned by unsupervised ML for 1D dynamics.}
    Each point corresponds to a distinct physical process~$\cE_k$.
    The vertical line at the bottom shows the exact 1D representation of each $\cE_k$.
    Half of the processes satisfy time-reversal symmetry (blue diamonds) while the other half of them do not (red circles).
    When fed with data from quantum-enhanced experiments, the ML model accurately discovers the underlying symmetry pattern.
    In contrast, the ML model fails to do so when fed with data from conventional experiments.
    \emph{(c) Representation learned by unsupervised ML for 2D dynamics.}
    \emph{(d) The geometry implemented on the Sycamore processor \cite{arute2019quantum}.}
    We consider two different classes of connectivity geometry for implementing 1D (top) and 2D (bottom) dynamics.
    \label{fig:AdvantageExpt2}
    }
\end{figure*}

In the conventional scenario, where copies of $\rho$ are measured one by one, the best known strategy is to use randomized Clifford measurements requiring an exponential number of copies to achieve the task with reasonable success probability~\cite{huang2020predicting, huang2021information}.
In the quantum-enhanced scenario, copies of $\rho$ are deposited in quantum memory two at a time, and a Bell measurement across the two copies is performed to extract a snapshot of the state.
We then feed the snapshots into a supervised machine learning (ML) model based on a gated recurrent neural network \cite{chung2014empirical, tang2015document, goodfellow2016deep} to make a prediction, as depicted in Fig.~\ref{fig:AdvantageExpt1}(a).
We train the neural network using noiseless simulation data for small system sizes $(n < 8)$.
Then we use the neural network to make predictions when we are provided with experimental data for large system sizes $(8 \leq n \leq 20)$. We report the prediction accuracy, which is equal to the probability for correctly answering whether $|\textrm{tr}\left(Q_1\rho\right)|$ or $|\textrm{tr}\left(Q_2\rho\right)|$ is larger.
Fig.~\ref{fig:AdvantageExpt1}(b) shows the performance of the ML model as we train the neural network.
Fig.~\ref{fig:AdvantageExpt1}(c) depicts, as a function of the system size $n$, the number of experiments needed in each scenario to achieve $70\%$ prediction accuracy.
We show the experimental results when using conventional and quantum-enhanced experiments, along with a theoretical lower bound on the number of experiments needed in the conventional scenario as proven in Appendix~\ref{sec:compare-abs-value}.
Despite the noisy storage and processing in the experimental device, we observe a significant quantum advantage.

The second task we studied, which pertains to Theorem~\ref{thm:processes}, is inspired by the recent observation that quantum-enhanced experiments can efficiently identify the symmetry class of a quantum evolution operator, while conventional experiments cannot \cite{aharonov2021quantum, chen2021exponential}.
An unknown $n$-qubit quantum evolution operator is presented, drawn either from the class of all unitary transformations, or from the class of time-reversal-symmetric unitary transformations (i.e., real orthogonal transformations).  
We consider whether an unsupervised ML can learn to recognize the symmetry class of the unknown evolution operator
based on data obtained from either quantum-enhanced experiments or conventional experiments.
An illustration is shown in Fig.~\ref{fig:AdvantageExpt2}(a).

In the conventional scenario, we repeatedly apply the unknown evolution operator to the initial state $|0\rangle^{\otimes n}$, and then measure each qubit of the output state in the $Y$-basis.
Under $T$-symmetric evolution the output state has purely real amplitudes; hence the 
expectation value of any purely imaginary observable, such as the Pauli-$Y$ operator, is always zero. In contrast, the expectation value of $Y$ after general unitary evolution is generically nonzero, but may be exponentially small and hence hard to distinguish from zero. 
In the quantum-enhanced scenario, we make use of $n$ additional memory qubits.
We prepare an initial state in which the $n$ system qubits are entangled with the $n$ memory qubits, evolve the system qubits under the unknown evolution operator, swap the system and memory qubits, evolve the system qubits again, and finally perform $n$ Bell measurements, each acting on one system qubit and one memory qubit. 

Each evolution operator is a one-dimensional or two-dimensional $n$-qubit quantum circuit as shown in Fig.~\ref{fig:AdvantageExpt2}(d).
After sampling many different evolution operators from both symmetry classes (and obtaining data from each sampled evolution multiple times), we use an unsupervised ML model (kernel PCA \cite{scholkopf1998nonlinear}) to find a one-dimensional representation of the evolution operators.
The representations learned by the unsupervised ML model are shown in Fig.~\ref{fig:AdvantageExpt2}(b, c).
Using the quantum-enhanced data, the ML model discovers a clean separation between the two symmetry classes, while there is no discernable separation into classes when using data from conventional experiments. The signal from the quantum-enhanced experiments is strong enough that the two classes are easily recognized without access to any labeled training data.

\section{Outlook}

We have investigated how quantum technology can enhance our ability to discover new phenomena occurring in nature. For a variety of tasks, we proved that quantum-enhanced strategies using quantum memory and quantum processing can predict properties of physical systems using exponentially fewer experiments than conventional strategies. This exponential advantage is achievable even if the amount of classical processing used in the conventional strategies is unlimited, and even when the physical system exhibits only classical correlations.
While many previous studies of quantum advantage have focused on computational tasks with known inputs, our work focuses instead on learning tasks, where the goal is to learn about an a priori unknown physical system.
This work provides a new approach to understanding and achieving quantum advantage in quantum machine learning \cite{biamonte2017quantum,broughton2021tensorflow} and quantum sensing \cite{degen2017quantum}.

Our experiments with up to 40 qubits in a superconducting quantum processor show that a substantial quantum advantage is already evident when using today’s noisy intermediate-scale quantum platforms \cite{preskill2018quantum}. These experiments demonstrate that supervised and unsupervised machine learning models \cite{goodfellow2016deep, mohri2018foundations} employing data obtained from quantum-enhanced experiments can predict properties and discover underlying structure in physical systems that are beyond the scope of conventional experiments. 

We envision that future quantum sensing systems will be able to transduce detected quantum data to a quantum memory and then process the stored data using a quantum computer. Lacking for now suitably advanced sensors and transducers, we have conducted proof-of-concept experiments in which quantum data is directly planted in our quantum processor. Nevertheless, the robust quantum advantage we have validated highlights the potential for advancing quantum platforms to unlock facets of nature that would otherwise remain concealed. 

\vspace{-1em}
\subsection*{Acknowledgments:}
\vspace{-1em}

The quantum hardware used for this experiment was developed by the Google Quantum AI hardware team, under the direction of Anthony Megrant, Julian Kelly and Yu Chen. Methods for device calibrations were developed by the physics team led by Vadim Smelyanskiy. Data was collected via cloud access through Google's Quantum Computing Service.
HH is supported by a Google PhD Fellowship.  JC is supported by a Junior Fellowship from the Harvard Society of Fellows, the Black Hole Initiative, as well as in part by the Department of Energy under grant {DE}-{SC0007870}. SC is supported by the National Science Foundation under Award 2103300 and was visiting the Simons Institute for the Theory of Computing while part of this work was completed. JP acknowledges funding from  the U.S. Department of Energy Office of Science, Office of Advanced Scientific Computing Research, (DE-NA0003525, DE-SC0020290), and the National Science Foundation (PHY-1733907). The Institute for Quantum Information and Matter is an NSF Physics Frontiers Center.

\bibliographystyle{apsrev4-1_with_title}
\bibliography{references}

\appendix
\onecolumngrid


\vspace{2.0em}
\begin{center}
\textbf{\Large Supplementary information}
\end{center}

\renewcommand{\appendixname}{APPENDIX}
\renewcommand{\thesubsection}{\MakeUppercase{\alph{section}}.\arabic{subsection}}
\renewcommand{\thesubsubsection}{\MakeUppercase{\alph{section}}.\arabic{subsection}.\alph{subsubsection}}
\makeatletter
\renewcommand{\p@subsection}{}
\renewcommand{\p@subsubsection}{}
\makeatother

\renewcommand{\figurename}{Supplementary Figure}
\setcounter{figure}{0}
\setcounter{secnumdepth}{3}

\bigskip

\noindent \textbf{\ref{sec:exper-detail}.~~\hyperref[sec:exper-detail]{Experimental details}} \dotfill\textbf{\pageref{sec:exper-detail}}
\medskip

\noindent \qquad \begin{minipage}{\dimexpr\textwidth-1.3cm}
 \hyperref[sec:gen-descr]{General description} $\bullet$ 
 \hyperref[sec:learn-state-detail]{Experiments on learning physical states} $\bullet$ 
 \hyperref[sec:learn-dyn]{Experiments on learning physical dynamics} $\bullet$
 \hyperref[sec:additional-exp]{Additional experimental results} $\bullet$
 \hyperref[sec:PerformanceData]{Performance and characterization data}
\end{minipage}
\medskip

\noindent \textbf{\ref{sec:review}.~~\hyperref[sec:review]{A brief review on quantum information theory}} \dotfill\textbf{\pageref{sec:review}}
\medskip

\noindent \qquad \begin{minipage}{\dimexpr\textwidth-1.3cm}
 \hyperref[sec:Qprocess]{Definition and properties of quantum processes} $\bullet$ 
 \hyperref[sec:POVMdef]{Definition and properties of POVMs}
\end{minipage}
\medskip

\noindent \textbf{\ref{sec:separation}.~~\hyperref[sec:separation]{Mathematical framework for proving exponential advantage}} \dotfill\textbf{\pageref{sec:separation}}
\medskip

\noindent \qquad \begin{minipage}{\dimexpr\textwidth-1.3cm}
 \hyperref[sec:tree-def]{Tree representation} $\bullet$ 
 \hyperref[sec:many-v-one]{Many-versus-one distinguishing tasks} $\bullet$ 
 \hyperref[sec:many-vs-many]{Many-versus-many distinguishing task} $\bullet$ 
 \hyperref[sec:part-many-v-one]{Partially-revealed many-versus-one distinguishing task} $\bullet$ 
 \hyperref[sec:noisedegradation]{Presence of noise}  $\bullet$ 
 \hyperref[sec:related-work]{Related works}
\end{minipage}
\medskip

\noindent \textbf{\ref{sec:shadow}.~~\hyperref[sec:shadow]{Predicting highly-incompatible observables}} \dotfill\textbf{\pageref{sec:shadow}}
\smallskip

\noindent \qquad \begin{minipage}{\dimexpr\textwidth-1.3cm}
 \hyperref[sec:expo-adv-abs-one]{Exponential advantage in predicting absolute value of a single observable} $\bullet$ 
 \hyperref[sec:uppbd-obs]{A constant upper bound for quantum-enhanced experiments} $\bullet$  \hyperref[sec:lowbd-obs]{An exponential lower bound for conventional experiments}
 \hyperref[sec:compare-abs-value]{An exponential lower bound for comparing absolute values} $\bullet$ 
\end{minipage}
\medskip

\noindent \textbf{\ref{sec:qpca}.~~\hyperref[sec:qpca]{Performing quantum principal component analysis}} \dotfill\textbf{\pageref{sec:qpca}}
\medskip

\noindent \qquad \begin{minipage}{\dimexpr\textwidth-1.3cm}
 \hyperref[sec:conv-exp-LB-pca]{An exponential lower bound for conventional experiments} $\bullet$ 
 \hyperref[sec:pseudo]{An exponential lower bound using pseudorandomness}
\end{minipage}
\medskip

\noindent \textbf{\ref{sec:evolution}.~~\hyperref[sec:evolution]{Learning a polynomial-time quantum process}} \dotfill\textbf{\pageref{sec:evolution}}
\medskip

\noindent \qquad \begin{minipage}{\dimexpr\textwidth-1.3cm}
 \hyperref[sec:learn-process-setting]{Problem setting} $\bullet$ 
 \hyperref[sec:rigor-learn-process]{Rigorous statements} $\bullet$
 \hyperref[sec:poly-proc]{Proof of polynomial upper bound in Theorem~\ref{thm:poly-upp-learn-proc}}
\end{minipage}
\medskip

\noindent \textbf{\ref{sec:hard-bdd-mem}.~~\hyperref[sec:hard-bdd-mem]{Predicting highly-incompatible observables with bounded quantum memory}} \dotfill\textbf{\pageref{sec:hard-bdd-mem}}
\medskip

\noindent \qquad \begin{minipage}{\dimexpr\textwidth-1.3cm}
 \hyperref[sec:background1]{Background and statement of results} $\bullet$ 
 \hyperref[sec:revew-bdd-qmem]{Review of learning tree framework for bounded quantum memories} $\bullet$
 \hyperref[sec:hardness-smallqmem]{Hardness result for small quantum memories}
\end{minipage}
\medskip

\section{Experimental details}
\label{sec:exper-detail}

In this section, we present the details for the physical experiments ran on the superconducting processor as well as the supervised/unsupervised machine learning models used to analyze the data.

\subsection{General description}
\label{sec:gen-descr}

The experiments were performed on a Google Sycamore processor containing up to 53 superconducting transmon qubits.
The largest error source in the Sycamore processor \cite{arute2019quantum} is qubit readout error, which ranges from $3\%$ to $7\%$.
The second largest error source is the two-qubit gate with an error around $0.5\%$ to $1.5\%$.
Single-qubit gates have the smallest error around $0.05\%$ to $0.5\%$.
The Sycamore chip was introduced in Ref.~\cite{arute2019quantum}, where additional details concerning the hardware implementation and performance can be found.
The Sycamore chip was controlled remotely using an internal cloud interface programmed using Cirq \cite{cirq_developers_2021_5182845} and TensorFlow Quantum \cite{broughton2021tensorflow}.
The layout of the chip including connectivity is depicted in Supp.~Fig.~\ref{fig:Geometry}(a).

For all of our experiments on learning about states and dynamics, the total number of qubits was varied from 4 to 40 qubits (where in each case half of the qubits are used to simulate a physical system).
As the system size was varied, the subset of qubits used for the experiments was varied in order to maximize experimental performance and minimize any overhead related to the 2D connectivity.
That is, the largest contiguous patches with low gate and measurement error rates were selected, and swap operations were used to meet connectivity requirements when necessary.
The experimental requirements for learning states, 1D dynamics, and 2D dynamics differ considerably in their experimental complexity.

In all experiments, the implementations of the unknown states or dynamics are performed in the quantum processor, where the learning algorithms do not know about them.
While this is only an emulation of the process of data collection from a physical system in an actual sensing experiment, it allows us to examine the proposed pipeline for quantum data processing in a situation where data collection is imperfect.

\subsection{Experiments on learning physical states}
\label{sec:learn-state-detail}

We separate this subsection into the concrete procedure for generating the unknown states, the conventional experiments we run, the quantum-enhanced experiments we run, and the supervised neural network model for making prediction based on data from quantum-enhanced experiments.

\subsubsection{Procedure for generating the class of unknown states}
\label{sec:state-I+P}

The state preparation we consider is relatively simple in that the unknown state $\rho$ is unentangled, but has strong non-local classical correlations.
In the experimental demonstration we consider
states of the form $\rho= 2^{-n}\left(I+ \alpha P\right)$, where $\alpha \in \{-0.95, 0.95\}$, $P=\bigotimes_{i=1}^n P_i$ is an $n$-qubit Pauli operator, and both $\alpha$ and  $P$ are unknown.
The state $\rho = 2^{-n}\left(I+ \alpha P\right)$ is prepared by a randomized constant-depth circuit described in the following.
To generate one copy of $\rho$, we introduce a parameter $\eta = \text{sign}(\alpha)$. Then for each qubit $i = 1, \ldots, n$, we do the following.
\begin{enumerate}
    \item If $P_i = I$, then we set qubit $i$ to be $\ketbra{0}{0}$ with probability $1/2$, and be $\ketbra{1}{1}$ with probability $1/2$.
    \item If $P_i \neq I$ and there exists $j > i$ such that $P_j \neq I$, then we set qubit $i$ to be one of the two eigenstates of $P_i$ with equal probability.
    We multiply $\eta$ by the eigenvalue ($+1$ or $-1$) of the selected eigenstate of $P_i$.
    \item If $P_i \neq I$ and there does not exist $j > i$ such that $P_j \neq I$, then we use the following procedure.
    \begin{enumerate}
        \item With probability $0.05$, we set qubit $i$ to be either $\ketbra{0}{0}$ or $\ketbra{1}{1}$ with equal probability.
        \item With probability $0.95$, we set qubit $i$ to be the positive eigenstate of $P_i$ if $\eta = +1$, and set qubit $i$ to be the negative eigenstate of $P_i$ if $\eta = -1$.
    \end{enumerate}
\end{enumerate}
By construction, the density operator prepared by this procedure is realized as an ensemble of pure states, where each pure state is a tensor product of Pauli operator eigenstates. Therefore, there is no quantum entanglement across different qubits. Furthermore, step 3 is designed to assure that $|\Tr\left(P\rho\right)| = 0.95$.

\begin{figure*}[t]
\centering
\includegraphics[width=0.98\textwidth]{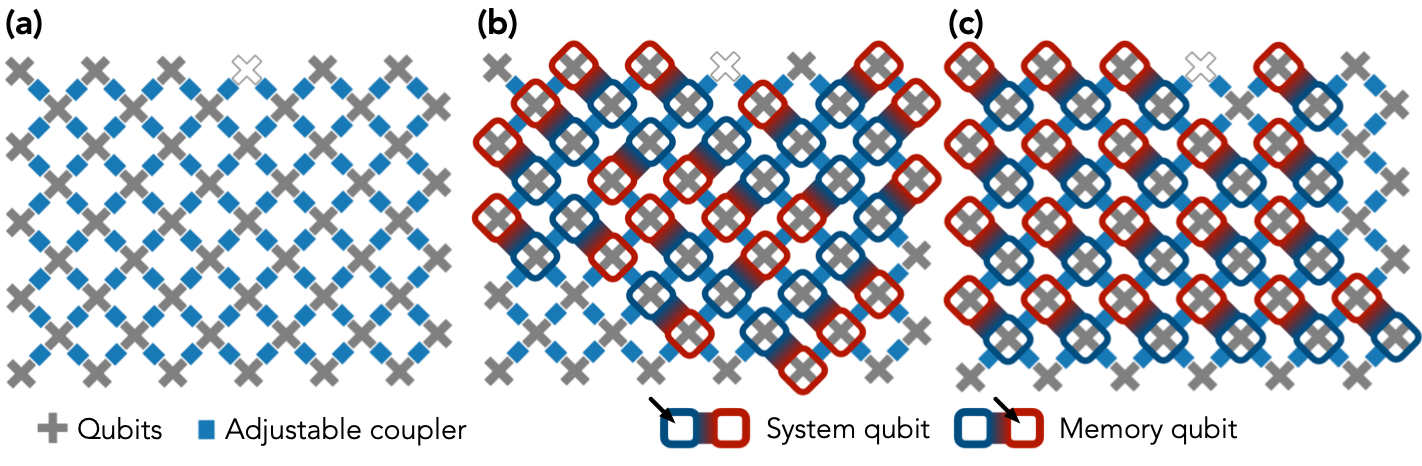}
    \caption{
    \emph{(a) Layout of a Google Sycamore processor.} There is a total of $53$ superconducting transmon qubits (the qubit corresponding to an empty cross is out of order). The blue rectangles show the adjustable couplers that can apply the entangling two-qubit gate $\mathrm{SYC}$ to neighboring qubits.
    \emph{(b) Layout used for learning states and for learning 1D dynamics.} We partition the $40$ qubits into $20$ system qubits and $20$ memory qubits.  Either an unknown state of the system qubits is prepared, or an unknown process is applied to the system qubits. 
    \emph{(c) Layout used for learning 2D dynamics.} 
    \label{fig:Geometry}
    }
\end{figure*}

\subsubsection{Conventional experiments}

In the conventional setting, the optimal strategy (up to logarithmic factors) for estimating expectation values of  high-weight $n$-qubit Pauli observables uses
classical shadow tomography based on randomized Clifford measurements \cite{huang2020predicting}.
Using this strategy, in each experiment we randomly sample a unitary transformation from the Clifford group, apply the sampled transformation to the unknown state $\rho$, and then measure in the computational basis. Although such randomized Clifford measurements can be executed using quantum circuits of polynomial size, the required circuits are too large to be performed accurately with today's noisy quantum devices except for quite modest values of $n$. Furthermore, the classical post-processing of the measurement results has complexity exponential in $n$.

In our conventional experiments, because randomized Clifford measurements are infeasible we instead use classical shadow tomography based on randomized Pauli measurements \cite{huang2020predicting}.
Using this strategy, in each experiment, we randomly sample from depth-$1$ Clifford circuits, apply the sampled circuit to $\rho$, and then measure in the computational basis. That is, for each of the $n$ qubits, we decide uniformly at random to measure one of the three Pauli observables $X$, $Y$, or $Z$. Many such measurements are performed, each time on a new copy of $\rho$, and the classical data collected is post-processed to predict expectation values of observables in the state $\rho$.

Classical shadow tomography based on randomized Pauli measurements is a powerful technique that enables classical ML models to predict quantum many-body ground states and quantum phases of matter with rigorous guarantees \cite{huang2021provably}.
However,
for the task of estimating the expectation value of an $n$-qubit Pauli observable that is announced after all measurements are completed, both randomized Clifford and randomized Pauli measurements require a number of experiments that scale exponentially in $n$, as we have proven in Appendix~\ref{sec:compare-abs-value}. Likewise, exponentially many experiments are needed to perform Task \ref{task:comp-abs} defined in Appendix~\ref{sec:compare-abs-value} with high success probability. 
In Fig.~\ref{fig:AdvantageExpt1}(b), we report the average prediction accuracy (the probability of performing Task \ref{task:comp-abs} successfully) over system sizes $n = 8, 10, 12, 14, 16, 18, 20$.
For each $n$-qubit state $\rho$, we conduct $1000$ experiments to obtain the measurement data.
The prediction accuracy is indicated by the gray point shown on the  vertical axis, which is only slightly better than random guessing ($0.5$).
For system sizes $n \geq 10$, the prediction accuracy is very close to $0.5$.
Classical shadow tomography based on randomized Pauli measurements is a statistical estimation procedure that has no training phase.
Therefore, the prediction accuracy for conventional experiments is a single point Fig.~\ref{fig:AdvantageExpt1}(b).
The training epoch on the horizontal axis in Fig.~\ref{fig:AdvantageExpt1}(b) is only for quantum-enhanced strategy.
In Fig.~\ref{fig:AdvantageExpt1}(c), we report the number of experiments required to achieve a prediction accuracy of at least $70\%$ for different system sizes.
We consider a maximum of $5000$ experiments.
For system size $n \geq 10$, we are unable to achieve at least $70\%$ prediction accuracy with $5000$ experiments. Hence we only show system size $n = 2, 4, 6, 8$ for conventional experiments in Fig.~\ref{fig:AdvantageExpt1}(c).

\subsubsection{Quantum-enhanced experiments}

Quantum-enhanced experiments are executed by performing an entangling Bell measurement across two copies of $\rho$. 
We prepare the state $\rho$ on the system qubits (marked blue in Supp.~Fig.~\ref{fig:Geometry}), swap the state to the memory qubits (marked red), prepare another state $\rho$ on the system qubits, then perform an entangling Bell measurement across the two copies of $\rho$.
Note that every preparation of $\rho$ generates a random product state according to a classical probability distribution described in Appendix~\ref{sec:state-I+P}.

For each system size $n = 2, \ldots, 20$, we choose $n$ qubits from among the $20$ pairs of qubits shown in Supp.~Fig.~\ref{fig:Geometry}(b); these pairs are selected to minimize errors in the state preparations, gates, and measurements. Because the state $\rho$ is not entangled, no entangling gates are used during the state preparation; therefore there is no advantage in choosing the pairs of qubits to be in proximity to one another.  
For each system size $n$, we use the same qubits for the conventional experiments as for the quantum-enhanced experiments, except that in conventional experiments we prepare the unknown state $\rho$ only on the system qubits, and then perform a randomized Pauli measurement of the system immediatly after the state preparation.

While a Bell measurement on a pair of qubits can be performed via a simple circuit containing one Hadamard gate, one CNOT gate, and two $Z$-basis measurements, we instead compile these operations into operations better suited for the Sycamore processor.
In particular, the native two-qubit entangling gate is the Sycamore Gate, which has a unitary matrix representation given by
\begin{align}
    \text{SYC} = \text{ iSWAP}^\dagger \text{ CPHASE}(-\pi/6) = \left( 
    \begin{array}{cccc}
    1 & 0 & 0 & 0 \\
    0 & 0 & -i & 0 \\
    0 & -i & 0 & 0 \\
    0 & 0 & 0 & e^{-i \pi/6}
    \end{array} \right).
\end{align}
Using this gate, the Bell state measurements may be performed by the gate sequence needed to unprepare a Bell state, which when compiled to Sycamore's native gates is expressed as a product of SYC gates and phased XZ gates ($\text{PhXZ}_i$).  The phased XZ gate $\text{PhXZ}_i$ on the $i$-th qubit is a native gate on the Sycamore device that can be expressed as
\begin{align}
    \text{PhXZ}(a, x, z)_i = Z^z_i Z^a_i X^x_i Z^{-a}_i.
\end{align}
where $X_i$ and $Z_i$ are the standard Pauli operators acting on qubit $i$, and the exponents $a$, $x$, $z$ are real numbers.
The particular angles ($a, x, z$) for each of the gates used in our experiments were compiled numerically via a variational optimization.
As the inverse Sycamore is not a native gate of the architecture, the compilations to hardware for inverse gates have to compensate for this difference, which we do numerically.
The full decomposition of all the gates and circuits we reference are provided as Cirq circuits in additional supplemental material.

Each quantum-enhanced experiment generates a classical bitstring of size $2n$.
We collect the bitstrings from all experiments and feed them into a neural network model.
In addition to the experimental data, the neural network model also takes in two Pauli strings $Q_1, Q_2$ and predicts which one of $|\Tr(Q_1 \rho)|$ and $|\Tr(Q_2 \rho)|$ is larger.
In Fig.~\ref{fig:AdvantageExpt1}(b), we repeat each quantum-enhanced experiments $500$ times.
This provides a fair comparison with conventional experiments because two copies of the unknown state $\rho$ are used in each quantum-enhanced experiment; therefore a total of $1000$ copies are consumed in both our conventional and quantum-enhanced experiments.
In Fig.~\ref{fig:AdvantageExpt1}(c), we repeat quantum-enhanced experiments for a maximum of $500$ times over different system sizes from $n = 2, 4, 6, 8, 10, 12, 14, 16, 18, 20$.

\subsubsection{A supervised neural network model using data from quantum-enhanced experiments}
\label{sec:NN-data-Q}

We train a supervised neural network model using noiseless simulation data from small system sizes. Then we use the trained neural network model on the noisy experimental data obtained from performing quantum-enhanced experiments.
The neural network model has three layers. Each of the outer layers runs the preceding inner layer multiple times.
In the following, we describe each layer of the neural network model.
\begin{enumerate}
    \item The \emph{inner layer} is a recurrent neural network based on gated recurrent unit (GRU) \cite{chung2014empirical, tang2015document, goodfellow2016deep}.
    The recurrent neural network takes in a size-$2n$ bitstring, corresponding to the measurement outcome from a single quantum-enhanced experiment, and an $n$-qubit Pauli operator $Q$, which can be represented as a size-$2n$ bitstring.
    The recurrent neural network outputs a two-dimensional real vector.
    Other popular choices of recurrent units, such as long short-term memory (LSTM) \cite{sepp1997lstm} or transformer \cite{vaswani2017attention}, could be used instead of GRU.
    \item The \emph{intermediate layer} is an aggregation layer. This layer runs the inner layer for all the bitstrings obtained from each of the quantum-enhanced experiments.
    For example, if we run the quantum-enhanced experiments for $100$ times, we would obtain $100$ size-$2n$ bitstring and we would run the inner layer for $100$ times over each bitstring.
    The intermediate layer outputs the average of the two-dimensional output vectors from the multiple runs of the inner layers.
    \item The \emph{outer layer} is inspired by a Siamese neural network (twin neural network) \cite{koch2015siamese}. This layer runs the intermediate layer twice, one for each of the two Pauli operators $Q_1$ and $Q_2$.
    Each intermediate layer generates a real-valued vector $x$ of dimension two, which we map to a single real value by considering $x_1 - x_2$.
    The outer layer compares the real value from the two intermediate layers and outputs $Q_1$ or $Q_2$ based on which one of them has a higher real value.
\end{enumerate}
The specific details of the above neural network structure is given in the accompanying code repository\footnote{\url{https://github.com/quantumlib/ReCirq/tree/master/recirq/qml_lfe}}.

In the inner layer, we create a recurrent neural network with an encoding layer that maps an integer between $0$ and $15$ to a vector of dimension $30$, a GRU with $30$ neurons, and a decoding layer that maps $30$ neurons to $2$ neurons.
Only the inner layer contains trainable parameters. The intermediate layer and the outer layer are both fixed operations based on outputs from the inner layer, which will not be updated.

Next, we discuss the process for training the neural network model.
We use noiseless simulation data (for small system sizes $n < 8$) to train the recurrent neural network.
During training, we pick a state $\rho = 2^{-n}\left(I+ \alpha P\right)$ where $\alpha \in \{-0.95, 0.95\}$ and $P$ is an $n$-qubit Pauli operator, and pick an $n$-qubit Pauli operator $Q$ that is equal to $P$ with probability $1/2$ and is not equal to $P$ with probability $1/2$.
We encode the training data into two tensors, \pythoninline{inp} and \pythoninline{target}.
The encoding is defined by the following.
\begin{itemize}
    \item The tensor \pythoninline{inp} is of size $b \times n$, where $b$ is the number of quantum-enhanced experiments we performed, $n$ is the number of qubits, and each entry of \pythoninline{inp} is an integer from $0$ to $15$.
    The $(t, i)$-th entry of \pythoninline{inp} encodes the component of $Q$ on qubit $i$ (a choice of $4$ for $I, X, Y, Z$) and the Bell measurement outcome on qubit $i$ from the $t$-th quantum-enhanced experiment (also a choice of $4$).
    Each entry takes a total of $16$ possible values.
    \item The tensor \pythoninline{target} is of size $1$. The entry in \pythoninline{target} is equal to $1$ if $P = Q$, and is equal to $0$ if $P \neq Q$.
\end{itemize}
We update the neural network model once using \pythoninline{inp} of size $b \times n$ and \pythoninline{target} of size $1$.
We are using the cross entropy loss and employ the Adam optimizer \cite{kingma2014adam}, which is a gradient-based optimization algorithm that adaptively estimates lower-order moments.
We generate multiple different states $\rho$ and $Q$ corresponding to different \pythoninline{inp} and \pythoninline{target} to train the neural network model.

During the training process, we are not using the \emph{outer layer}.
Also, we simultaneously run the $b$ repetitions of the inner layer for each outcome from a single quantum-enhanced experiment by leveraging parallel computing.
Then, we average over the $b$ repetitions of the inner layer.
Also, the output of the neural network model is a two-dimensional real vector, denoted as $v = (v_0, v_1)$.
When \pythoninline{target} is $a \in \{0, 1\}$, the loss function is given by
\begin{equation}
    -\log\left( \frac{\euler^{v_a}}{\euler^{v_0} + \euler^{v_1}} \right).
\end{equation}
The two real values $v_0$, $v_1$ are combined to produce a probability distribution
\begin{equation} \label{eq:prob-NN}
    \frac{\euler^{v_0}}{\euler^{v_0} + \euler^{v_1}} = 1 - \frac{1}{\euler^{v_0 - v_1} + 1}, \quad \frac{\euler^{v_1}}{\euler^{v_0} + \euler^{v_1}} = \frac{1}{\euler^{v_0 - v_1} + 1},
\end{equation}
indicating which of $a=0$ and $a=1$ is more likely.
If $v_0 - v_1$ is large, then $a=0$ corresponding to $P \neq Q$ is more likely. On the other hand, if $v_0 - v_1$ is small, then $a=1$ corresponding to $P = Q$ is more likely.
We compute the gradient through back-propagation and update the model using the Adam optimizer \cite{kingma2014adam}.

Finally, we discuss the prediction process in the neural network model.
Due to the significant amount of measurement errors, we employ a form of measurement error mitigation.
We first characterize the measurement errors for every qubit assuming the zero state preparations and $X$-gates are perfect.
For each qubit $i$, we obtain a $2\times 2$ matrix specifying the probability to measure $0$ or $1$ if the qubit is in $\ketbra{0}{0}$ or $\ketbra{1}{1}$.
We store that as a list of $2 \times 2$ matrices called \pythoninline{calib_2x2}.
We then expand the data, referred to as \pythoninline{data} in the pseudo-code, obtained from the quantum-enhanced experiments, which is a two-dimensional array of size $b \times (2n)$.
Basically, we expand each measurement to $20$ measurements with a real-valued coefficient associated to each of the expanded measurements.
Therefore, \pythoninline{data_expanded} is a two-dimensional array of size $(20 b) \times (2n)$ and \pythoninline{coefficients} is a one-dimensional array of size $20 b$.
\begin{python}
def noise_inversion(data, calib_2x2, inverse_cnt=20):
    Set data_expanded as an empty array
    Set list_of_coefficients as an empty array
    
    for t from 0 to b-1:
        for r from 0 to inverse_cnt-1:
            Set single_data as an empty array
            Set coefficient as 1.0
            
            for i from 0 to 2n-1:
                Set p as calib_2x2[i][1, 1] if data[t][i] = 0
                Set p as calib_2x2[i][0, 0] if data[t][i] = 1
                
                With probability 1-p do:
                    single_data.append(1-data[t][i])
                    coefficient *= -1
                Else do:
                    single_data.append(data[t][i])

            Append single_data to data_expanded
            Append coefficient to list_of_coefficients
            
    return data_expanded, list_of_coefficients
\end{python}

After obtaining \pythoninline{data_expanded}, we construct two tensors \pythoninline{inp1} and \pythoninline{inp2} corresponding to the same experimental data \pythoninline{data_expanded}, but different Pauli operators $Q_1, Q_2$.
Both \pythoninline{inp1} and \pythoninline{inp2} are tensors of size $(20 b) \times n$, where $b$ is the number of quantum-enhanced experiments we performed, $n$ is the number of qubits, and each entry of \pythoninline{inp1} and \pythoninline{inp2} is an integer from $0$ to $15$ similar to the training process.
Then the neural network make a prediction using the two input tensors \pythoninline{inp1} and \pythoninline{inp2}.
In the \emph{outer layer}, the neural network model runs the \emph{intermediate layer} (as well as the multiple repetitions of \emph{inner layer}) for each of the two input tensors to obtain two 2D vectors denoted as $u = (u_0, u_1), v = (v_0, v_1)$.
From the discussion given around Eq.~\eqref{eq:prob-NN}.
If $u_0 - u_1$ is small, then it is more likely that $P = Q_1$.
If $v_0 - v_1$ is small, then it is more likely that $P = Q_2$.
The neural network hence compare $u_0 - u_1$ and $v_0 - v_1$ to predict whether $P = Q_1$ or $P = Q_2$.

\subsection{Experiments on learning physical dynamics}
\label{sec:learn-dyn}

For the task of learning about physical dynamics in 1D and 2D, we considered unitary transformations implemented by 1D and 2D random quantum circuits. We generated many random circuits, half of which are time-reversal symmetric (i.e., real orthogonal), and half of which are general unitary circuits without any symmetry. 
For each of these circuits, we performed both conventional and quantum-enhanced experiments to generate classical measurement data.
This data was fed to an unsupervised machine learning model to learn a low-dimensional classical representation of the physical dynamics. We wished to see whether the unsupervised ML model could recognize the difference between time-reversal symmetric dynamics and general dynamics.
The results summarized in Fig.~\ref{fig:AdvantageExpt2} were obtained in experiments analyzing $180$ different circuits in each of the the two classes, using methods described below.
The largest quantum circuits we ran on the Sycamore processor are presented in Table~\ref{tab:scrambling circuits}.

\begin{table}[h]
    \centering
    \begin{tabular}{c|c|c|c}
         & \,\,\textbf{Number of qubits}\,\, & \, \textbf{Number of gates} \, & \, \textbf{Circuit depth} \, \\
        \hline\hline
        \textbf{1D dynamics}\,\, & \textbf{40} & \textbf{842} & \textbf{40} \\
        \textbf{2D dynamics}\,\, & \textbf{40} & \textbf{1388} & \textbf{54}
    \end{tabular}
    \caption{Circuit information for the experiments on learning physical dynamics. \label{tab:scrambling circuits}}
\end{table}

In \cite{aharonov2021quantum}, a restricted subclass of conventional strategies was shown to require an exponential number of experiments to distinguish between general unitary dynamics and time-reversal-symmetric dynamics.
In \cite{chen2021exponential}, some of the authors of the present work have shown that an exponential number of experiments are required for this task even when arbitrary conventional strategies are allowed. 
Furthermore, it is plausible that under appropriate cryptographic assumptions, the superpolynomial difficulty of characterizing quantum dynamics in conventional experiments would persist even for psuedo-random dynamical processes that can be efficiently generated on a quantum computer.  However, at present, explicit constructions of cryptographically-secure pseudo-random unitaries are not known.  In light of this, in our experiments we resort to studying random quantum circuits similar to those used for demonstrating quantum computational supremacy \cite{arute2019quantum}.

In this subsection, we provide further details regarding how our samples of 1D dynamics and 2D dynamics are generated, how our experiments are conducted, and how our unsupervised machine learning model works. As we will see, the unsupervised ML successfully learns to classify quantum circuits into symmetry classes when provided with data from quantum-enhanced experiments, but not when provided with data from conventional experiments. 

\subsubsection{1D dynamics}

We use the layout provided in Supp.~Fig.~\ref{fig:Geometry}(b).
A 1D circuit is implemented along the 1D line connecting all the qubits circled in blue.
For system size $n < 20$, we consider a contiguous region on the 1D line with the smallest gate/measurement error.
For general unitary dynamics, we use a 1D version of the quantum supremacy circuit \cite{arute2019quantum}.
The quantum supremacy circuit in 1D interleaves between a layer of random single-qubit gates and a layer of two-qubit entangling gates, namely $\mathrm{SYC}$ gates applied to neighboring qubits.
We alternate the partitioning of the two-qubit entangling gates, e.g., $(1, 2), (3, 4), (5, 6) \longleftrightarrow (2, 3), (4, 5)$ for $n = 6$.

For $T$-symmetric (time-reversal-symmetric) dynamics, the single-qubit gates are real orthogonal $2\times 2$ matrices of the form $\euler^{-it Y}$, where $Y$ is the Pauli-$Y$ matrix and $t$ is a randomly chosen real number.
In addition, we replace the two-qubit entangling gate $\mathrm{SYC}$ by a $T$-symmetric two-qubit entangling gate
\begin{equation}
    V = (U_3 \otimes U_4) \mathrm{SYC} (U_1 \otimes U_2),
\end{equation}
where $U_1, U_2, U_3, U_4$ are appropriately chosen single-qubit gates.
In order to find a suitable choice of $U_1, \ldots, U_4$ such that the two-qubit gate $V$ is time-reversal symmetric, we employ a numerical optimization.
We parameterize each of $U_i$ as $\exp(\rmi (a_i X + b_i Y + c_i Z))$, where $a_i, b_i, c_i \in \mathbb{R}$ are initialized randomly.
There are a total of $12$ variables.
Next, we define the loss function to be equal to the Frobenius norm of the imaginary part of the matrix $V$ (after fixing the top left entry of $V$ to be real).
We then perform gradient descent to minimize the loss function and terminate once we have found that the loss function is below $10^{-9}$.
We then replace $\mathrm{SYC}$ by $V$.

In the conventional experiment, we begin with $\ketbra{0^n}{0^n}$ on the system qubits, evolve under the 1D dynamics, and measure in the $Y$-basis.
We also considered using randomized Pauli measurement at the end, but the performance for measuring in the $Y$-basis is slightly better.
The rationale is that the output state under $T$-symmetric evolution has purely real amplitudes; hence the expectation value of any purely imaginary observable, such as the Pauli-$Y$ operator, is always zero.
In contrast, the expectation value of $Y$ after a general unitary evolution is non-zero.
$T$-symmetric unitaries are nevertheless hard to distinguish from general unitaries in the conventional setting because the expectation value of $Y$ is exponentially small for general unitaries; theefore an exponentially large number of experiments are needed to discern its nonzero value. 
Indeed, the result in \cite{aharonov2021quantum, chen2021exponential} shows that conventional strategies require an exponential number of experiments for distinguishing $T$-symmetric unitaries from general unitaries. 

In the quantum-enhanced experiment, we prepare a Bell state $\frac{1}{\sqrt{2}}(\ket{00} + \ket{11})$ for every pair of system and memory qubits.
Then we evolve the system qubits under the unknown dynamics.
After the evolution, we swap the system and the memory qubits.
Then we evolve the system qubits under the unknown dynamics again.
Finally, we measure every pair of system and memory qubits in the Bell basis.
Each quantum-enhanced experiment generates a $2n$-bit string.
We perform gradient descent to find our implementation of the Bell state preparation, swap operation, and Bell measurement using the native gates in the Sycamore processor.

\begin{figure*}[t]
\centering
\includegraphics[width=0.98\textwidth]{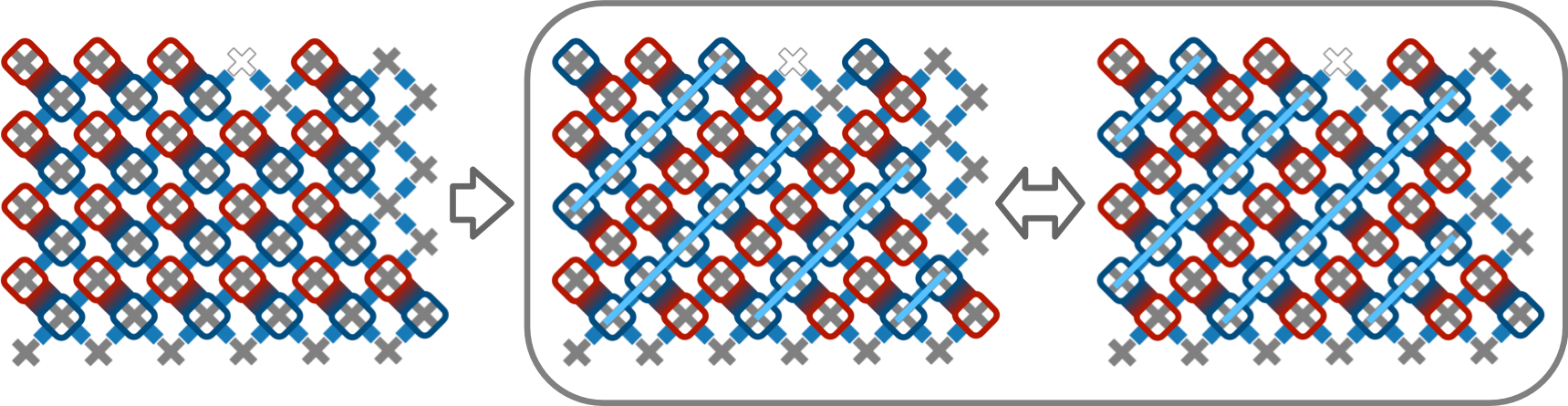}
    \caption{
    To implement the 2D dynamics, we first move from the layout on the left to the layout in the middle (by swapping some pairs of system and memory qubits).
    Then, we iterate between the layout in the middle and the layout in the right (by swapping all pairs of system and memory qubits).
    \label{fig:2Dscrambling}
    }
\end{figure*}

\subsubsection{2D dynamics}

For our 2D circuits we use the layout provided in Supp.~Fig.~\ref{fig:Geometry}(c) which is also shown as the leftmost layout in Supp.~Fig.~\ref{fig:2Dscrambling}.
In the leftmost layout, none of the system qubits (circled blue) are connected to one another.
In order to implement 2D dynamics, we first swap some pairs of the system and memory qubits to obtain the layout shown in the middle.
In the middle layout, we can see that many of the system qubits are connected (the light blue line).
We implement a depth-$4$ 1D random quantum circuit for each light blue line.
Each depth-$4$ circuit corresponds to 1 layer of single-qubit gates, 1 layer of two-qubit gates, 1 layer of single-qubit gates, and 1 layer of two-qubit gates.
The partitioning of the two layers of two-qubit gates are different.
Then, we swap all pairs of qubits to obtain the layout shown on the right.
The right-most layout connects a different set of system qubits (the light blue line).
We again implement a depth-$4$ 1D random quantum circuit for each light blue line.
After that, we move back to the middle layout and repeat for multiple rounds. After sufficiently many repetitions, the $n$ qubits become globally entangled.

\subsubsection{An unsupervised machine learning model}

For each circuit, we create a feature vector by obtaining statistics for each bit in the measurement outcome bitstring.
In conventional experiments, each experiment produces an $n$-bit measurement outcome.
In quantum-enhanced experiments, each experiment produces a $2n$-bit measurement outcome.
In the following, we consider $\ell = n$ or $2n$ depending on either we are running conventional or quantum-enhanced experiments.
For an $\ell$-bit measurement outcome, we obtain a feature vector of size $2\ell$ including the first and second moment of each bit.
After constructing a feature vector for each circuit, we map the feature vector to an infinite-dimensional reproducing kernel Hilbert space (corresponding to a Gaussian kernel) that includes all the polynomial expansions of the feature vector.
Then, we find a low-dimensional subspace in the infinite-dimensional Hilbert space using principal component analysis (PCA) \cite{scholkopf1998nonlinear}.
The entire procedure can be performed efficiently using kernel PCA \cite{scholkopf1998nonlinear}.
Kernel PCA is implemented using scikit-learn \cite{sklearn_api}.

In Fig.~\ref{fig:AdvantageExpt2} of the main text, we show a one-dimensional subspace found by the unsupervised ML model, for both 1D and 2D random quantum circuits.
We can use this one-dimensional representation to classify the circuits into two classes (by splitting the one-dimensional representation in the middle).
Then, we can evaluate the accuracy of the unsupervised ML model by checking the percentage of circuits that are correctly classified as general circuits or as $T$-symmetric circuits.
A two-dimensional subspace found by the unsupervised ML model, and an assessment of classification accuracy, are discussed in Appendix~\ref{sec:additional-exp}

\subsection{Additional experimental results}
\label{sec:additional-exp}

\begin{figure*}[t]
\centering
\includegraphics[width=0.75\textwidth]{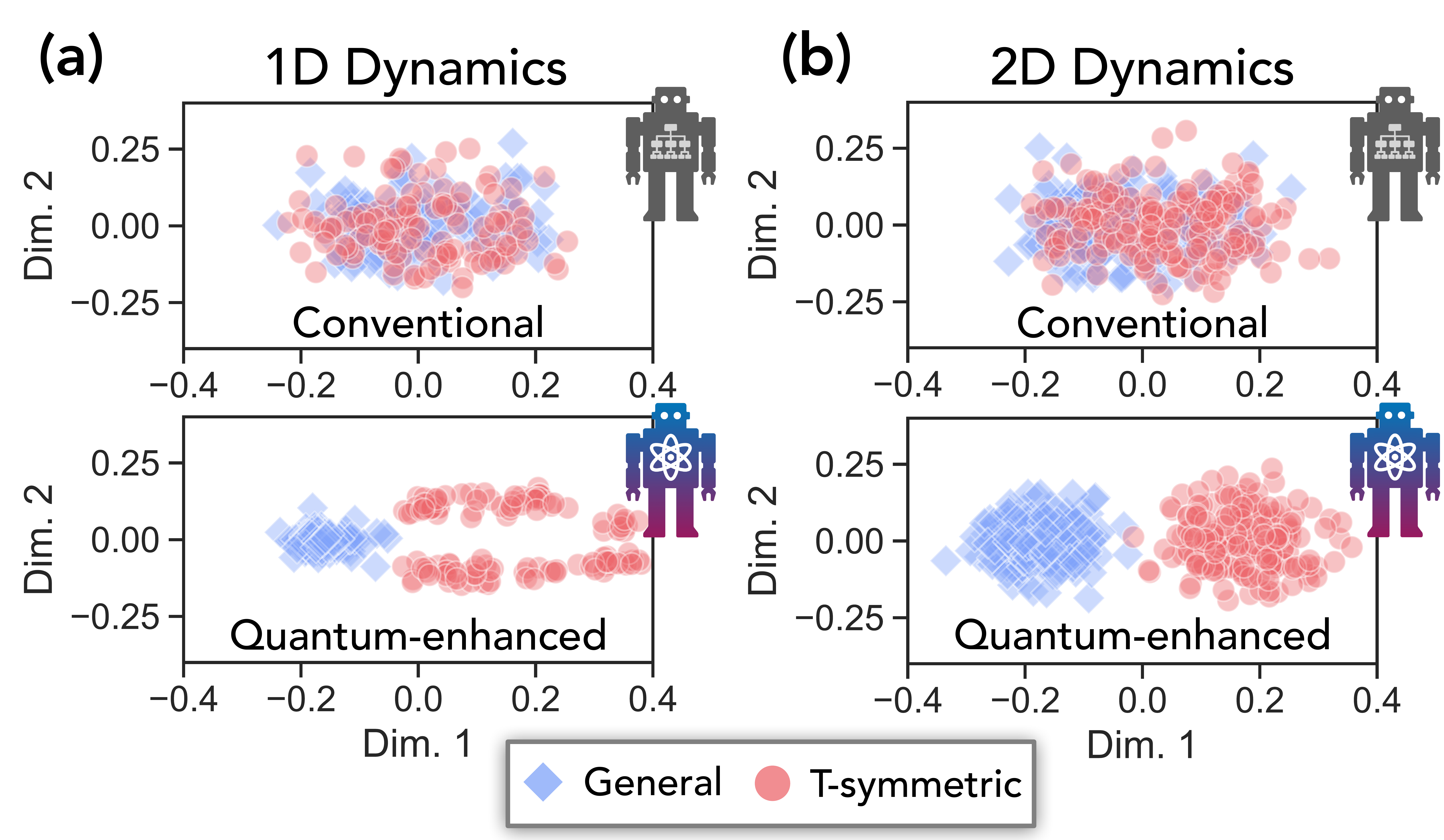}
    \caption{
    \emph{Two-dimensional representation learned by unsupervised ML for (a) 1D dynamics and (b) 2D dynamics.}
    Each point in the two-dimensional plane corresponds to a distinct physical process.
    Half of the processes have time-reversal symmetry (blue diamonds) while the other half do not (red circles).
    When fed with data from quantum-enhanced experiments, the ML model accurately discovers the underlying symmetry pattern.
    In contrast, the ML model fails to do so when fed with data from conventional experiments.
    \label{fig:2DPCA}
    }
\end{figure*}

In Supp.~Fig.~\ref{fig:2DPCA}, we provide the two-dimensional representations learned by unsupervised ML for the various random quantum circuits investigated in our conventional and quantum-enhanced experiments. 
(One-dimensional representations are presented in Fig.~\ref{fig:AdvantageExpt2} in the main text.)
We see that in the second dimension found by unsupervised ML using quantum-enhanced experiments for 1D dynamics, $T$-symmetric dynamics are clustered into two groups.
Further inspection shows that the unsupervised ML model has learned substructure corresponding to the parity of the depth of the evolution (recall that the depth is always an integer).
In principle, the unsupervised ML model should be able to learn a wide variety of structures in the dynamics.
Notably, we see that it places the distinction between general unitary dynamics and $T$-symmetric dynamics as the major axis (the first dimension), and places less prominent structure as the second major axis (the second dimension).

In Supp.~Fig.~\ref{fig:acc-PCA}, we provide the accuracy of the unsupervised ML model for distinguishing between general unitary dynamics and $T$-symmetric dynamics.
We see a substantial advantage for using the quantum-enhanced strategy in both the physical experiments and the noiseless simulation.
We perform brute-force noiseless simulation for conventional experiments because the system size is at most $20$.
The noiseless simulation for quantum-enhanced experiments uses the fact that $(U \otimes U) \frac{1}{\sqrt{2^n}}\sum_{i=0}^{2^n-1} \ket{i i} = (U U^T \otimes I) \frac{1}{\sqrt{2^n}}\sum_{i=0}^{2^n-1} \ket{i i}$, hence we can effectively reduce the simulation to a system size at most $20$.

\begin{figure*}[t]
\centering
\includegraphics[width=0.98\textwidth]{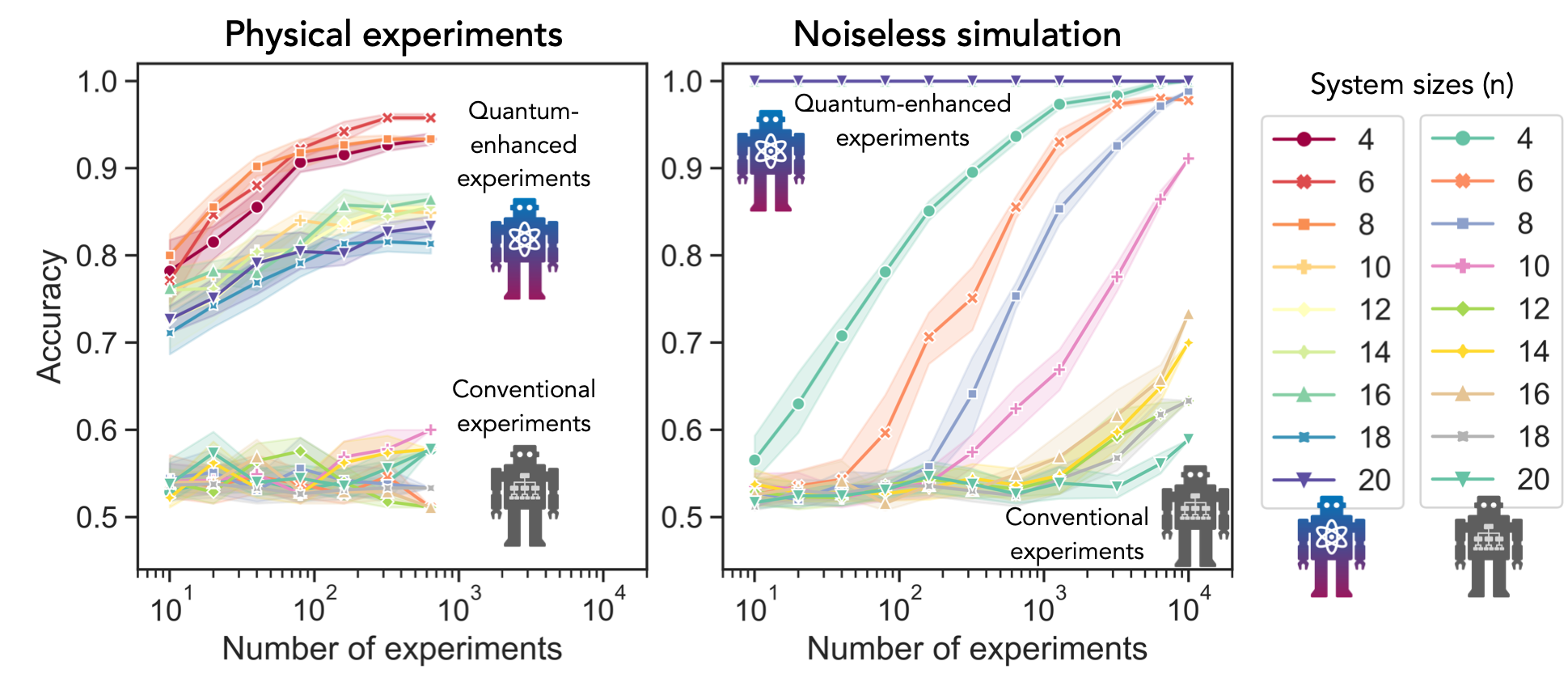}
    \caption{
    \emph{Accuracy of the unsupervised ML model for classifying general unitary and $T$-symmetric dynamics.}
    For each system size, we generate $100$ different circuits for each of the two classes (general and $T$-symmetric).
    The one-dimensional representation found by the unsupervised ML model is used to classify the $200$ circuits into two classes.
    We consider both physical experiments and noiseless simulations. Accuracy is plotted as a function of the number of experiments in both the conventional and quantum-enhanced settings.
    \label{fig:acc-PCA}
    }
\end{figure*}

\subsection{Performance and characterization data}
\label{sec:PerformanceData}

The performance of the device was characterized before each run.
The measurement data is collected explicitly and used for measurement error mitigation in the prediction process of the supervised neural network model (see discussion in the last part of Appendix~\ref{sec:NN-data-Q}).
The task of learning quantum states is largely limited by the qubit measurement fidelities.
A representative sample of the data from the device is reported in Fig.~\ref{fig:ReadoutError}, where one can see that typical readout errors (conflated with errors in preparing zero and one states) range from 3\% to 7\%.
For transmons, the $|0\rangle$ preparation has a small error; hence Supp.~Fig.~\ref{fig:ReadoutError}(a) is dominated by the readout error.
Furthermore, the single-qubit gate error (shown in Supp.~Fig.~\ref{fig:GateError}) is much smaller than the error shown in Fig.~\ref{fig:ReadoutError}(b), hence the error shown in Supp.~Fig.~\ref{fig:ReadoutError}(b) is mostly due to readout errors rather than gate errors.
During the actual run of the experiments, we avoid using qubits with the worst readout errors by checking the measurement errors before the experiment and selecting the layout accordingly.

The task of learning quantum dynamics involves circuits of higher complexity and hence  is limited by both measurement errors and errors in two-qubit gates. 
For these experiments, we report in Supp.~Fig.~\ref{fig:GateError} the quality of the single-qubit and two-qubit gates across the device. This data was obtained via parallel cross-entropy benchmarking and single-qubit randomized benchmarking.
The typical single-qubit gate error is around $0.001$ to $0.005$, while the typical two-qubit gate error is around $0.01$ to $0.05$.

\begin{figure*}[htp]
\centering
\includegraphics[width=0.94\textwidth]{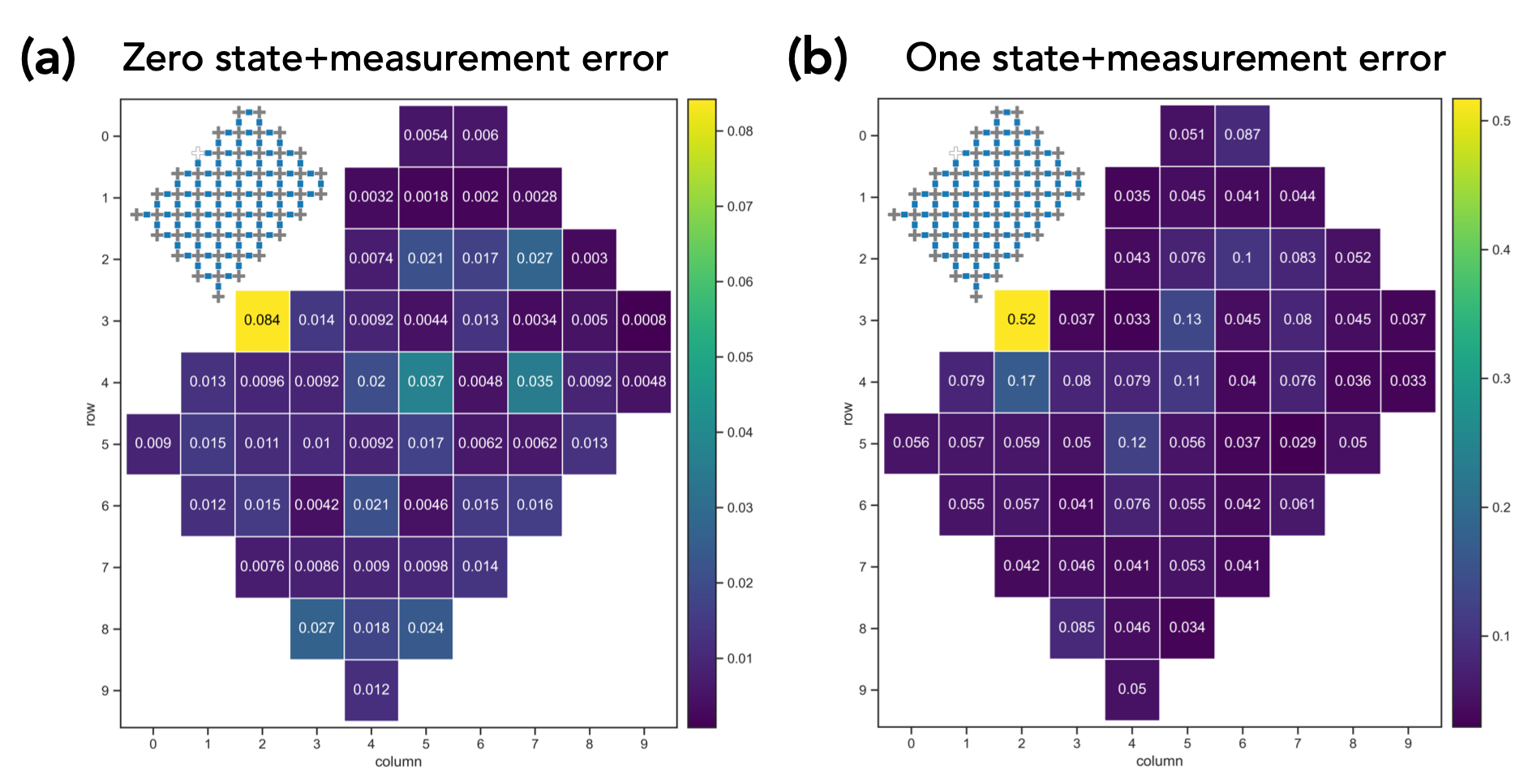}
    \caption{\emph{Sycamore state preparation and measurement error data.}\\
    \emph{(a) $|0\rangle$ state preparation and measurement error.} We prepare a noisy zero state $\ket{0}$ and measure in the computational basis using the noisy single-qubit readout. We show the probability of measuring $|1\rangle$ in the qubit readout.\\
    \emph{(b) $|1\rangle$ state preparation and measurement error.} We prepare a noisy one state $\ket{1}$ and measure in the computational basis using the noisy single-qubit readout.  We show the probability of measuring $|0\rangle$ in the qubit readout.\\
    While these values change over time, we present here a representative sample of the error. One can see that in accordance with physical expectations based on T$_1$ errors, the readout in the physical $1$ state is substantially higher than the $0$ state.
    \label{fig:ReadoutError}
    }
\includegraphics[width=0.94\textwidth]{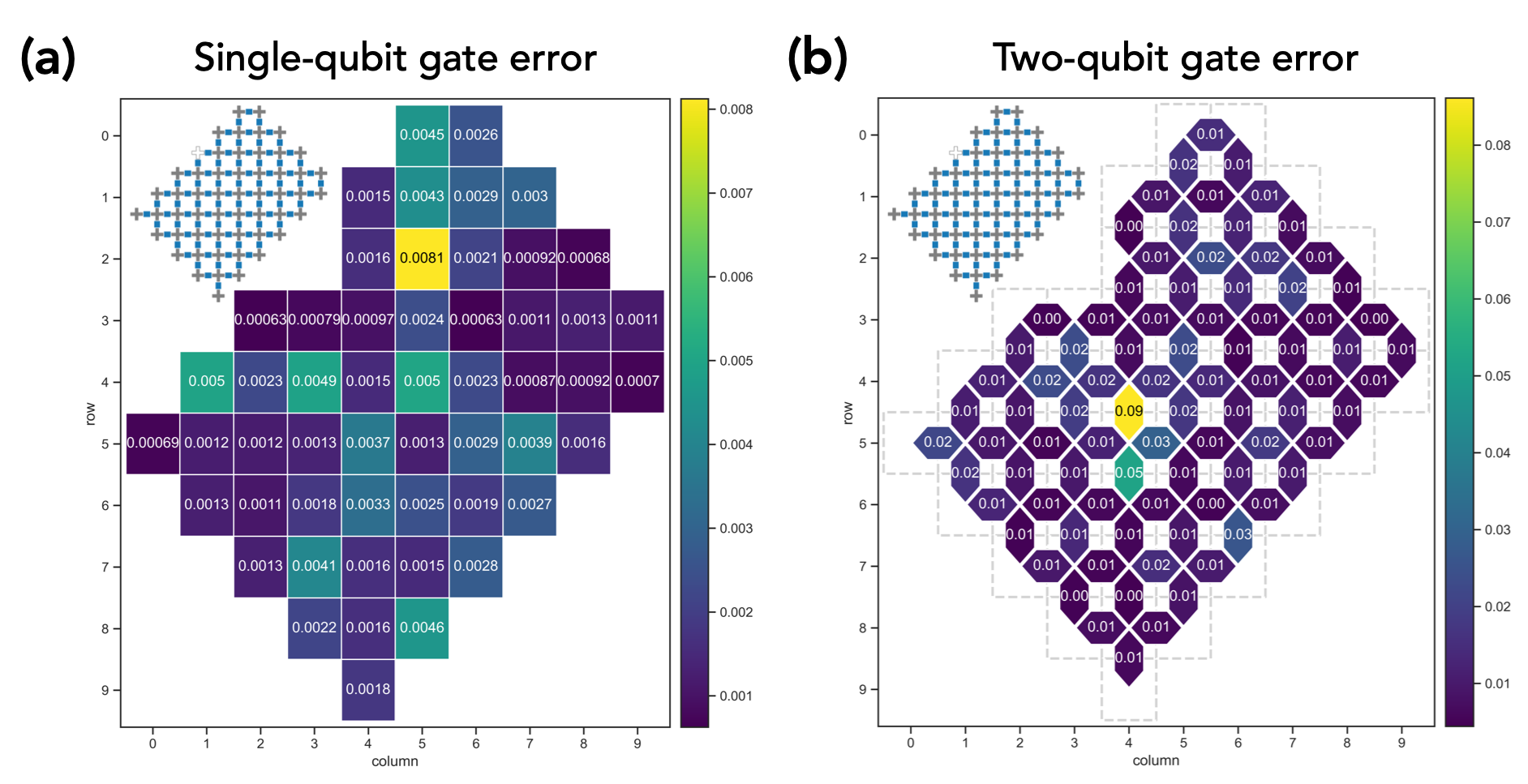}
    \caption{\emph{Sycamore single- and two-qubit gate error data.}\\
    \emph{(a) Single-qubit gate error.} The figure shows the error of single-qubit gates across the chip using parallel single-qubit randomized benchmarking. \\
    \emph{(b) Two-qubit gate error.} The figure shows the error across the chip of two-qubit gates being executed in parallel, as to account for errors that occur during simultaneous operation of qubits. We can see that the distribution of errors varies across the chip couplers, showing the extent to which performance is non-uniform.
    \label{fig:GateError}
    }
\end{figure*}

\section{A brief review on quantum information theory}
\label{sec:review}

In this section we review some relevant definitions and basic results in quantum information theory which are leveraged throughout our problem statements and proofs.  Specifically, we will discuss quantum processes, which are a general mathematical formalism for describing physical processes, and positive operator-valued measures (POVMs), which encompass all possible physical measurements.
Readers familiar with these concepts can skip this section.

\subsection{Definition and properties of quantum processes}
\label{sec:Qprocess}

For concreteness, let us consider a Hilbert space $\mathcal{H}_S \simeq \mathbb{C}^d$.  Here the subscript $S$ stands of `system', since the Hilbert space describes the space of states of some particular system we wish to study.  Given a density matrix $\rho$ on this Hilbert space, we might ask: how can it evolve in time?  The Schr\"{o}dinger equations tells us that a state can evolve via unitary time evolution, and as such a density matrix can evolve by $\rho \mapsto U \rho U^\dagger$.  However, there is a more general type of time evolution allowed by quantum mechanics.  Suppose that we append our Hilbert space by another $\mathcal{H}_E \simeq \mathbb{C}^{d'}$ which describes an external environment.  The joint Hilbert space is then the tensor product $\mathcal{H}_S \otimes \mathcal{H}_E$.  We can imagine having an initial state $\rho_S \otimes \rho_E$ which factorizes between the system and environment, and then evolving the state by a unitary on the joint Hilbert space which couples the system and environment: $\rho_S \otimes \rho_E \mapsto U_{SE} (\rho_S \otimes \rho_E) U_{SE}^\dagger$.  If we only have access to $\mathcal{H}_S$, then our knowledge of $U_{SE} (\rho_S \otimes \rho_E) U_{SE}^\dagger$ is described by performing a partial trace over the environment, namely $\text{tr}_E\!\left(U_{SE} (\rho_S \otimes \rho_E) U_{SE}^\dagger\right)$.  As such, if we are only aware of the initial density matrix $\rho_S$ on $\mathcal{H}_S$, then only having access to $\mathcal{H}_S$ the time evolution would appear to be
\begin{equation}
\label{E:SEchannel1}
\rho_S \longmapsto \text{tr}_E\!\left(U_{SE} (\rho_S \otimes \rho_E) U_{SE}^\dagger\right)\,.
\end{equation}
Note that, viewed as time evolution on $\rho_S$ alone, the above map is not unitary.  This is because information in $\rho_S$ can leak into the environment, and similarly information from the environment can influence the state on our system $S$ of interest.  This mapping is an example of a quantum process, which can be more compactly notated as $\rho_S \mapsto \mathcal{C}[\rho_S]$.  Here, $\mathcal{C}$ is a map from density matrices on $\mathcal{H}_S$ to (other) density matrices on $\mathcal{H}_S$.
We visualize this dynamical process in Supp.~Fig.~\ref{fig:qprocess}.

Our quantum process $\mathcal{C}$ has two properties that are worth highlighting:
\begin{enumerate}
    \item $\mathcal{C}$\textit{ is trace-preserving}.  This means that $\text{tr}(\mathcal{C}[\rho_S]) = \text{tr}(\rho_S)$.  The equality follows from the definition of $\mathcal{C}[\rho]$ via the right-hand side of~\eqref{E:SEchannel1}, since
    \begin{align}
    \text{tr}\!\left(\mathcal{C}[\rho_S]\right) &= \text{tr}_S\!\left(\text{tr}_E\!\left(U_{SE} (\rho_S \otimes \rho_E) U_{SE}^\dagger\right)\right) = \text{tr}\!\left(U_{SE} (\rho_S \otimes \rho_E) U_{SE}^\dagger\right)\\
    &= \text{tr}(\rho_S \otimes \rho_E) = \text{tr}(\rho_S) \,\text{tr}(\rho_E) = \text{tr}(\rho_S)\,,
    \end{align}
    where we have used the cyclicity of the trace to cancel $U_{SE}$ with $U_{SE}^\dagger$, and have also leveraged $\text{tr}(\rho_E) = 1$.
    \item $\mathcal{C}$\textit{ is completely positive}.  Suppose we append our system Hilbert space $\mathcal{H}_S$ by ancillas $\mathcal{H}_A$ to arrive at the joint Hilbert space $\mathcal{H}_A \otimes \mathcal{H}_S$. Then complete positivity means that for any density matrix $\rho_{AS}$ on this joint system (and for any choice of ancilla Hilbert space), $(\text{Id}_A \otimes \mathcal{C})[\rho_{AS}]$ is positive-semidefinite; here $\text{Id}_A$ acts as the identity on the ancillas.  To see why this property holds, we can write out $(\text{Id}_A \otimes \mathcal{C})[\rho_{AS}]$ more explicitly:
    \begin{equation*}
    (\text{Id}_A \otimes \mathcal{C})[\rho_{AS}] = \text{tr}_E\!\left((I_A \otimes U_{SE}) (\rho_{AS} \otimes \rho_E) (I_A \otimes U_{SE}^\dagger)\right)\,.
    \end{equation*}
    Since the right-hand side is merely performing a unitary transformation on the density matrix $\rho_{AS} \otimes \rho_E$ and then tracing out a subsystem (i.e.~the environment subsystem), positive semi-definiteness is preserved.
\end{enumerate}
We have thus shown that our $\mathcal{C}$ is a completely positive, trace-preserving (CPTP) linear map from density matrices on $\mathcal{H}_S$ to density matrices on $\mathcal{H}_S$.  Henceforth, when we refer to a map as being CPTP, we will implicitly suppose that the map is linear.  Moreover, we will interchangeably call a CPTP map a quantum process.

What is not immediately obvious is the following fact:
\begin{theorem}[Stinespring dilation]
Any CPTP map $\mathcal{C}$ taking density matrices on $\mathcal{H}_S \simeq \mathbb{C}^d$ to density matrices on $\mathcal{H}_S \simeq \mathbb{C}^d$ can be written in the form
\begin{equation*}
\mathcal{C}[\rho_S] = \text{\rm tr}_E\!\left(U_{SE} (\rho_S \otimes \rho_E) U_{SE}^\dagger\right)
\end{equation*}
for any $\rho_S$, where $U_{SE}$, $\rho_E$, and the dimension of the environment $d'$ are all fixed.
\end{theorem}
\noindent This theorem means that any CPTP map on a density matrix can be realized as a unitary operation on a larger system, i.e.~coupling the density matrix to an appropriate environment and evolving the joint state and ultimately tracing out the environment.  In this sense, a quantum process is the most general form of evolution of a density matrix.  Note that a special case of a quantum process is simply a unitary channel, i.e.~$\mathcal{C}[\rho] = U \rho U^\dagger$.  A way of summarizing the above Theorem is that a quantum process that is not a unitary channel can be thought of as implementing open system dynamics.

\begin{figure}[t]
    \centering
    \includegraphics[width=0.7\textwidth]{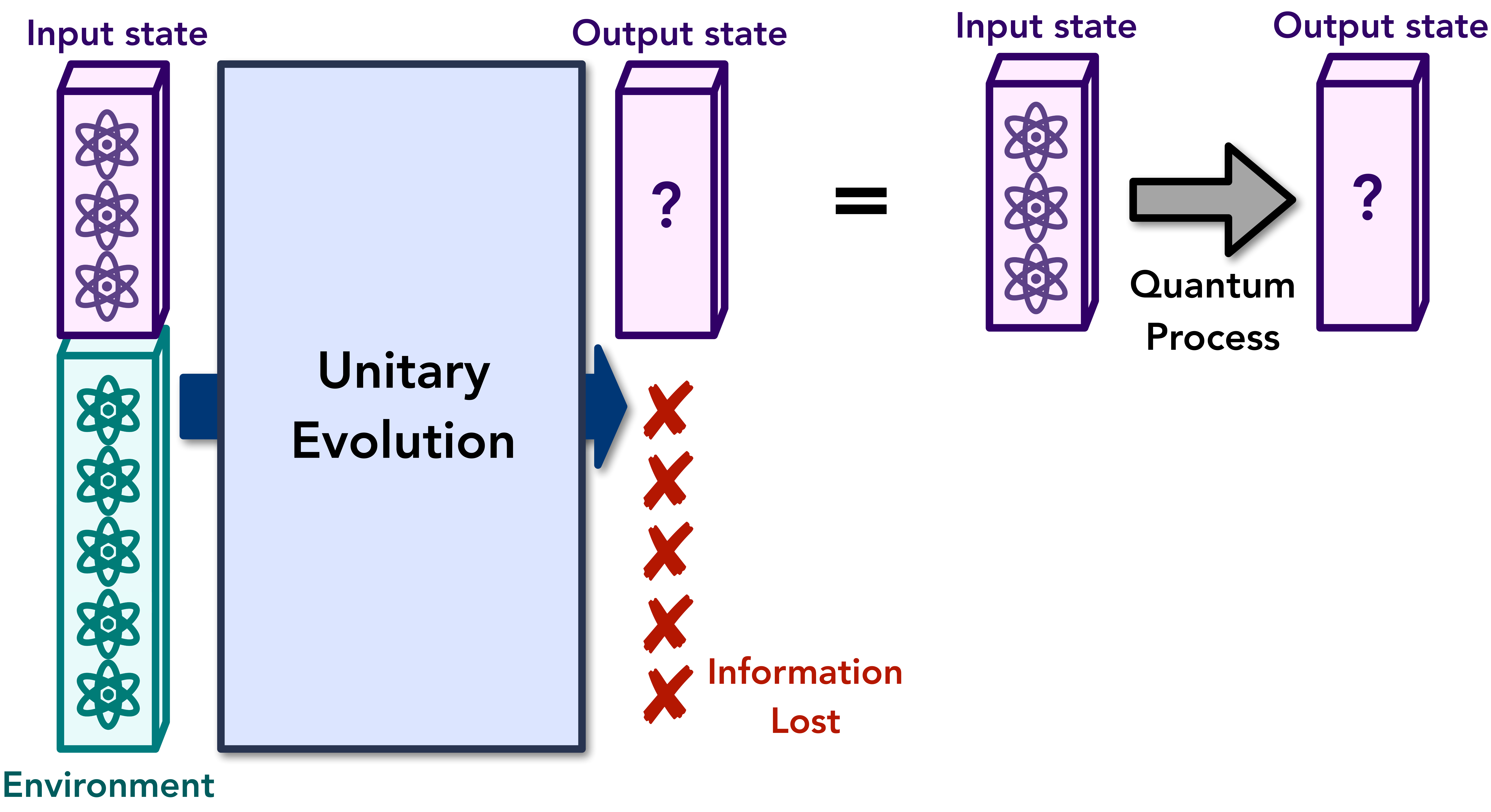}
    \caption{\emph{Illustration of quantum process: a formalism for describing physical processes.} Quantum process is also known as quantum operation or quantum dynamical map, and is often referred to as quantum channel in quantum communication theory. }
    \label{fig:qprocess}
\end{figure}

\subsection{Definition and properties of POVMs}
\label{sec:POVMdef}

The most conventional way to measure a quantum state $|\psi\rangle$ is by decohering it with respect to a complete orthonormal basis.  More specifically, suppose that $|\psi\rangle \in \mathbb{C}^d$ and we choose some complete orthonormal basis $\{|i\rangle\}_{i=0}^{d-1}$ of $\mathbb{C}^d$.  Then upon measuring $|\psi\rangle$ with respect to this basis, we will measure $|\psi\rangle$ to be in the state $|i\rangle$ with probability $\text{Prob}(i) = |\langle i |\psi\rangle|^2$.  Analogously for a density matrix $\rho$ on the same Hilbert space, if we measure it with respect to the same orthonormal basis we will measure the state to be $|i\rangle \langle i|$ with probability $\text{Prob}(i) = \text{tr}(|i\rangle \langle i| \, \rho)$.

There is a nice way of conceptualizing measurements which will admit useful generalizations.  First, let us develop some notation.  We define $\Pi_i = |i\rangle \langle i|$ which is the projector onto state $|i\rangle$, and will speak of the collection of projectors $\{\Pi_i\}_{i=0}^{d-1}$.  It is readily seen that $\sum_{i=0}^{d-1} \Pi_i = I$ since this is just a resolution of the identity.  Observe that each $\Pi_i$ is Hermitian and positive semi-definite.  Now suppose we append to our Hilbert space another copy $\mathbb{C}^d$. Then we can define a unitary on both copies which acts by
\begin{equation}
\label{E:pointermeas1}
U \big(|\psi\rangle \otimes |0\rangle\big) = \sum_{i=0}^{d-1} \Pi_i |\psi\rangle \otimes |i\rangle
\end{equation}
for any $|\psi\rangle$.  Note that, as required of a unitary,
\begin{equation*}
\big( \langle \psi| \otimes \langle 0| \big) U^\dagger U \big(|\psi\rangle \otimes |0\rangle \big) = \sum_{i,j=0}^{d-1} \langle \psi| \Pi_i \Pi_j |\psi\rangle \langle i | j\rangle = \sum_{i=0}^{d-1} \langle \psi| \Pi_i^2 |\psi\rangle = \sum_{i=0}^{d-1} \langle \psi| \Pi_i |\psi\rangle = 1
\end{equation*}
on account of $\Pi_i^2 = \Pi_i$ and $\sum_{i=0}^{d-1} \Pi_i = I$.  Given the right-hand side of~\eqref{E:pointermeas1}, we can make a measurement on the appended Hilbert space in the $\{|i\rangle\}_{i=0}^{d-1}$ basis; we will then measure the appended register to be in the state $|i\rangle$ with probability
\begin{equation}
\text{Prob}(i) = \big( \langle \psi| \otimes \langle 0| \big) U^\dagger \big(I \otimes |i\rangle \langle i| \big) U \big(|\psi\rangle \otimes |0\rangle \big) = \text{tr}(\Pi_i  |\psi\rangle \langle \psi|) = |\langle i|\psi\rangle|^2\,.
\end{equation}
Similarly, if we consider $\rho \otimes |0\rangle \langle 0|$, conjugate by $U$, and then measure the state of the ancilla, the probability of measuring the ancilla to be $|i\rangle$ is $\text{Prob}(i) = \text{tr}(\Pi_i \rho) = \text{tr}(|i\rangle \langle i|\, \rho)$.

We can think about the above in terms of the following procedure.  First we prepare a state $|\psi\rangle$; then we bring in an ancilla $|0\rangle$ and cause the two states to interact such that the ancilla goes into a state $|i\rangle$ upon coupling with the $|i\rangle$-component of $|\psi\rangle$.  This leads to the right-hand side of~\eqref{E:pointermeas1}.  The ancilla can be thought of as a proxy for the readout of a measurement apparatus: upon reading off the value of $|i\rangle$, we are informed that the state $|\psi\rangle$ has been projected into its $|i\rangle$-component.

\begin{figure}[t]
    \centering
    \includegraphics[width=0.65\textwidth]{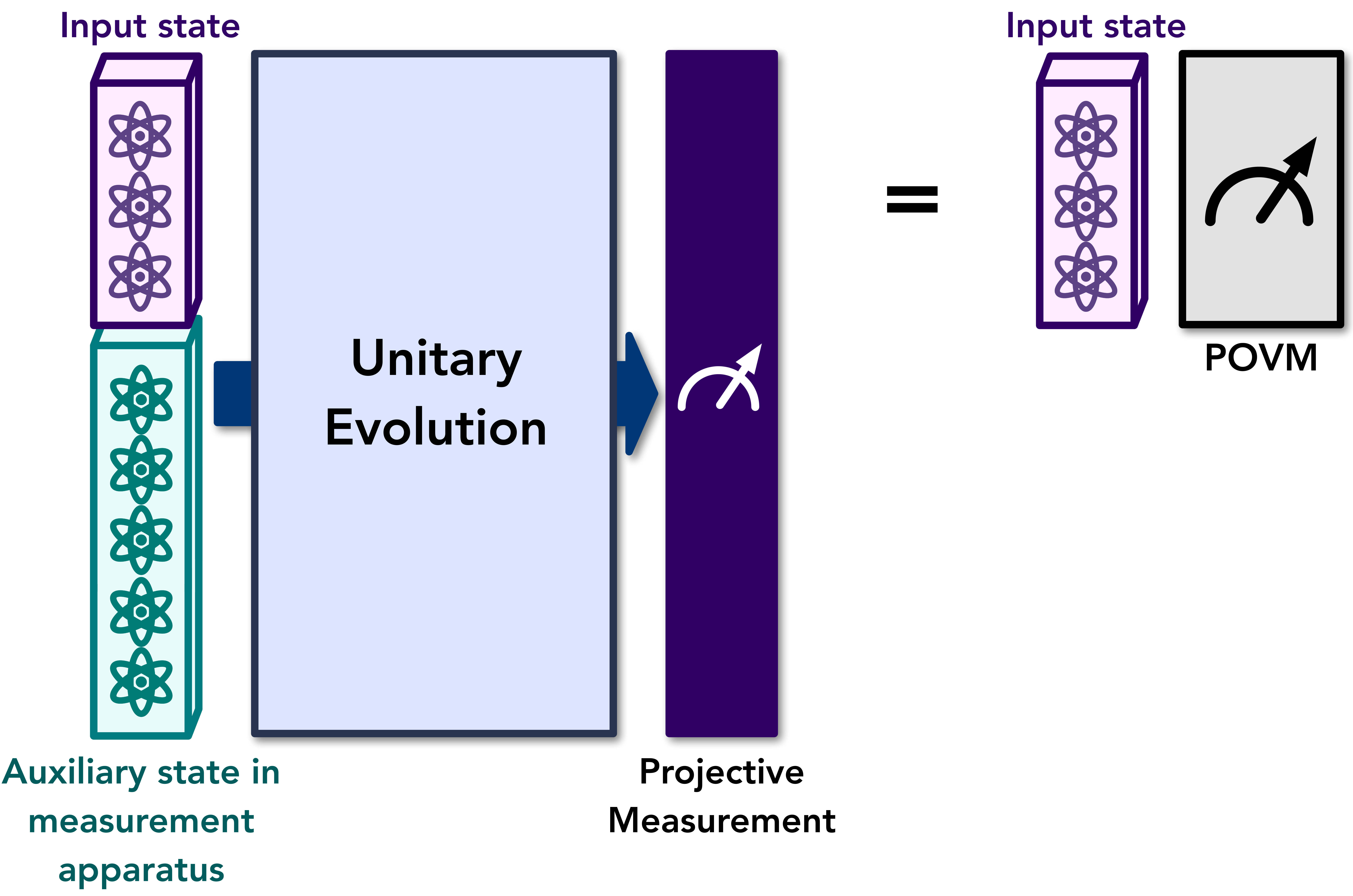}
    \caption{\emph{Illustration of POVM: a formalism encompassing all physical measurements.} POVM considers a composition of the input state with an auxiliary state in the measurement apparatus (which can be thought of as a set of ancilla qubits) that undergoes an unitary evolution, followed by a projective measurement. }
    \label{fig:povm}
\end{figure}

This type of procedure can be generalized as follows.  Suppose we have a set of $N$ $d \times d$ operators $\{M_i\}_{i=0}^{N-1}$ satisfying the completeness relation $\sum_{i=0}^{N-1} M_i^\dagger M_i = I$.  Let us append to our Hilbert space $\mathbb{C}^d$ and ancillary Hilbert space $\mathbb{C}^N$ with complete orthonormal basis $\{|i\rangle\}_{i=0}^{N-1}$.  Then we can consider a unitary map
\begin{equation}
\label{E:pointermeas2}
U \big(|\psi\rangle \otimes |0\rangle\big) = \sum_{i=0}^{N-1} M_i |\psi\rangle \otimes |i\rangle\,.
\end{equation}
The fact that $\big( \langle \psi| \otimes \langle 0| \big) U^\dagger U \big(|\psi\rangle \otimes |0\rangle \big) = 1$ can be checked using the completeness relation $\sum_{i=0}^{N-1} M_i^\dagger M_i = I$.  Now if we measure the ancilla with respect to the $\{|i\rangle\}_{i=0}^{N-1}$ basis, then we will measure the ancilla to be in the state $|i\rangle$ with probability $\text{Prob}(i) = | M_i |\psi\rangle|^2$.  If we performed an analogous procedure at the level of density matrices, namely starting with a state $\rho \otimes |0\rangle \langle 0|$, conjugating both sides by $U$, and then measuring the ancilla in the $\{|i\rangle\}_{i=0}^{N-1}$ basis, we would measure the ancilla to be $|i\rangle$ with probability $\text{Prob}(i) = \text{tr}(M_i^\dagger M_i \rho)$.
We visualize the above procedure in Supp.~Fig.~\ref{fig:povm}.

We can abstract this procedure into what is called a \textit{positive operator-valued measure} (POVM):
\begin{definition}[POVM]
A POVM is a set Hermitian, positive semi-definite operators $\{F_i\}_{i=0}^{N-1}$ on $\mathbb{C}^d$ satisfying the completeness relation $\sum_{i=0}^{N-1} F_i = I$.  A POVM measurement is a procedure in which, given a state $\rho$ on $\mathbb{C}^d$, an ancillary measurement apparatus registers the index $i$ with probability $\text{\rm tr}(F_i \rho)$.
\end{definition}
\noindent This relates to our previous procedure as follows.  We simply decompose $F_i = M_i^\dagger M_i$ (say, by a Cholesky decomposition) and perform the procedure previously stated with the $M_i$'s.

We remark that the term `measure' is used above in two distinct ways.  When we speak of a POVM, the M means measure in the sense of measure theory, since we can think of $\{F_i\}_{i=0}^{N-1}$ as comprising a type of discrete measure on the space of operator on $\mathbb{C}^d$.  Otherwise, we use `measure' in the sense of measurement.

A useful fact is that given a POVM $\{F_i\}_{i=0}^{N-1}$, we can refine it into another, larger POVM $\{F_{i,j}\}_{i=0, j=0}^{N-1, d-1}$ such that (1) each $F_{i,j}$ is rank-1, and (2) a POVM measurement of $\{F_{i,j}\}_{i=0, j=0}^{N-1, d-1}$ can simulate a POVM measurement of $\{F_i\}_{i=0}^{N-1}$.  Let us explain this construction.  Since each $F_i$ is a positive semi-definite Hermitian operator, we can diagonalize each operator as $F_i = \sum_{j=0}^{d-1} \lambda_{j}^{(i)} |v_{j}^{(i)}\rangle \langle v_{j}^{(i)}|$.  Then let $F_{i,j} := \lambda_j^{(i)} |v_{j}^{(i)}\rangle \langle v_{j}^{(i)}|$ which is manifestly positive semi-definite, Hermitian, and rank-1; it is also clear that $\sum_{i=0}^{N-1} \sum_{j=0}^{d-1} F_{i,j} = \sum_{i=0}^{N-1} F_i = I$. We can use a POVM measurement of $\{F_{i,j}\}_{i=0, j=0}^{N-1, d-1}$ to simulate a POVM measurement of $\{F_i\}_{i=0}^{N-1}$ by simply summing measurement results:
\begin{equation}
 \sum_{j=0}^{d-1} \text{tr}(F_{i,j} \rho) = \text{tr}(F_i \rho)\,.
\end{equation}
Accordingly, we can without loss of generality choose to work with rank-1 POVMs, since we can use these to simulate any other POVMs.

\section{Mathematical framework for proving exponential advantage}
\label{sec:separation}

One of the central ingredients for establishing exponential advantage is to prove an exponential lower bound for any learning algorithm with only external classical memory.
In this section, we present the basic framework for proving such lower bounds.
This framework enables us to establish a suite of exponential advantage in various physically-relevant tasks, which will be presented after this section.
The purpose of this section is to provide the readers with the essential tools to prove quantum advantage in the tasks they want to study.

The basic tools include the tree representation of learning algorithms (Appendix~\ref{sec:tree-def}), reduction to distinguishing tasks and  information-theoretic lower bounds (Appendix~\ref{sec:many-v-one-red}~and~\ref{sec:many-vs-many}).  Many of these techniques were introduced and leveraged in~\cite{chen2021exponential}.
We also present a novel \emph{partially-revealed many-versus-one distinguishing task} in Appendix~\ref{sec:part-many-v-one} that is crucial for realizing the advantage in practice.
Then, in Appendix~\ref{sec:noisedegradation}, we discuss how having noise in the unknown physical states and dynamics only makes the lower bounds for conventional experiments larger.
This result on the presence of noise is simple to establish but also crucial in practice because there is often noise in the unknown physical states and dynamics.
There are some other techniques presented in a theory paper written by some of the authors~\cite{chen2021exponential}, such as a multi-linear tensor analysis on the learning tree, which may be of interest to some readers.

\subsection{Tree representation}
\label{sec:tree-def}

We begin by presenting the tree representation for analyzing algorithms with only classical memory~\cite{chen2021exponential}.
The key idea is to track changes in the classical memory state in the algorithm using a graph, which we can take to be a rooted tree.
We consider each node $u$ of the graph to be a classical memory state.
Based on the memory state, the algorithm performs an experiment to obtain a measurement outcome $s$.

\subsubsection{Experiments for learning physical world}

To motivate the definitions in the sequel, we separately describe the two types of experimental setups that we focus on in this work: one on learning an unknown physical state and the other on learning an unknown process~\cite{huang2021information, aharonov2021quantum, chen2021exponential}.
\begin{itemize}
    \item \emph{Learning an unknown physical state}: A physical state is represented by a density matrix $\rho$. An algorithm leveraging the classical memory state $u$ measures the physical system $\rho$ using a rank-$1$ POVM $\{w^u_s \ketbra{\phi^u_s}{\phi^u_s} \}$ with $\sum_{s} w_s^u \ketbra{\phi^u_s}{\phi^u_s} = \Id$.
    Note that from the discussion in Appendix~\ref{sec:POVMdef}, we can always consider rank-$1$ POVMs only.
    The measurement outcome $s$ occurs with probability
    \begin{equation}
     w_s^u \bra{\phi^u_s} \rho \ket{\phi^u_s}.
    \end{equation}
    Here, the rank-$1$ POVM $\{w^u_s \ketbra{\phi^u_s}{\phi^u_s} \}$ depends on the classical memory state $u$.
    \item \emph{Learning an unknown physical process}: A physical process is represented by a quantum process $\cE$ (equivalently, a CPTP map). An algorithm leveraging the memory state $u$ prepares an initial state $\ket{\psi^u}$, feeds it into the physical evolution $\mathcal{E}$, and measures the output state $\mathcal{E}(\ketbra{\psi^u}{\psi^u})$ with a rank-$1$ POVM $\{w^u_s \ketbra{\phi^u_s}{\phi^u_s} \}$ with $\sum_{s} w_s^u \ketbra{\phi^u_s}{\phi^u_s} = \Id$.
    The outcome $s$ is obtained from the experiment with probability 
    \begin{equation}
        w_s^u \bra{\phi^u_s} \mathcal{E}(\ketbra{\psi^u}{\psi^u}) \ket{\phi^u_s}\,.
    \end{equation}
    In this case, both the initial state and the measurement depend on the classical memory state $u$.
    
    We consider the initial state $\ket{\psi^u}$ to be an $(n+n')$-qubit state, where $\mathcal{E}$ acts on the first $n$ qubits.
    The rank-$1$ POVM $\{w^u_s \ketbra{\phi^u_s}{\phi^u_s} \}$ is on an $(n+n')$-qubit state.
\end{itemize}

\subsubsection{Dynamics of the learning algorithm}

The classical memory state of the learning algorithm is initialized in a certain state, which we represent by the root node $r$.
The memory state of the algorithm begins at the root $r$. Each measurement outcome~$s$ resulting from a single experiment causes the algorithm to transition to a node neighbor of $r$.
Whenever the algorithm obtains a measurement outcome $s$, the memory state changes.
This is represented by a transition from a node $u$ to another node $v$,
\begin{equation}
    u \xrightarrow{s} v.
\end{equation}
The directed edge $e = (u, v, s)$ from $u$ to $v$ represents the transition of the memory in an algorithm when we receive the measurement outcome $s$.
We illustrate the transition under a single experiment in Supp.~Fig.~\ref{fig:Tree}(a).

If a different $s$ leads us to the same node, then the algorithm is not retaining full information of the measurement outcome. An example is given in Supp.~Fig.~\ref{fig:Tree}(b).
Since we do not limit the size of the classical memory, there is no need to lose (or forget) information.
Hence, all the outgoing edges of the root node $r$ indexed by the measurement outcome $s$ will point to distinct nodes.
The same argument holds for any node in the graph.
More precisely, every outgoing edge from a node $u$ will connect to a node $v$, such that $v$ has exactly one incoming edge (the edge is from $u$).
The only node in the graph without an incoming edge is the root $r$.
This is exactly the definition of a directed rooted tree $\mathcal{T}$.
We will focus on an algorithm that performs $T$ experiments. This means the depth of the tree, namely the number of edges in any root-to-leaf path, will be $T$.
The tree representation is shown in Supp.~Fig.~\ref{fig:Tree}(c).

\begin{figure}[t]
    \centering
    \includegraphics[width=0.9\textwidth]{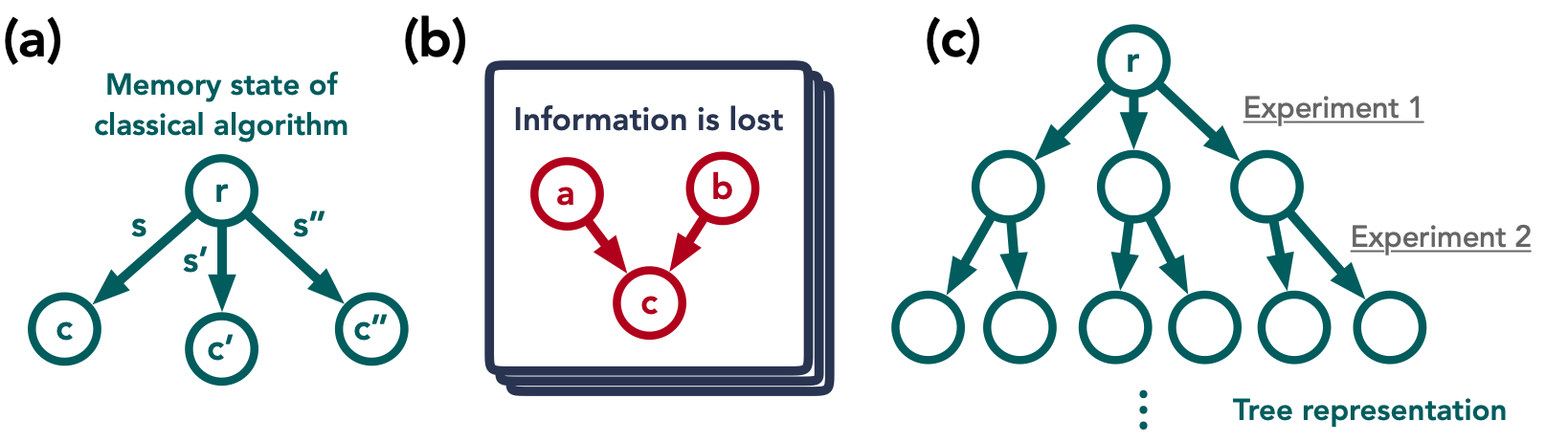}
    \caption{\emph{Illustration of the tree representation for a learning algorithm.} \emph{(a) Dynamics of memory.} The memory state changes based on the measurement outcome $s$. \emph{(b) No cycles.} If two memory states $a, b$ transition into the same memory state $c$, then some information is lost. \emph{(c) Tree representation of an algorithm.} When no information is lost, the transition graph of the memory states must be a directed tree. Each layer of the tree corresponds to one experiment. After $T$ experiments, the memory state is represented by a node in the $T$-th layer.  }
    \label{fig:Tree}
\end{figure}

When we execute the algorithm to achieve a certain task (such as to verify entanglement, or learn a model of the physical system), the entire dynamic process of how the memory state changes will be represented by a path from the root $r$ to a leaf $\ell$ in the tree $\mathcal{T}$,
\begin{equation}
    u_0 = r \xrightarrow{s_1} u_1 \xrightarrow{s_2} u_2 \xrightarrow{s_3} \ldots \xrightarrow{s_{T-1}} u_{T-1} \xrightarrow{s_T} u_T = \ell.
\end{equation}
To establish a lower bound against any learning algorithm for a particular task, we need to analyze each such path along with the probability that the path is taken.

\subsection{Many-versus-one distinguishing tasks}
\label{sec:many-v-one}

\subsubsection{Reduction}
\label{sec:many-v-one-red}

In a learning task, we often want the learning algorithm to be able to make accurate predictions about some properties of the unknown physical system or dynamics.
We will have a set of states (the mathematical representation of a physical system) or a set of channels (the mathematical representation of physical dynamics) to which we assume the unknown system or dynamics belong.
The basic technique we employ in all of our proofs is to pick out one of the states/channels as the null hypothesis, and consider all the rest as the alternative hypothesis~\cite{huang2021information, chen2021exponential}.
Let $\cX$ denote the set of possible states/channels.
\begin{itemize}
    \item \emph{Null hypothesis}: The unknown state/channel is an element $X_0 \in \cX$. To establish a tight lower bound, we should choose an $X_0$ that we think is \emph{close} to every other state/channel in $\cX$.
    \item \emph{Alternative hypothesis}: The unknown state/channel is a random element in $\cX \setminus \{X_0\}$.
\end{itemize}
Furthermore, we need to choose $X_0$ such that the desired property we would like to learn enables us to distinguish between $X_0$ and the entire set of $\cX \setminus \{X_0\}$.

To prove a lower bound against any classical algorithm, we try to answer the following question.
\begin{center}
    \emph{How hard is it to distinguish the alternative hypothesis from the null hypothesis?}
\end{center}
Because the alternative hypothesis consists of many elements and the null hypothesis consists of only one element,
we refer to this distinguishing task as the \emph{many-versus-one distinguishing task}.

\subsubsection{Information-theoretic lower bound}
\label{sec:many-v-one-info}

\begin{figure}[t]
    \centering
    \includegraphics[width=0.88\textwidth]{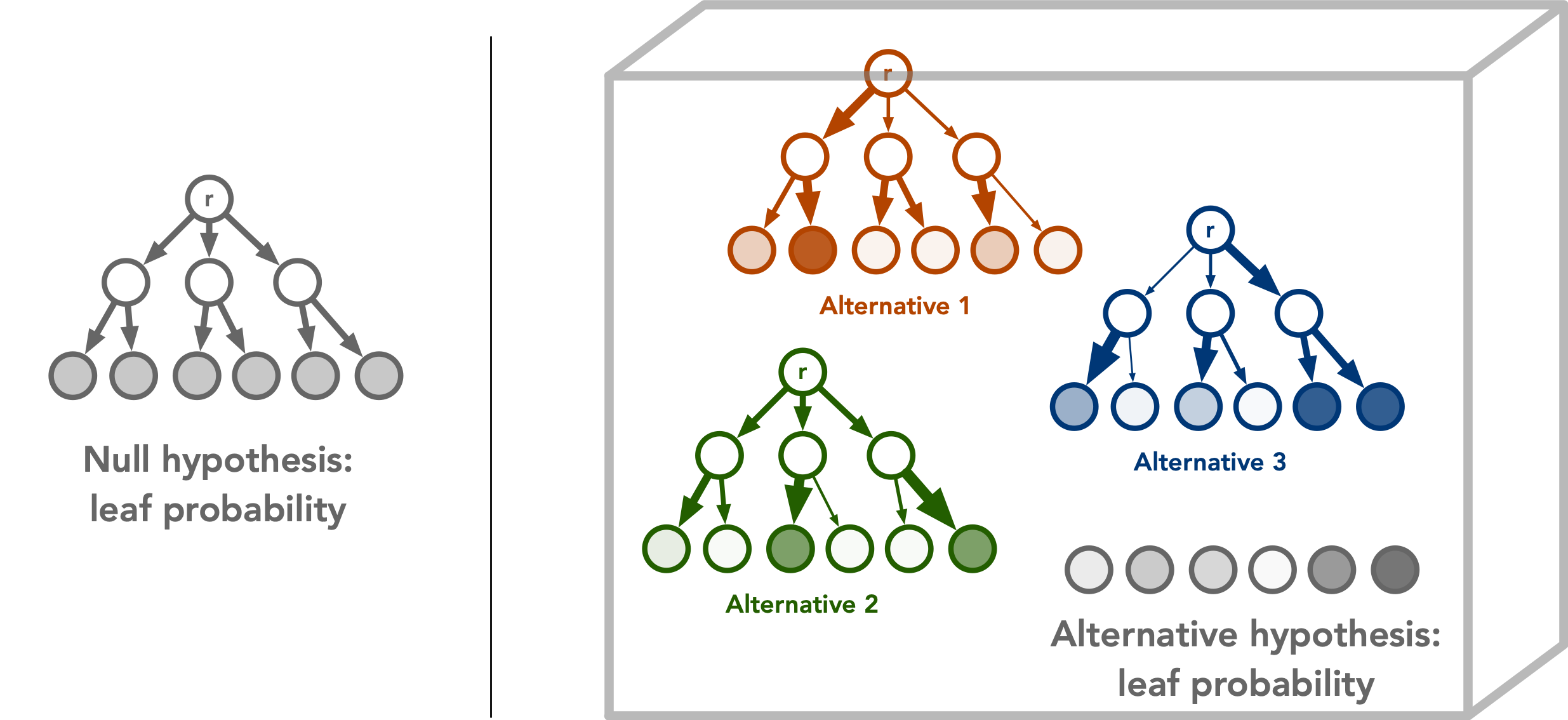}
    \caption{\emph{Illustration for the leaf probability distribution.} The leaf probability distribution depends on the unknown physical state/process and the learning algorithm. In the null hypothesis, we have a single state/process, which gives rise to a probability distribution over leaves. In the alternative hypothesis, there are multiple possible states/processes. Each state/process produces a different leaf probability distribution. The leaf probability distribution for the alternative hypothesis is the average of all the leaf probability distributions. }
    \label{fig:leaf-alt}
\end{figure}

In order to establish a lower bound for the many-versus-one distinguishing task, we need to first discuss the leaf probability distribution in the tree representation of the learning algorithm~\cite{huang2021information, chen2021exponential}.
Recall that depending on the unknown state/process, the transition probabilities among the memory states in the learning algorithm will be different.
This is because the outcome probability for each experiment differs when the unknown state/process differs.
Therefore, the probability to traverse a certain path in the tree representing an execution of the learning algorithm will change according to the unknown state/process.
Hence, the probability to arrive at a particular leaf node in the depth-$T$ tree will change.
An illustration is given in Supp.~Fig.~\ref{fig:leaf-alt}.

For each element $X$ in $\cX$, the set of all admissible states/channels, we write the probability distribution over leaves as
\begin{equation}
    p_X( \ell ), \quad \ell: \mbox{leaf node of the tree}.
\end{equation}
The probability distribution over the leaves $\ell$ for the null hypothesis and for the alternative hypothesis are respectively
\begin{equation}
    p_{X_0}(\ell) \,\,\,\, \mbox{and} \,\, \E_{X \in \cX \setminus \{X_0\}} p_{X}(\ell).
\end{equation}
The probability distribution over $X \in \cX \setminus \{X_0\}$ in the expectation $\E_{X \in \cX \setminus \{X_0\}}$ is arbitrary.
We should choose the probability distribution that yields the largest lower bound.

Suppose that the null hypothesis and the alternative hypothesis are true with probability $1/2$ each.
If we want to use the memory state of the learning algorithm to distinguish between the null hypothesis and the alternative hypothesis, then the success probability of any procedure is upper bounded by
\begin{equation}
    \frac{1}{2} + \frac{1}{2} \mathrm{TV}\left(p_{X_0}, \E_{X \in \cX \setminus X_0} p_{X} \right) = \frac{1}{2} + \frac{1}{4} \sum_{\ell} \left| p_{X_0}(\ell) - \E_{X \in \cX \setminus X_0} p_{X}(\ell) \right|,
\end{equation}
which is also known as LeCam's two-point method.
$\mathrm{TV}(p_0, p_1)$ is the total variation distance between the two probability distributions $p_0, p_1$.

Intuitively, as we perform more experiments, the depth of the tree increases and the total variation distance between the leaf probability distribution increases.
If we want to achieve a prediction accuracy of $p \geq \tfrac{1}{2}$, then we need the total variation distance to be lower bounded by,
\begin{equation} \label{eq:TVbound}
    \mathrm{TV}\left(p_{X_0}, \E_{X \in \cX \setminus X_0} p_{X} \right) \geq 2 p - 1.
\end{equation}
On the other hand, the total variation distance can be upper bounded by a monotonically increasing function of the number of experiments~$T$ equal to the depth of the tree.
Altogether, this allows us to lower bound the number of experiments $T$ by a function of the success probability $p$.

\subsection{Many-versus-many distinguishing task}
\label{sec:many-vs-many}

Sometimes, it is easier to first reduce the learning task to a many-versus-many distinguishing task before reducing to a many-versus-one task. This technique is used in Appendix~\ref{sec:qpca} to prove exponential advantage for quantum principal component analysis.
Consider $\cX$ to be the set of allowed states/channels.
We consider a subset $\cA \subseteq \cX$, and define $\cB = \cX \setminus \cA$.
Here, we consider the following two hypotheses.
\begin{itemize}
    \item Hypothesis A: The unknown state/channel is a random element in $\cA$.
    \item Hypothesis B: The unknown state/channel is a random element in $\cB$.
\end{itemize}
Assume each hypothesis happens with probability $1/2$.
The goal is to distinguish which hypothesis is true.
Because $\cA, \cB$ can contain many elements in $\cX$, we refer to this as the many-versus-many distinguishing task.
In Supp. Fig.~\ref{fig:many-vs-one}, we visualize the difference between the \emph{many-versus-many distinguishing task} and the other tasks.

Following a similar derivation as the many-versus-one distinguishing task, for any learning algorithm in the conventional setting, we represent the algorithm as a learning tree.
Given a tree representation $\cT$, the success probability for any procedure to distinguish hypothesis A and B using the final memory state of the learning algorithm is upper bounded by
\begin{equation}
    \frac{1}{2} + \frac{1}{2}\mathrm{TV}\left(\E_{X \in \cA} p_{X}, \E_{X \in \cB} p_{X} \right),
\end{equation}
where $p_X$ is a probability distribution over the leaf nodes of the tree $\cT$ when the unknown state/channel is $X$.
Hence, if we want to achieve a prediction accuracy of $p \geq \tfrac{1}{2}$, then we need the total variation distance to be lower bounded by
\begin{equation} \label{eq:TVbound-many-vs-many}
    \mathrm{TV}\left(\E_{X \in \cA} p_{X}, \E_{X \in \cB} p_{X} \right) \geq 2 p - 1\,.
\end{equation}
This last inequality will be important for establishing the lower bound on the number of experiments $T$.

\subsection{Partially-revealed many-versus-one distinguishing task}
\label{sec:part-many-v-one}

\begin{figure}[t]
    \centering
    \includegraphics[width=1.0\textwidth]{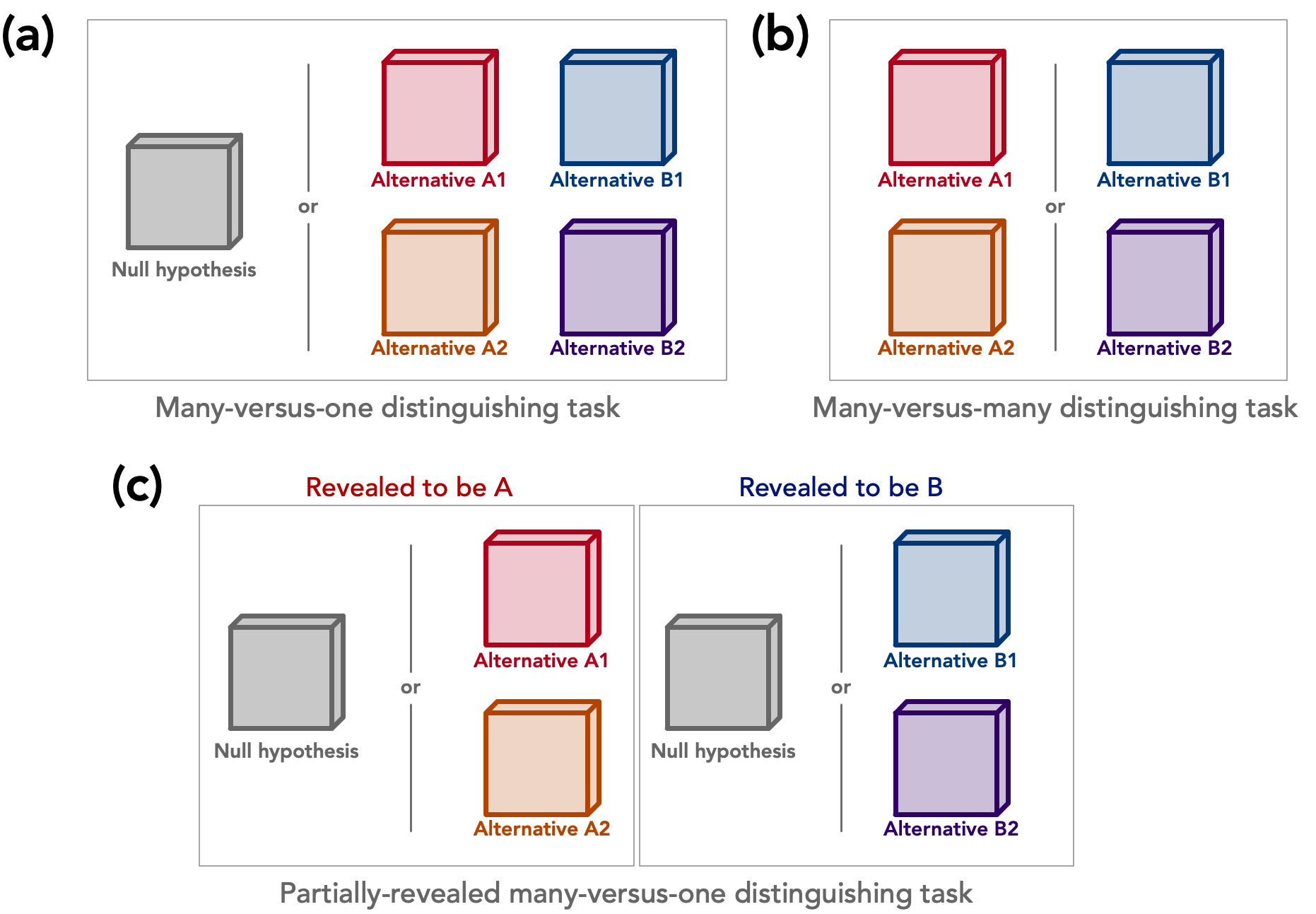}
    \caption{\emph{Visualization of the different distinguishing tasks.} \emph{(a)} In the many-versus-one distinguishing task, we are distinguishing between the null hypothesis and the alternative hypothesis (which can be one of many alternatives).
    \emph{(b)} In the many-versus-many distinguishing task, we want to distinguish between two hypotheses (each of which can be one of many alternatives).
    \emph{(c)} In the partially-revealed many-versus-one distinguishing task, some information about the alternatives is revealed, which makes the distinguishing task easier. }
    \label{fig:many-vs-one}
\end{figure}

In some tasks, it can be challenging to verify if an algorithm has learned accurately without revealing some information to the learning algorithm.
Here, we consider a setting where after learning, the algorithm is additionally given some partial information about the underlying state/process.

The set $\cX$ of all admissible states/processes can be represented as follows,
\begin{equation}
    X = (\xi, \chi) \in \cX,
\end{equation}
where $\xi$ is the information that will be revealed during prediction and $\chi$ remains hidden.
After performing all experiments, the algorithm can obtain $\xi^*$, such that the unknown state/process $X$ is guaranteed to be either
\begin{equation}
    X_0 \quad \mbox{or} \quad (\xi^*, \chi) \in \cX \setminus \{X_0\},
\end{equation}
i.e. the null hypothesis or an element of the alternative hypothesis.
In Supp. Fig.~\ref{fig:many-vs-one}, we visualize the difference between the \emph{partially-revealed many-versus-one distinguishing task} and other tasks.
Due to the additional information revealed to the learning algorithm, the distinguishing task becomes easier.
However, in many examples, we show that revealing a significant amount of information to a learning algorithm that only has external classical memory will not significantly help its distinguishing power.

Suppose that after revealing the information $\xi^*$ to the learning algorithm, the conditional probability for whether the unknown state/process $X$ is $X_0$ (null hypothesis) or one of $(\xi^*, \chi) \in \cX \setminus \{X_0\}$ (alternative hypothesis) is still uniform, i.e., $1/2$ and $1/2$.
Then similar to when the information is not revealed, the success probability of any procedure to distinguish between null and alternative hypothesis is upper bounded by
\begin{equation}
    \frac{1}{2} + \frac{1}{2} \mathrm{TV}\left(p_{X_0}, \E_{(\xi^*, \chi) \in \cX \setminus \{X_0\}} p_{(\xi^*, \chi)} \right) = \frac{1}{2} + \frac{1}{4} \sum_{\ell} \left| p_{X_0}(\ell) - \E_{(\xi^*, \chi) \in \cX \setminus \{X_0\}} p_{(\xi^*, \chi)}(\ell) \right|.
\end{equation}
When $\xi^*$ is chosen randomly, the average success probability is upper bounded by
\begin{equation}
    \frac{1}{2} + \frac{1}{2} \E_{\xi^*} \mathrm{TV}\left(p_{X_0}, \E_{(\xi^*, \chi) \in \cX \setminus \{X_0\}} p_{(\xi^*, \chi)} \right).
\end{equation}
Again, similar to the discussion before, in order to achieve a prediction accuracy of $p \geq 1/2$, we need to satisfy the inequality
\begin{equation} \label{eq:TVbound-rev}
    \E_{\xi^*} \mathrm{TV}\left(p_{X_0}, \E_{(\xi^*, \chi) \in \cX \setminus \{X_0\}} p_{(\xi^*, \chi)} \right) \geq 2p - 1.
\end{equation}
The left-hand side of the above inequality can be upper bounded by a monotonically increasing function of $T$, hence we can obtain a lower bound on $T$.
Note that by Jensen's inequality, the left hand side of the above inequality is larger than the left hand side in Eq.~\eqref{eq:TVbound}, so we will obtain a weaker lower bound on $T$.
This makes sense because making accurate predictions with partially-revealed information is easier.

\subsection{Presence of noise}
\label{sec:noisedegradation}

So far, we have considered protocols for learning quantum states or quantum processes in the absence of noise.  There are several forms of noise we can consider: (i) noise on input states, (ii) noise on the POVMs which are measured, and (iii) noise on the quantum process (if there is one).  Let us prove the following result:
\begin{theorem}[Noise cannot decrease the lower bound]
If the upper bound
\begin{equation}
\label{E:noisyupper1}
    \mathrm{TV}\left(\E_{X \in \cA} p_{X}, \E_{X \in \cB} p_{X} \right) \leq 2p -1
\end{equation}
holds for all learning protocols with a classical memory, then this same bound holds for all learning protocols with a classical memory in the presence of noise. Because the upper bound on total variation distance applies when noise is present, so does the lower bound on the number of experiments needed to achieve the distinguishing task. 
\label{thm:noisylowerbound}
\end{theorem}
\begin{proof}
Consider first the setting of learning an unknown physical state $\rho$.  Suppose we have a learning protocol with a classical memory described by a learning tree $\mathcal{T}$.  At node $u$ in the protocol, we measure the state $\rho$ with the POVM $\{F_s^u\}_s$.  We will measure the $s$th outcome with probability
\begin{align}
\label{E:probus1}
\text{tr}(F_s^u \rho)\,.
\end{align}
If there is noise on $\rho$, we can use $\rho \mapsto \mathcal{N}[\rho]$ for some noise quantum process $\mathcal{N}$.  Likewise if there is noise on the POVM, we can use $F_s^u \mapsto \mathcal{M}^\dagger[F_s^u]$ for a noise quantum process $\mathcal{M}$.  Then the probability of the $s$th outcome is instead
\begin{equation}
\text{tr}(\mathcal{M}^\dagger[F_s^u] \,\mathcal{N}[\rho]) = \text{tr}((\mathcal{N}^\dagger\circ\mathcal{M}^\dagger)[F_s^u]\, \rho)\,. 
\end{equation}
But $\{(\mathcal{N}^\dagger\circ\mathcal{M}^\dagger)[F_s^u]\}_s$ also forms a POVM.  We can apply this same argument to each node in the tree; note that the noise channels can be node-dependent.  The result is that we simply get a new learning tree with classical memory, with POVM's augmented by the noise channels.  But since by hypothesis~\eqref{E:noisyupper1} holds for all learning protocols with a classical memory, the bound evidently still holds in the noisy setting.

In the setting where we are learning a physical process $\mathcal{E}$, the argument is similar.  Given a learning tree $\mathcal{T}$ for learning the physical process, at node $u$ we (i) prepare the state $\rho^u$, (ii) apply the physical process $\mathcal{E}$, and (iii) measure with the POVM $\{F_s^u\}_s$ and obtain outcome $s$ with probability
\begin{equation}
\text{tr}(F_s^u \, \mathcal{E}[\rho^u])\,.
\end{equation}
If the initial state is noisy, we can implement this by a channel mapping $\rho^u \mapsto \mathcal{N}[\rho^u]$.  If $\mathcal{E}$ is noisy, this can be implemented by $\mathcal{E} \mapsto \mathcal{D}\circ \mathcal{E} \circ \mathcal{F}$.  Finally, if the POVM is noisy, we can implement this by $F_s^u \mapsto \mathcal{M}^\dagger[F_s^u]$.  In these circumsntances, the probability of the $s$th outcome is instead
\begin{equation}
\text{tr}(\mathcal{M}^\dagger[F_s^u] \, (\mathcal{D}\circ\mathcal{E}\circ\mathcal{F})[\mathcal{N}[\rho^u]]) = \text{tr}((\mathcal{D}^\dagger \circ M^\dagger)[F_s^u] \, \mathcal{E}[\mathcal{F} \circ \mathcal{N}[\rho^u]])\,.
\end{equation}
But $\{(\mathcal{D}^\dagger \circ M^\dagger)[F_s^u]\}_s$ also forms a POVM, and $(\mathcal{F} \circ \mathcal{N})[\rho^u]$ is also a valid choice of input state.  The same argument can be used for noise channels applied at every node in the tree, and moreover the noise can be node-dependent.  The result is that we just get a modified learning tree with classical memory, which by assumption satisfies~\eqref{E:noisyupper1}, as desired.
\end{proof}
In summary, we have shown that if a task is hard for all learning protocols with classical memory, then the task is still just as hard (if not harder) in the presence of noise.

\subsection{Related works}
\label{sec:related-work}

We begin with existing works that study the separation between conventional and quantum-enhanced strategies for learning physical systems and dynamics.
In \cite{bubeck2020entanglement}, they establish a polynomial separation between conventional and quantum-enhanced strategies for testing if a state is maximally mixed or not.
In \cite{huang2021information, aharonov2021quantum, chen2021exponential}, exponential separations between conventional and quantum-enhanced strategies are established for tasks regarding the learning of physical systems and dynamics.

We also mention some relevant works that study other classes of strategies for learning or characterizing physical systems and dynamics. It was shown in Ref. \cite{mohseni_dcqd_2007} that the dynamics of open quantum systems with dimension $d^{n}$, where $d$ is a prime, can be fully reconstructed with a quadratically fewer experiments over conventional quantum process tomography, with a quantum-enhanced strategy consist of $n$ auxiliary systems of same dimensions $d$ and performing generalized Bell-sate preparations and generalized Bell-state measurements. The results in \cite{haah2017sample} give a polynomial separation between a restricted class of conventional strategies and quantum-enhanced strategies for learning the complete description of a quantum state.
The results in \cite{huang2020predicting} give an exponential separation between a restricted class of conventional strategies and quantum-enhanced strategies for learning to predict properties of a quantum state.
\cite{chen2021quantum} establish an exponential separation between ancilla-free strategies and ancilla-assisted strategies for learning the eigenvalues in Pauli channels.
\cite{chen2021hierarchy} give an exponential separation between restricted quantum-enhanced strategies and quantum-enhanced strategies for learning about a quantum state.
\cite{anshu2021distributed} consider a problem on learning two spatially separated quantum states using local quantum learning algorithms and give an exponential separation between having a quantum or a classical communication channel between the local quantum learning algorithms.
In \cite{coudron2020computations, chia2020need}, an exponential separation between two bounded-depth quantum learning algorithms are given for learning about an exponential-time quantum process.

\section{Predicting highly-incompatible observables}
\label{sec:shadow}

The first task we study using the framework of the previous section involves learning about a physical system represented by an $n$-qubit state $\rho$. We provide an illustration of the task in Supp.~Fig.~\ref{fig:learn-obs}.
\begin{itemize}
    \item In conventional experiments, we consider algorithms that can measure each copy of $\rho$ one at a time. The algorithm can choose to perform any POVM measurement on each copy, where the POVM measurement can be chosen adaptively based on the outcomes of previous experiments.
    \item In quantum-enhanced experiments, we consider algorithms that can use a quantum computer to act collectively on multiple copies of $\rho$ to obtain entangled measurement data.
\end{itemize}
In both scenarios, we consider all quantum data to be used during the learning phase, and we are left only with classical measurement data.
After this learning phase, the learner is then asked to provide accurate predictions for the expectation value of an observable $O$, using the classical data obtained from the experiments.
The observable $O$ is selected from an exponentially large set $\{O_1, O_2, \dots O_M\}$, where $O_1, \ldots, O_M$ may not be mutually commuting and $M$ is exponential in $n$.

Note that when the observables in the set are not mutually commuting, it is impossible to measure all of them simultaneously.
Hence, a na\"{i}ve algorithm in the conventional scenario would be to measure the exponential number of observables individually, which would result in exponential sample complexity.

\begin{figure}[t]
    \centering
    \includegraphics[width=0.95\textwidth]{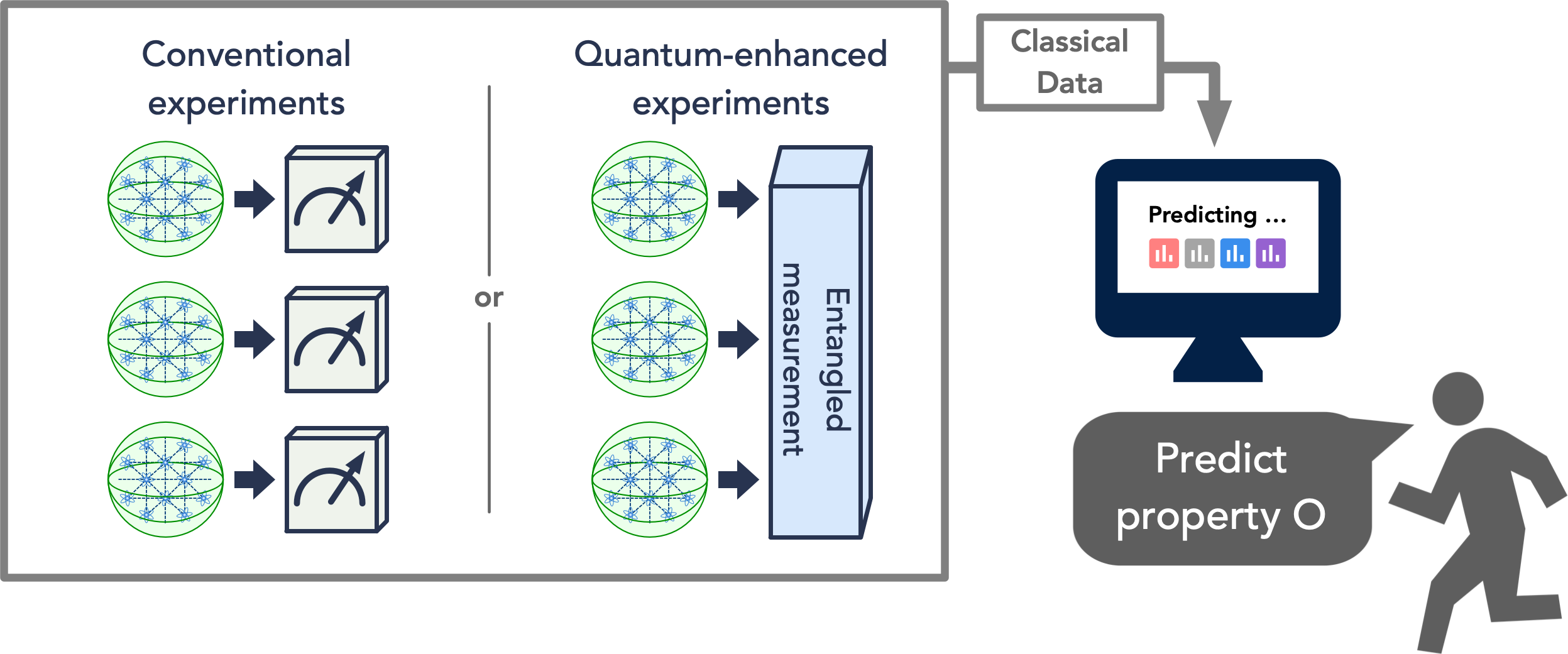}
    \caption{\emph{Illustration for the task of predicting highly-incompatible observables.}
    The unknown quantum state $\rho$ is represented by the green sphere. Conventional experiments measure each copy of state $\rho$ individually, and the measurements can depend adaptively on previous measurements. Quantum-enhanced experiments store many copies of $\rho$ in a quantum memory, process the copies with a quantum computer, and produce an entangled measurement outcome. The classical data obtained from the experiments are used to predict a property of $\rho$. }
    \label{fig:learn-obs}
\end{figure}

\subsection{Exponential advantage in predicting absolute value of a single observable}
\label{sec:expo-adv-abs-one}

We will prove that even predicting the absolute value of just a single observable requires exponentially many copies in the conventional scenario.
In contrast, an algorithm with quantum memory can predict the expectation values for $M$ arbitrary observables from only $\mathcal{O}(n \log(M) / \epsilon^4)$ copies of $\rho$ through the procedure known as shadow tomography \cite{aaronson2019shadow, aaronson2019gentle, buadescu2020improved}.
Hence, even if we would like to predict an exponential number of observables, an algorithm with quantum memory only needs a polynomial number of copies.

In fact, for certain natural instances, we can show an even more dramatic separation. Specifically, for the following states and observables, we will show how to achieve an exponential versus \emph{constant} separation.
\begin{definition}[Separation instance] \label{def:hardrhoO}
Consider a distribution $\cD$ over $n$-qubit state $\rho$ and observable $O$.
\begin{enumerate}
    \item With probability $1/2$, the state is $\rho = I / 2^n$ and $O \in \{I, X, Y, Z\}^{\otimes n} \setminus \{I^{\otimes n}\}$ is chosen uniformly at random.
    \item With probability $1/2$, the state is $\rho = (I+ 0.9 sP) / 2^n$ and $O = P$, where $s = \{\pm 1\}$ with equal probability and $P \in \{I, X, Y, Z\}^{\otimes n} \setminus \{I^{\otimes n}\}$ is chosen uniformly at random.\footnote{While $0.9$ is used in the definition $\rho = (I+ 0.9 sP) / 2^n$, any constant value smaller than $1$ is sufficient to obtain the exponential separation.
    A technical difficulty arises when we consider $(I + sP) / 2^n$, and it is unclear whether this difficulty is fundamental.}
\end{enumerate}
\end{definition}
The $n$-qubit state $\rho$ considered in the above definition does not contain any quantum entanglement.
The state $\rho$ can be written as a classical probability distribution over tensor products of single-qubit states.
Despite the lack of quantum entanglement, we can still achieve an exponential versus constant separation.
This result is a substantial improvement over the result established in \cite{huang2021information}.
In \cite{huang2021information}, some of the authors showed an $\Omega(2^{n/3})$ versus $\mathcal{O}(n)$ separation between conventional and quantum-enhanced strategies, but the task was to predict an exponential number of observables, which can only be verified using an exponential amount of time.

\begin{theorem}[Exponential advantage in predicting highly-incompatible observables] \label{thm:obs-adv}
We sample an $n$-qubit state $\rho$ and an observable $O$ according to $\cD$ given in Definition~\ref{def:hardrhoO}; both of these are unknown to the algorithm.
The algorithm then learns about $\rho$ through conventional or quantum-enhanced experiments.
After the learning phase, we ask the learning algorithm to predict $\left|\Tr(O \rho)\right|$.
\begin{itemize}
    \item \emph{Upper bound:} There is an algorithm in the \textbf{quantum-enhanced} scenario using only $\mathcal{O}(1)$ copies of $\rho$ to predict up to $0.25$ additive error with probability at least $0.8$.
    \item \emph{Lower bound:} For any algorithm in the \textbf{conventional} scenario, it needs at least $\Omega(2^n)$ copies of $\rho$ to predict up to $0.25$ additive error with probability at least $0.8$.
\end{itemize}
\end{theorem}

Here, we are using the standard Big-$\mathcal{O}$ and Big-$\Omega$ notations: $f = \Omega(g)$ if there is an $n_0, C > 0$ such that $\forall n > n_0$, $f(n) \geq C g(n)$; and $f = \mathcal{O}(g)$ if there is an $n_0, M > 0$ such that $\forall n > n_0$, $|f(n)| \leq M g(n)$.
We separate the proof of Theorem~\ref{thm:obs-adv} into the following two subsections.
In Appendix~\ref{sec:uppbd-obs}, we prove a constant upper bound for quantum-enhanced experiments for this task.
In Appendix~\ref{sec:lowbd-obs}, we prove an exponential lower bound for conventional experiments for the same task.

\subsection{A constant upper bound for quantum-enhanced experiments}
\label{sec:uppbd-obs}

The learning algorithm in the quantum-enhanced scenario builds on results presented in \cite{huang2021information}.
We separate the protocol into the learning phase, where entangled measurements are performed, and the prediction phase, where we predict the desired properties.

\subsubsection{Learning phase}

Consider $N_{\mathrm{Q}}$ rounds of two-copy entangled measurements.
In round $t \in \{1, \ldots, N_{\mathrm{Q}}\}$, for every $k\in\{1, \ldots, n\}$ we measure the $k$-th qubit from the first and second copies of $\rho$ in the Bell basis to obtain
\begin{equation}\label{eq:Bell-basis-meas}
    S^{(t)}_k \in \Big\{ \ketbra{\Psi^+}{\Psi^+}, \ketbra{\Psi^-}{\Psi^-}, \ketbra{\Phi^+}{\Phi^+}, \ketbra{\Phi^-}{\Phi^-} \Big\},
\end{equation}
where the Bell basis encompasses four maximally entangled two-qubit states. Here, $\ket{\Omega} = \tfrac{1}{\sqrt{2}} \left( \ket{00}+\ket{11} \right)$ is the Bell state, and we additionally have
\begin{align*}
    \ket{\Psi^+} &= I \otimes I \ket{\Omega}
    =\frac{1}{\sqrt{2}} \left( \ket{00} + \ket{11} \right), &\quad \quad \ket{\Psi^-} &=   I \otimes Z \ket{\Omega}
    = \frac{1}{\sqrt{2}} \left( \ket{00} - \ket{11} \right),\\
   \ket{\Phi^+} &= I \otimes X \ket{\Omega} 
    = \frac{1}{\sqrt{2}} \left( \ket{01} + \ket{10} \right), & \quad \quad \ket{\Phi^-} &= \mathrm{i} I \otimes Y \ket{\Omega}
    =\frac{1}{\sqrt{2}} \left( \ket{01} - \ket{10} \right).
\end{align*}
Then, we efficiently store the measurement data $S^{(t)}_k, \forall k = 1, \ldots, n, \forall t = 1, \ldots, N_{\mathrm{Q}}$ in a classical memory with $2 n N_{\mathrm{Q}}$ classical bits.

\subsubsection{Prediction phase}

Given an observable $O$ drawn from $\{I, X, Y, Z\}^{\otimes n} \setminus \{I^{\otimes n}\}$, we can use the block of classical memory obtained in the learning phase to estimate $|\Tr(O \rho)|$.
First let us consider the case where $\rho$ is a single-qubit state.
When we measure $\rho\otimes \rho$ in the Bell basis, the measurement outcome $S$ is a projector onto one of the four Bell states given in Eq.~\eqref{eq:Bell-basis-meas}.
Let $\sigma\in \{I,X,Y,Z\}$ be any Pauli matrix.
Each Bell state is an eigenstate of $\sigma\otimes\sigma$ with an eigenvalue $\pm 1$.
The probability that the Bell measurement outcome $S$ is an eigenstate of $\sigma\otimes\sigma$ with eigenvalue $+1$ is
\begin{equation}
    \mathrm{Prob}(+) = \frac{1}{2}\Tr\left((I\otimes I+\sigma\otimes\sigma)(\rho\otimes\rho)\right),
\end{equation}
while the $-1$ eigenvalue  occurs with probability
\begin{equation}
 \mathrm{Prob}(-) = \frac{1}{2}\Tr\left((I\otimes I-\sigma\otimes\sigma)(\rho\otimes\rho)\right).   
\end{equation}
Therefore, we have
\begin{equation}
    \E\left[\Tr\left((\sigma \otimes \sigma) S\right)\right] 
    = \mathrm{Prob}(+) - \mathrm{Prob}(-) = \Tr \left((\sigma\otimes\sigma)(\rho\otimes\rho)\right) = |\Tr (\sigma\rho)|^2,
\end{equation}
where $\E$ denotes the expectation with respect to the probability distribution over Bell measurement outcomes.
We see that the entangling Bell measurement enables us to estimate the absolute value $|\Tr (\sigma\rho)|$ for any Pauli matrix $\sigma \in \left\{I,X,Y,Z\right\}$.

We can generalize this observation to the case where $\rho$ is an $n$-qubit state, and each pair of qubits in $\rho\otimes \rho$ is measured in the Bell basis to yield the outcomes $\{S_k, k = 1,2, \dots, n\}$.
If $O = \sigma_{1} \otimes \cdots \otimes \sigma_{n}$ is a Pauli observable, then as in the $n=1$ case the Bell state $S_k$ is an eigenstate of $\sigma_k\otimes \sigma_k$ with eigenvalue $\pm 1$ for each $k$.
This implies that $\bigotimes_{k=1}^n S_k$ is an eigenstate of $O\otimes O$ with an eigenvalue $\pm 1$.
In particular, let us consider the product
\begin{equation}
     \prod_{k=1}^n \Tr\left((\sigma_{k} \otimes \sigma_{k}) S_k\right)
     =\pm 1.
\end{equation}
This product is equal to $+1$ when the tensor product of the Bell measurement outcomes $\bigotimes_{k=1}^n S_k$ is an eigenstate of $O\otimes O$ with eigenvalue $+1$, and it is $-1$ when $\otimes_{k=1}^n S_k$ is an eigenstate of $O\otimes O$ with eigenvalue $-1$.
We conclude that
\begin{align}\label{eq:bell-expected}
     \E\left[\prod_{k=1}^n \Tr\left((\sigma_{k} \otimes \sigma_{k}) S_k\right)\right] 
     &=\E\left[ \Tr\left( (O \otimes O) \bigotimes_{k=1}^n S_k \right) \right]\nonumber\\
    & = \mathrm{Prob}(O\otimes O =+1) - \mathrm{Prob}(O\otimes O=-1)
      \nonumber \\
      &=\Tr\left( (O \otimes O) (\rho \otimes \rho) \right) \nonumber \\
     &= |\Tr(O \rho)|^2,
\end{align}
where $\E$ denotes the expectation with respect to the probability distribution of Bell measurement outcomes.
The above derivation shows that the $n$-qubit entangling Bell measurement enables us to estimate the absolute value $|\Tr (O \rho)|$ for any $O$ considered in Definition~\ref{def:hardrhoO}.

Because Equation~\eqref{eq:bell-expected} relates the probability distribution of Bell measurement outcomes to the absolute value $|\Tr(O \rho)|$, we can estimate $|\Tr(O \rho)|$ accurately by repeatedly making entangling Bell measurements on successive pairs of copies of $\rho$ sufficiently many times.
Specifically, in the learning phase, we perform the entangling Bell measurement on $N_{\mathrm{Q}}$ pairs of copies of $\rho$, and collect the measurement data $\{S_k^{(t)}\}$ in the classical memory, where $k=1,2,\dots, n$ labels the qubit pairs, and $t=1,2, \dots, N_{\mathrm{Q}}$ labels the different rounds of measurements. 
For any given $n$-qubit Pauli observable $O = \sigma_{1} \otimes \cdots \otimes \sigma_{n}$, we consider the following estimator
\begin{equation}
    \hat{a}(O) = \frac{1}{N_Q} \sum_{t = 1}^{N_Q} \prod_{k=1}^n \Tr\left((\sigma_{k} \otimes \sigma_{k}) S^{(t)}_k\right),
\end{equation}
which can be computed efficiently in time $\mathcal{O}(nN_{\mathrm{Q}})$.

Using the expectation evaluated in Equation~(\ref{eq:bell-expected}), we can apply Hoeffding's inequality to
show that the estimate $\hat{a}(O)$ is equal to the expectation value $\Tr((O \otimes O)(\rho \otimes \rho)) = |\Tr(O \rho)|^2$ up to a small error with high probability.
The formal statement is given below.
\begin{lemma}\label{lem:hoeffPauli}
Given $N_{\mathrm{Q}} = \Theta(\log(1 / \delta) / \epsilon^2)$. For any observable $O$ considered in Definition~\ref{def:hardrhoO}, we have
\begin{equation}
    \left|\hat{a}(O) - |\Tr(O \rho)|^2\right| \leq \epsilon,
\end{equation}
with probability at least $1 - \delta$.
\end{lemma}
\noindent To obtain an estimate for the absolute value $|\Tr(O \rho)|$, we consider the estimate
\begin{equation}
    \hat{b} = \sqrt{\max(0, \hat{a})}.
\end{equation}
We can show the inequalities
\begin{align}
    |\Tr(O \rho)|^2 - \epsilon \leq \hat{a} \leq |\Tr(O \rho)|^2 + \epsilon
    \implies \max(0, \sqrt{|\Tr(O \rho)|^2} - \sqrt{\epsilon}) \leq \hat{b} \leq \sqrt{|\Tr(O \rho)|^2} + \sqrt{\epsilon},
\end{align}
using the fact that $\sqrt{x + y} \leq \sqrt{x} + \sqrt{y}$.

In the final step of the upper bound proof, we use Lemma~\ref{lem:hoeffPauli} to obtain the following result.
As long as $N_{\mathrm{Q}} = \mathcal{O}(1)$, we can estimate the absolute value of $\Tr(O \rho)$ for any observable $O$ given in Definition~\ref{def:hardrhoO} to an error $0.25$ with probability at least $0.8$.
\begin{corollary} \label{cor:absPauli}
Let $N_{\mathrm{Q}} = \Theta(1)$. For any observable $O$ considered in Definition~\ref{def:hardrhoO}, we have
\begin{equation}
    \left|\hat{b} - |\Tr(O \rho)|\right| \leq 0.25,
\end{equation}
with probability at least $0.8$.
\end{corollary}
\noindent This concludes the constant upper bound for quantum-enhanced experiments in Theorem~\ref{thm:obs-adv}.

\subsection{An exponential lower bound for conventional experiments}
\label{sec:lowbd-obs}

The proof begins with a reduction to the partially-revealed many-versus-one distinguishing task followed by bounding the total variation distance.

\subsubsection{Reduction to partially-revealed many-versus-one distinguishing task}

We consider the following partially-revealed many-versus-one distinguishing task discussed in Appendix~\ref{sec:part-many-v-one}, namely where:
\begin{itemize}
    \item The null hypothesis is $I / 2^n$.
    \item The alternative hypothesis is $(I + 0.9 s P) / 2^n$.
\end{itemize}
The partially revealed information is the Pauli operator $P$.
Recall the following from Definition~\ref{def:hardrhoO},
\begin{enumerate}
    \item With probability $1/2$, the state is $\rho = I / 2^n$ and $O \in \{I, X, Y, Z\}^{\otimes n} \setminus \{I^{\otimes n}\}$ is sampled uniformly at random. (Corresponds to the null hypothesis)
    \item With probability $1/2$, the state is $\rho = (I+ 0.9 sP) / 2^n$ and $O = P$, where $s = \{\pm 1\}$ with equal probability and $P \in \{I, X, Y, Z\}^{\otimes n} \setminus \{I^{\otimes n}\}$ uniformly. (Corresponds to the alternative hypothesis)
\end{enumerate}
For $\rho = I / 2^n$, we have $|\Tr(O \rho)| = 0$.
For $\rho = (I+ 0.9 sP) / 2^n$, we have $|\Tr(O \rho)| = 0.9$.
Therefore, if an algorithm could predict $|\Tr(O \rho)|$ to $0.25$ error with probability at least $1 - \delta$, it could be used to distinguish between the null and alternative hypotheses with success probability at least $1 - \delta$.

\subsubsection{Total variation distance}

From the information-theoretic lower bound for partially-revealed many-versus-one distinguishing task given in Appendix~\ref{sec:part-many-v-one}, if we let $p_{\rho}(\ell)$ be the leaf probability distribution under $\rho$, then
\begin{equation} \label{eq:TVLB}
    \E_{P \in \{I, X, Y, Z\}^{\otimes n } \setminus \{I^{\otimes n}\}} \mathrm{TV}\left(p_{I / 2^n}, \E_{s \in \{\pm 1\}} p_{(I+ 0.9 sP) / 2^n} \right) \geq 1 - 2 \delta.
\end{equation}
For each leaf node $\ell$, we consider the path from the root to $\ell$,
\begin{equation}
    u_0 = r \xrightarrow{s_1} u_1 \xrightarrow{s_2} u_2 \xrightarrow{s_3} \ldots \xrightarrow{s_{T-1}} u_{T-1} \xrightarrow{s_T} u_T = \ell.
\end{equation}
At each node $u$, we perform a POVM measurement $\{w^u_s \ketbra{\phi^u_s}{\phi^u_s }\}_s$ on $\rho$  to obtain an outcome $s$ with probability
\begin{equation}
     w^u_s \bra{\phi^u_s} \rho \ket{\phi^u_s}.
\end{equation}
Hence, we can write down the probability to arrive at the leaf $\ell$ as
\begin{equation}
    p_{\rho}(\ell) = \prod_{t=1}^{T} w^{u_{t-1}}_{s_t} \bra{\phi^{u_{t-1}}_{s_t}} \rho \ket{\phi^{u_{t-1}}_{s_t}}.
\end{equation}
Recalling the definition of total variation distance, note that for any probability distributions $p_A,p_B$ for which $p_A(\ell) > 0$ whenever $p_B(\ell) > 0$,
\begin{equation}
    \mathrm{TV}(p_A, p_B) = \frac{1}{2}\sum_{\ell} |p_A(\ell) - p_B(\ell)| = \sum_{\ell}  \max(0, p_A(\ell) - p_B(\ell)) = \sum_{\ell} p_A(\ell)\cdot \max\left(0,1 - \frac{p_B(\ell)}{p_A(\ell)}\right), \label{eq:tv_fdiv}
\end{equation}
where the last equality follows from $\max(ax, ay) = a \max(x, y)$ for all $x, y \in \mathbb{R}$ and all $a \geq 0$.

Observe that for leaf $\ell$,
\begin{align}
    \frac{\E_{s\in\{\pm 1\}} p_{(I+0.9sP)/2^n}(\ell)}{p_{I/2^n}(\ell)} &= \frac{\E_{s \in \{\pm 1\}} \prod_{t=1}^{T} w^{u_{t-1}}_{s_t} \bra{\phi^{u_{t-1}}_{s_t}} \frac{I+ 0.9 sP}{2^n} \ket{\phi^{u_{t-1}}_{s_t}}}{\prod_{t=1}^{T} w^{u_{t-1}}_{s_t} \bra{\phi^{u_{t-1}}_{s_t}} \frac{I}{2^n} \ket{\phi^{u_{t-1}}_{s_t}}} \\
    &=  \E_{s\in\{\pm 1\}} \prod_{t=1}^{T} \left(1 + 0.9 s \bra{\phi^{u_{t-1}}_{s_t}} P \ket{\phi^{u_{t-1}}_{s_t}} \right). \label{eq:lr}
\end{align}
Combining \eqref{eq:tv_fdiv} and \eqref{eq:lr}, we can express the total variation distance inside the expectation in \eqref{eq:TVLB} as
\begin{equation}
    \mathrm{TV}\left(p_{I / 2^n}, \E_{s \in \{\pm 1\}} p_{(I+ 0.9 sP) / 2^n} \right)
    = \sum_{\ell} p_{I / 2^n}(\ell) \, \max\left( 0, \,\, 1 - \E_{s \in \{\pm 1\}} \prod_{t=1}^{T} \left(1 + 0.9 s \bra{\phi^{u_{t-1}}_{s_t}} P \ket{\phi^{u_{t-1}}_{s_t}} \right) \right)
\end{equation}

\subsubsection{Upper bound for total variation distance}

We analyze one of the terms in the total variation distance using Jensen's inequality (note that $\exp(x)$ is a convex function in $x$).
\begin{align}
    &\E_{s \in \{\pm 1\}} \prod_{t=1}^{T} \left(1 + 0.9 s \bra{\phi^{u_{t-1}}_{s_t}} P \ket{\phi^{u_{t-1}}_{s_t}} \right)\\
    &= \E_{s \in \{\pm 1\}}  \exp\left[ \sum_{t=1}^T \log\left(1 + 0.9 s \bra{\phi^{u_{t-1}}_{s_t}} P \ket{\phi^{u_{t-1}}_{s_t}} \right) \right]\\
    &\geq  \exp\left[ \E_{s \in \{\pm 1\}} \sum_{t=1}^T \log\left(1 + 0.9 s \bra{\phi^{u_{t-1}}_{s_t}} P \ket{\phi^{u_{t-1}}_{s_t}} \right) \right]\\
    &=  \exp\left[ \sum_{t=1}^T \frac{1}{2} \log\left(1 - 0.81 \bra{\phi^{u_{t-1}}_{s_t}} P \ket{\phi^{u_{t-1}}_{s_t}}^2 \right) \right]\\
    &=  \prod_{t=1}^{T} \sqrt{1 - 0.81 \bra{\phi^{u_{t-1}}_{s_t}} P \ket{\phi^{u_{t-1}}_{s_t}}^2}\,.
\end{align}
We can then upper bound the total variation distance as
\begin{align}
    &\mathrm{TV}\left(p_{I / 2^n}, \E_{s \in \{\pm 1\}} p_{(I+ 0.9 sP) / 2^n} \right)\\
    &\leq \sum_{\ell} p_{I / 2^n}(\ell) \, \max\Bigg( 0, \,\, 1 - \prod_{t=1}^{T} \sqrt{1 - 0.81 \bra{\phi^{u_{t-1}}_{s_t}} P \ket{\phi^{u_{t-1}}_{s_t}}^2} \Bigg)\\
    &= \sum_{\ell} p_{I / 2^n}(\ell) \, \Bigg( 1 - \prod_{t=1}^{T} \sqrt{1 - 0.81 \bra{\phi^{u_{t-1}}_{s_t}} P \ket{\phi^{u_{t-1}}_{s_t}}^2} \Bigg)\,. \label{eq:TVUB}
\end{align}
The last equality follows from the fact that all eigenvalues of $P$ are $\pm 1$, hence $1 \geq \prod_{t=1}^{T} \sqrt{1 - 0.81 \bra{\phi^{u_{t-1}}_{s_t}} P \ket{\phi^{u_{t-1}}_{s_t}}^2}$.

\subsubsection{Lower bound for the number of measurements}

We can combine Eq.~\eqref{eq:TVUB} and Eq.~\eqref{eq:TVLB} to find
\begin{align}
\E_{P \in \{I, X, Y, Z\}^{\otimes n } \setminus \{I^{\otimes n}\}} \sum_{\ell} p_{I / 2^n}(\ell) \, \Bigg( 1 - \prod_{t=1}^{T} \sqrt{1 - 0.81 \bra{\phi^{u_{t-1}}_{s_t}} P \ket{\phi^{u_{t-1}}_{s_t}}^2} \Bigg) \geq 1 - 2 \delta.
\end{align}
By linearity of expectation, we have
\begin{align} \label{eq:TV-UBLB}
\sum_{\ell} p_{I / 2^n}(\ell) \, \Bigg( 1 - \E_{P \in \{I, X, Y, Z\}^{\otimes n } \setminus \{I^{\otimes n}\}} \prod_{t=1}^{T} \sqrt{1 - 0.81 \bra{\phi^{u_{t-1}}_{s_t}} P \ket{\phi^{u_{t-1}}_{s_t}}^2} \Bigg) \geq 1 - 2 \delta.
\end{align}
We analyze the expectation value term in the summand as follows:
\begin{align}
    &\E_{P \in \{I, X, Y, Z\}^{\otimes n } \setminus \{I^{\otimes n}\}} \prod_{t=1}^{T} \sqrt{1 - 0.81 \bra{\phi^{u_{t-1}}_{s_t}} P \ket{\phi^{u_{t-1}}_{s_t}}^2} \label{eq:sqrt0.81}\\
    &= \E_{P \in \{I, X, Y, Z\}^{\otimes n } \setminus \{I^{\otimes n}\}}
    \exp\left[ \frac{1}{2} \sum_{t=1}^T \log\left(1 - 0.81 \bra{\phi^{u_{t-1}}_{s_t}} P \ket{\phi^{u_{t-1}}_{s_t}}^2\right) \right] \\
    &\geq
    \exp\left[ \frac{1}{2} \sum_{t=1}^T \E_{P \in \{I, X, Y, Z\}^{\otimes n } \setminus \{I^{\otimes n}\}} \log\left(1 - 0.81 \bra{\phi^{u_{t-1}}_{s_t}} P \ket{\phi^{u_{t-1}}_{s_t}}^2\right) \right] \\
    &\geq \exp\left[ \frac{1}{2} \sum_{t=1}^T \E_{P \in \{I, X, Y, Z\}^{\otimes n } \setminus \{I^{\otimes n}\}} -1.701 \bra{\phi^{u_{t-1}}_{s_t}} P \ket{\phi^{u_{t-1}}_{s_t}}^2 \right] \\
    &= \exp\left[ -0.8505 \sum_{t=1}^T \frac{1}{2^n+1} \right] = \exp\left( - \frac{0.8505 T}{2^n+1} \right). \label{eq:0.8505}
\end{align}
The second line follows from Jensen's inequality because $\exp(x)$ is convex in $x$.
The third line uses $\log(1-x) \geq -2.1x, \forall x \in [0, 0.82]$.
The fourth line uses the fact that
\begin{equation}
    \E_{\substack{P \in \{I, X, Y, Z\}^{\otimes n} \setminus \{I^{\otimes n}\}}} P \otimes P = \frac{2^n \mathrm{SWAP} - I \otimes I}{4^n - 1},
\end{equation}
hence $\E_{\substack{P \in \{I, X, Y, Z\}^{\otimes n} \setminus \{I^{\otimes n}\}}} \bra{\phi^{u_{t-1}}_{s_t}} P \ket{\phi^{u_{t-1}}_{s_t}}^2 = \frac{2^n - 1}{4^n - 1} = \frac{1}{2^n + 1}$.

Combining the analysis with Eq.~\eqref{eq:TV-UBLB}, we find that
\begin{equation}
    \sum_{\ell} p_{I / 2^n}(\ell) \, \Bigg( 1 - \exp\left( - \frac{0.8505 T}{2^n+1} \right) \Bigg) \geq 1 - 2 \delta.
\end{equation}
Because $\sum_{\ell} p_{I / 2^n}(\ell) = 1$, we have
\begin{equation}
    1 - \exp\left( - \frac{0.8505 T}{2^n+1} \right) \geq 1 - 2 \delta.
\end{equation}
Together, the following lower bound on the number of experiments can be obtained:
\begin{equation}
    T \geq \frac{2^n + 1}{0.8505} \log\left(\frac{1}{2 \delta}\right).
\end{equation}
After setting $\delta = 0.2$ (corresponding to a success probability of at least $0.8$), we conclude the exponential lower bound for conventional experiments in Theorem~\ref{thm:obs-adv}.

\subsection{An exponential lower bound for comparing absolute values}
\label{sec:compare-abs-value}

In the physical experiment presented in the main text, we considered a slightly different task that also yields an exponential lower bound for conventional experiments.
This slightly different task has two main differences when compared with the task described in Appendix~\ref{sec:lowbd-obs}.
First, we do not consider the maximally mixed state $I/2^n$.
Second, we ask the learner to predict which of two observables $O_1, O_2$ has a larger absolute value.
The task description is given below.

\begin{task}[Comparing absolute values] \label{task:comp-abs}
There is an unknown state $\rho = (I+ 0.9 sP) / 2^n$ where $s = \{\pm 1\}$ and $P \in \{I, X, Y, Z\}^{\otimes n} \setminus \{I^{\otimes n}\}$ are both sampled uniformly at random.
The algorithm learns about $\rho$ through conventional or quantum-enhanced strategies.
The algorithm transforms all quantum data to classical data.
After learning, we present the learning algorithm with
\begin{equation}
O_1 = P, \,\, O_2 = Q  \,\,\quad \mathrm{or} \,\,\quad O_1 = Q, \,\, O_2 = P,
\end{equation}
with equal probability, where $Q \neq P \in \{I, X, Y, Z\}^{\otimes n} \setminus \{I^{\otimes n}\}$ is sampled uniformly.
The learning algorithm succeeds if it correctly classifies whether $|\Tr(O_1 \rho)| > |\Tr(O_2 \rho)|$ or $|\Tr(O_1 \rho)| < |\Tr(O_2 \rho)|$.
\end{task}

Using the procedure presented in Appendix~\ref{sec:uppbd-obs}, it is not hard to show that quantum-enhanced strategies could accomplish the above task with classification accuracy (i.e., the probability that the classification is correct) at least $1 - \delta$ from only $\mathcal{O}(\log(1 / \delta))$ experiments.
In contrast, we have the following theorem for conventional experiments.

\begin{theorem}[Exponential lower bound for Task~\ref{task:comp-abs}] \label{thm:comp-abs}
A learning algorithm in the conventional setting (without quantum memory) requires at least
\begin{equation}
    \frac{(2^n + 1)}{0.8505} \log\left(\frac{2}{1+2\delta}\right)
\end{equation}
experiments to accomplish Task~\ref{task:comp-abs} with an accuracy of $1 - \delta$, for a given $\delta > 0$.
\end{theorem}

\subsubsection{Lower bound for total variation distance}

Here we begin the proof of Theorem~\ref{thm:comp-abs}.
Task~\ref{task:comp-abs} is closely related to the partially-revealed many-versus-one distinguishing task, but is not exactly the same.
We will utilize a slightly different information-theoretic bound for this task.
Let us define the following notation
\begin{equation}
    \rho_{sP} \equiv \frac{I+ 0.9 sP}{2^n}.
\end{equation}
We consider a learning algorithm in the conventional setting. We consider the probability distribution $p_{\rho}(\ell)$ over the leaf node $\ell$ when the underlying state is $\rho$.
Recall that the leaf node $\ell$ is the final memory state of the learning algorithm.
Any procedure that makes the prediction based on the final memory state of the learning algorithm must have a classification accuracy upper bounded by
\begin{equation} \label{eq:PQ-TV}
    \frac{1}{(4^n - 1)(4^n - 2)} \sum_{P \neq Q \in \{I, X, Y, Z\}^{\otimes n} \setminus \{I^{\otimes n}\}} \left[ \frac{1}{2} + \frac{1}{2} \mathrm{TV}\left( \E_{s \in \{\pm 1\}} p_{\rho_{sP}}(\ell), \E_{s \in \{\pm 1\}} p_{\rho_{sQ}}(\ell) \right) \right]\,.
\end{equation}

To understand why the above inequality holds, consider a fixed $P \neq Q$.
There is an equal probability that the underlying state is $\rho_{+P}, \rho_{-P}, \rho_{+Q},$ or $\rho_{-Q}$ because $s \in \{\pm 1\}$ with equal probability and $O_1 = P, O_2 = Q$ or $O_1 = Q, O_2 = P$ with probability $1/2$.
In order to distinguish the event of $\rho_{+P}, \rho_{-P}$ from the event  $\rho_{+Q}, \rho_{-Q}$ based on the leaf node $\ell$, we need the two distributions $\E_{s \in \{\pm 1\}} p_{\rho_{sP}}(\ell)$ and $\E_{s \in \{\pm 1\}} p_\ell(\rho_{sQ})$ to be sufficiently distinct.
Formally, one can show that the success probability is upper bounded by
\begin{equation}
    \frac{1}{2} + \frac{1}{2} \mathrm{TV}\left( \E_{s \in \{\pm 1\}} p_{\rho_{sP}}(\ell), \E_{s \in \{\pm 1\}} p_{\rho_{sQ}}(\ell) \right)
\end{equation}
using LeCam's two-point method, see e.g.~Lemma 1 in~\cite{yu1997assouad}.
To achieve the above success probability, one can use the maximum likelihood protocol that outputs $P$ if $\E_{s \in \{\pm 1\}} p_{\rho_{sP}}(\ell) > \E_{s \in \{\pm 1\}} p_{\rho_{sQ}}(\ell)$ and outputs $Q$ otherwise.
Because $P, Q$ are both chosen uniformly at random (but distinct), the average classification accuracy is given in Eq.~\eqref{eq:PQ-TV}.
If the learning algorithm could achieve an accuracy of $1 - \delta$, we would have
\begin{equation}
    (1 - \delta) \leq \frac{1}{2} + \frac{1}{2(4^n - 1)(4^n - 2)} \sum_{P \neq Q \in \{I, X, Y, Z\}^{\otimes n} \setminus \{I^{\otimes n}\}} \mathrm{TV}\left( \E_{s \in \{\pm 1\}} p_{\rho_{sP}}(\ell), \E_{s \in \{\pm 1\}} p_{\rho_{sQ}}(\ell) \right).
\end{equation}
This implies that
\begin{equation} \label{eq:1-2delta-lowerbound}
    1 - 2\delta \leq \frac{1}{(4^n - 1)(4^n - 2)} \sum_{P \neq Q \in \{I, X, Y, Z\}^{\otimes n} \setminus \{I^{\otimes n}\}} \mathrm{TV}\left( \E_{s \in \{\pm 1\}} p_{\rho_{sP}}(\ell), \E_{s \in \{\pm 1\}} p_{\rho_{sQ}}(\ell) \right).
\end{equation}

\subsubsection{Upper bound for total variation distance}

We now perform triangle inequalities and reuse inequalities in Appendix~\ref{sec:lowbd-obs} to upper bound the total variation distance:
\begin{align}
    \mathrm{TV}\left( \E_{s \in \{\pm 1\}} p_{\rho_{sP}}(\ell), \E_{s \in \{\pm 1\}} p_{\rho_{sQ}}(\ell) \right) &\leq \mathrm{TV}\left(p_{I/2^n}(\ell), \E_{s \in \{\pm 1\}} p_{\rho_{sP}}(\ell) \right) + \mathrm{TV}\left(p_{I/2^n}(\ell), \E_{s \in \{\pm 1\}} p_{\rho_{sQ}}(\ell) \right)\\
    &\leq \sum_{\ell} p_{I / 2^n}(\ell) \, \Bigg( 1 - \prod_{t=1}^{T} \sqrt{1 - 0.81 \bra{\phi^{u_{t-1}}_{s_t}} P \ket{\phi^{u_{t-1}}_{s_t}}^2} \Bigg)\nonumber \\
    &\qquad+ \sum_{\ell} p_{I / 2^n}(\ell) \, \Bigg( 1 - \prod_{t=1}^{T} \sqrt{1 - 0.81 \bra{\phi^{u_{t-1}}_{s_t}} Q \ket{\phi^{u_{t-1}}_{s_t}}^2} \Bigg)\\
    &= 1 - \sum_{\ell} p_{I / 2^n}(\ell) \prod_{t=1}^{T} \sqrt{1 - 0.81 \bra{\phi^{u_{t-1}}_{s_t}} P \ket{\phi^{u_{t-1}}_{s_t}}^2} \nonumber \\
    &\qquad+ 1 - \sum_{\ell} p_{I / 2^n}(\ell) \prod_{t=1}^{T} \sqrt{1 - 0.81 \bra{\phi^{u_{t-1}}_{s_t}} Q \ket{\phi^{u_{t-1}}_{s_t}}^2}. \label{eq:sqrt0.81<Q>}
\end{align}
The first line is triangle inequality.
The second inequality uses Eq.~\eqref{eq:TVUB} for both $P$ and $Q$.
The equality step uses $\sum_{\ell} p_{I / 2^n}(\ell) = 1$.
Therefore, we have
\begin{align}
    &\frac{1}{(4^n - 1)(4^n - 2)} \sum_{P \neq Q \in \{I, X, Y, Z\}^{\otimes n} \setminus \{I^{\otimes n}\}} \mathrm{TV}\left( \E_{s \in \{\pm 1\}} p_{\rho_{sP}}(\ell), \E_{s \in \{\pm 1\}} p_{\rho_{sQ}}(\ell) \right)\\
    &\leq 1 - \frac{1}{(4^n - 1)(4^n - 2)} \sum_{P \neq Q \in \{I, X, Y, Z\}^{\otimes n} \setminus \{I^{\otimes n}\}} \left[ \sum_{\ell} p_{I / 2^n}(\ell) \prod_{t=1}^{T} \sqrt{1 - 0.81 \bra{\phi^{u_{t-1}}_{s_t}} P \ket{\phi^{u_{t-1}}_{s_t}}^2} \right] \nonumber \\
    &\qquad+ 1 - \frac{1}{(4^n - 1)(4^n - 2)} \sum_{P \neq Q \in \{I, X, Y, Z\}^{\otimes n} \setminus \{I^{\otimes n}\}} \left[ \sum_{\ell} p_{I / 2^n}(\ell) \prod_{t=1}^{T} \sqrt{1 - 0.81 \bra{\phi^{u_{t-1}}_{s_t}} Q \ket{\phi^{u_{t-1}}_{s_t}}^2} \right] \\
    &= 1 -  \sum_{\ell} p_{I / 2^n}(\ell) \E_{P \in \{I, X, Y, Z\}^{\otimes n} \setminus \{I^{\otimes n}\}} \prod_{t=1}^{T} \sqrt{1 - 0.81 \bra{\phi^{u_{t-1}}_{s_t}} P \ket{\phi^{u_{t-1}}_{s_t}}^2} \nonumber \\
    &\qquad+ 1 - \sum_{\ell} p_{I / 2^n}(\ell) \E_{Q \in \{I, X, Y, Z\}^{\otimes n} \setminus \{I^{\otimes n}\}} \prod_{t=1}^{T} \sqrt{1 - 0.81 \bra{\phi^{u_{t-1}}_{s_t}} Q \ket{\phi^{u_{t-1}}_{s_t}}^2} \\
    \label{E:third1}
    &= 1 -  \sum_{\ell} p_{I / 2^n}(\ell) \exp\left( - \frac{0.8505 T}{2^n+1} \right) + 1 - \sum_{\ell} p_{I / 2^n}(\ell) \exp\left( - \frac{0.8505 T}{2^n+1} \right)\\
    &= 2 -  2 \exp\left( - \frac{0.8505 T}{2^n+1} \right). \label{eq:2x0.8505}
\end{align}
In the above inequalities, the first inequality uses Eq.~\eqref{eq:sqrt0.81<Q>}.
The equality thereafter uses the following analysis,
\begin{align}
    &\frac{1}{(4^n - 1)(4^n - 2)} \sum_{P \neq Q \in \{I, X, Y, Z\}^{\otimes n} \setminus \{I^{\otimes n}\}} f(P) \\
    &= \frac{1}{4^n - 1} \sum_{P \in \{I, X, Y, Z\}^{\otimes n} \setminus \{I^{\otimes n}\}} f(P) \left( \frac{1}{4^n - 2} \sum_{\substack{Q \in \{I, X, Y, Z\}^{\otimes n} \setminus \{I^{\otimes n}\} \\ \mathrm{s.t.,} \, Q\neq P}} 1  \right) \\
    &= \frac{1}{4^n - 1} \sum_{P \in \{I, X, Y, Z\}^{\otimes n} \setminus \{I^{\otimes n}\}} f(P) \\
    &= \E_{P \in \{I, X, Y, Z\}^{\otimes n} \setminus \{I^{\otimes n}\}} f(P)\,,
\end{align}
where $f(P) = \sum_{\ell} p_{I / 2^n}(\ell) \prod_{t=1}^{T} \sqrt{1 - 0.81 \bra{\phi^{u_{t-1}}_{s_t}} Q \ket{\phi^{u_{t-1}}_{s_t}}^2}$, as well as linearity of expectation, i.e.
\begin{equation}
    \E_{P \in \{I, X, Y, Z\}^{\otimes n} \setminus \{I^{\otimes n}\}} \sum_{\ell} p_{I / 2^n}(\ell) = \sum_{\ell} p_{I / 2^n}(\ell) \E_{P \in \{I, X, Y, Z\}^{\otimes n} \setminus \{I^{\otimes n}\}}.
\end{equation}
The third step in Eq.~\eqref{E:third1} uses Eq.~\eqref{eq:sqrt0.81}~to~\eqref{eq:0.8505}, and the final step uses $\sum_{\ell} p_{I / 2^n}(\ell) = 1$.

\subsubsection{Combining upper and lower bounds}

We can combine the lower bound obtained in Eq.~\eqref{eq:1-2delta-lowerbound} and the upper bound in Eq.~\eqref{eq:2x0.8505} to find
\begin{equation}
    1 - 2\delta \leq 2 -  2 \exp\left( - \frac{0.8505 T}{2^n+1} \right).
\end{equation}
Basic algebraic manipulations give
\begin{equation}
    \frac{0.8505 T}{2^n+1} \geq - \log\left(\frac{1}{2} + \delta\right) = \log\left(\frac{2}{1+2\delta}\right).
\end{equation}
We have thus concluded the desired lower bound
\begin{equation}
    T \geq \frac{(2^n + 1)}{0.8505} \log\left(\frac{2}{1+2\delta}\right)
\end{equation}
stated in Theorem~\ref{thm:comp-abs}.

\section{Performing quantum principal component analysis}
\label{sec:qpca}

The (first) principal component of a nonnegative Hermitian matrix $A$ is the eigenvector of $A$ with the largest eigenvalue.
Here, we consider a well-known task called quantum principal component analysis (PCA), which can be achieved efficiently using the quantum algorithm given in \cite{lloyd2014quantum}. The formal definition of the quantum PCA task is given in the following definition.

\begin{task}[Quantum principal component analysis task]
Let $\rho$ be an unknown $n$-qubit mixed state whose top eigenvector $\ket{\phi}$ has eigenvalue larger than all other eigenvalues by a constant factor independent of $n$.
Given a fixed observable $O$, we would like to predict $\bra{\phi}O\ket{\phi}$ up to a small additive error.
\end{task}
We can accomplish this task using the quantum PCA algorithm in \cite{lloyd2014quantum}. In this algorithm, multiple copies of $\rho$ are used in a protocol that approximates the unitary operator $\sum_t |t\rangle\langle t| \otimes \exp(-i \rho t)$, where $t$ is the ``time'' stored in an auxiliary register. 
By applying this conditional $\exp(-i \rho t)$ operation to the initial state $\rho= \sum_i \lambda_i |\phi_i\rangle\langle \phi_i|$, performing the quantum Fourier transform and measuring the auxiliary register, we read out an eigenvalue $\lambda_i$ and prepare the corresponding eigenstate $|\phi_i\rangle$ with probability $\lambda_i$. The eigenvalue can be measured with constant accuracy, and the eigenstate prepared with constant fidelity, using a constant number of copies of $\rho$. Once $|\phi_i\rangle$ has been prepared, we can measure the observable $O$ in this state.

By assumption, the largest eigenvalue of $\rho$ is a constant independent of $n$, and furthermore is greater than all other eigenvalues by a constant. Hence by repeating the above procedure a constant number of times, we can estimate $\langle \phi|O|\phi\rangle$ to constant accuracy, where $|\phi\rangle$ is the eigenstate of $\rho$ with the largest eigenvalue. In contrast, in Appendix~\ref{sec:conv-exp-LB-pca} we show that algorithms that can only learn about $\rho$ through conventional experiments require an exponential number of copies of $\rho$. Bringing these arguments all together, we establish the following theorem.

\begin{theorem}[Exponential advantage for quantum principal component analysis] \label{thm:qPCA}
Let $\rho$ be an unknown $n$-qubit mixed state (for $n > 1$) whose top eigenvector $\ket{\phi}$ has eigenvalue larger than all the other eigenvalues by a constant, and let $Z_1$ be an observable which is equal to the Pauli-$Z$ operator on the first qubit.
Algorithms learn about $\rho$ through conventional or quantum-enhanced experiments.
\begin{itemize}
    \item \emph{Upper bound:} There is an algorithm in the \textbf{quantum-enhanced} scenario using only $\mathcal{O}(1)$ copies of $\rho$ to predict $\bra{\phi} Z_1 \ket{\phi}$ up to $0.25$ error with probability at least $0.8$.
    \item \emph{Lower bound:} Any algorithm in the \textbf{conventional} scenario needs at least $\Omega(2^{n/2})$ copies of $\rho$ to predict $\bra{\phi} Z_1 \ket{\phi}$ up to $0.25$ error with probability at least $0.8$.
\end{itemize}
\end{theorem}

Instead of estimating $\bra{\phi} Z_1 \ket{\phi}$, we can also consider the near-term proposals \cite{cotler2019quantum, huggins2020virtual} to obtain some information about the principal component $\ket{\psi}$ of an unknown state $\rho$ as follows.
The proposals consider $\rho^M / \Tr(\rho^M)$, which approaches $\ketbra{\psi}{\psi}$ when $M$ is large.
In particular, \cite{huggins2020virtual} shows that one can efficiently estimate $\Tr(Z_1 \rho^2) / \Tr(\rho^2)$ by performing entangling Bell measurements over at most two copies of $\rho$ at a time.
By the analysis in \cite{huggins2020virtual}, if the eigenvalue associated to the principal component of $\rho$ is a constant, then $\Tr(Z_i \rho^2) / \Tr(\rho^2)$ can be estimated to any constant error by performing quantum-enhanced experiments over a constant number of copies of $\rho$. In contrast, we show that if one can only measure a single copy of $\rho$ at a time, exponentially many copies are necessary to estimate $\Tr(Z_i \rho^2)/\Tr(\rho^2)$.

\begin{theorem}[Exponential advantage for near-term quantum principal component analysis] \label{thm:qPCA-near}
Suppose we are given an observable $Z_1$ which is equal to the Pauli-$Z$ operator on the first qubit, as well as an $n$-qubit state $\rho$ (for $n > 1$) where an eigenvector $\ket{\phi}$ of $\rho$ has an eigenvalue that is larger than all the other eigenvalues by a constant.
We consider algorithms which learn about $\rho$ through conventional or quantum-enhanced experiments.  Then we have the following bounds:
\begin{itemize}
    \item \emph{Upper bound:} There is an algorithm in the \textbf{quantum-enhanced} scenario using only $\mathcal{O}(1)$ copies of $\rho$ to predict $\Tr(Z_1 \rho^2) / \Tr(\rho^2)$ up to $0.25$ error with probability at least $0.8$.
    \item \emph{Lower bound:} Any algorithm in the \textbf{conventional} scenario needs at least $\Omega(2^{n/2})$ copies of $\rho$ to predict $\Tr(Z_1 \rho^2) / \Tr(\rho^2)$ up to $0.25$ error with probability at least $0.8$.
\end{itemize}
\end{theorem}

\subsection{An exponential lower bound for conventional experiments}
\label{sec:conv-exp-LB-pca}

The lower bound proofs for both Theorem~\ref{thm:qPCA} and Theorem~\ref{thm:qPCA-near} are essentially the same.
We first reduce quantum PCA (or near-term analogs thereof) to a many-versus-many distinguishing task. Then we bound the total variation distance to arrive at the exponential lower bound.

\subsubsection{Reduction to many-versus-many distinguishing task}

We begin by considering a many-versus-many distinguishing task, as discussed in Appendix~\ref{sec:many-vs-many}.
The two hypotheses are given below.

\begin{itemize}
    \item Hypothesis A: The unknown $n$-qubit state $\rho$ is given by
    \begin{equation}
        \rho_A(\ket{\psi}) = \frac{1}{2} \ketbra{0}{0} \otimes \ketbra{\psi}{\psi} + \frac{1}{2} \ketbra{1}{1} \otimes \frac{I}{2^{n-1}},
    \end{equation}
    where $\ket{\psi}$ is an fixed $(n-1)$-qubit pure state, sampled at the outset from the Haar measure.
    \item Hypothesis B: The unknown $n$-qubit state $\rho$ is given by
    \begin{equation}
        \rho_B(\ket{\psi}) = \frac{1}{2} \ketbra{1}{1} \otimes \ketbra{\psi}{\psi} + \frac{1}{2} \ketbra{0}{0} \otimes \frac{I}{2^{n-1}},
    \end{equation}
    where $\ket{\psi}$ is again a fixed $(n-1)$-qubit pure state, sampled at the outset from the Haar measure.
\end{itemize}

It is not hard to see that in hypothesis A, the principal component (largest eigenvector) is $\ket{\phi} = \ket{0} \otimes \ket{\psi}$.
On the other hand, in hypothesis B, the principal component (largest eigenvector) is $\ket{\phi} = \ket{1} \otimes \ket{\psi}$.
Hence, $\bra{\phi} Z_1 \ket{\phi} = 1$ in hypothesis A, but $\bra{\phi} Z_1 \ket{\phi} = -1$ in hypothesis B.
If an algorithm in the conventional scenario can predict $\bra{\psi} Z_1 \ket{\psi}$ up to $0.25$ error with probability at least $0.8$, then we can use the output from the algorithm to distinguish between hypotheses A and B with a success probability of at least $0.8$.

Similarly, we have $\Tr(Z_1 \rho^2) / \Tr(\rho^2) = (2^{n-1} - 1) / (2^{n-1} + 1)$ in hypothesis A and $\Tr(Z_1 \rho^2) / \Tr(\rho^2) = -(2^{n-1} - 1) / (2^{n-1} + 1)$ in hypothesis B.
If an algorithm in the conventional scenario can predict $\Tr(Z_1 \rho^2) / \Tr(\rho^2)$ up to $0.25$ error with probability at least $0.8$, then we can use the output from the algorithm to distinguish between hypothesis A and B with a success probability of at least $0.8$.

Together, a lower bound for distinguishing hypotheses A and B using conventional experiments immediately implies a lower bound for both Theorem~\ref{thm:qPCA} and Theorem~\ref{thm:qPCA-near}.

\subsubsection{Total variation distance}

As in previous sections, let $p_{\rho}(\ell)$ denote the probability to arrive at the leaf node $\ell$ using the learning algorithm in the conventional setting when the unknown state is $\rho$.
If the algorithm can distinguish between hypotheses A and B with success probability $0.8$, then using Eq.~\eqref{eq:TVbound-many-vs-many} we have
\begin{equation}
    \mathrm{TV}\left(\E_{\ket{\psi}} p_{\rho_A(\ket{\psi})}, \E_{\ket{\psi}} p_{\rho_B(\ket{\psi})} \right) \geq 0.6.
\end{equation}
From the triangle inequality, we have
\begin{equation} \label{eq:TVTV0.6}
    \mathrm{TV}\left(\E_{\ket{\psi}} p_{\rho_A(\ket{\psi})}, p_{I / 2^n} \right) + \mathrm{TV}\left(\E_{\ket{\psi}} p_{\rho_B(\ket{\psi})}, p_{I / 2^n} \right) \geq 0.6.
\end{equation}
For each leaf node $\ell$, we consider the path from the root to $\ell$,
\begin{equation}
    u_0 = r \xrightarrow{s_1} u_1 \xrightarrow{s_2} u_2 \xrightarrow{s_3} \ldots \xrightarrow{s_{T-1}} u_{T-1} \xrightarrow{s_T} u_T = \ell.
\end{equation}
At each node $u$, we perform a POVM measurement $\{w^u_s \ketbra{\phi^u_s}{\phi^u_s }\}_s$ on $\rho$ to obtain an outcome $s$ with probability
\begin{equation}
     w^u_s \bra{\phi^u_s} \rho \ket{\phi^u_s}.
\end{equation}
Hence, we can write down the probability to arrive at the leaf $\ell$ as
\begin{equation}
    p_{\rho}(\ell) = \prod_{t=1}^{T} w^{u_{t-1}}_{s_t} \bra{\phi^{u_{t-1}}_{s_t}} \rho \ket{\phi^{u_{t-1}}_{s_t}}.
\end{equation}
We will use $\rho(\ket{\psi})$ to denote either $\rho_A(\ket{\psi})$ or $\rho_B(\ket{\psi})$. Then recalling \eqref{eq:tv_fdiv}, we have
\begin{equation}
    \mathrm{TV}\left(\E_{\ket{\psi}} p_{\rho(\ket{\psi})}, p_{I / 2^n} \right) = \sum_{\ell} p_{I / 2^n}(\ell) \max\left(0, \,\, 1 - \E_{\ket{\psi}} \prod_{t=1}^{T} 2^n \bra{\phi^{u_{t-1}}_{s_t}} \rho(\ket{\psi}) \ket{\phi^{u_{t-1}}_{s_t}} \right). \label{eq:TV-PCA}
\end{equation}

\subsubsection{Upper bound for total variation distance}

The central quantity to control in our bound on the total variation distance is
\begin{equation}
    \E_{\ket{\psi}} \prod_{t=1}^{T} 2^n \bra{\phi^{u_{t-1}}_{s_t}} \rho(\ket{\psi}) \ket{\phi^{u_{t-1}}_{s_t}}.
\end{equation}
Without loss of generality, let us consider $\rho(\ket{\psi}) = \rho_A(\ket{\psi}) = \frac{1}{2} \ketbra{0}{0} \otimes \ketbra{\psi}{\psi} + \frac{1}{2} \ketbra{1}{1} \otimes \frac{I}{2^{n-1}}$.
Suppose each $\ket{\phi^{u_{t-1}}_{s_t}}$ takes the form
\begin{equation}
    \ket{\phi^{u_{t-1}}_{s_t}} = \alpha^{u_{t-1}}_{s_t} \ket{0} \otimes \ket{\phi^{u_{t-1}, 0}_{s_t}} + \beta^{u_{t-1}}_{s_t} \ket{1} \otimes \ket{\phi^{u_{t-1}, 1}_{s_t}},
\end{equation}
where $\alpha^{u_{t-1}}_{s_t}, \beta^{u_{t-1}}_{s_t} \in \mathbb{C}$ and $\left|\alpha^{u_{t-1}}_{s_t}\right|^2 + \left|\beta^{u_{t-1}}_{s_t}\right|^2 = 1$.
Then we have
\begin{align}
    \E_{\ket{\psi}} \prod_{t=1}^{T} 2^n \bra{\phi^{u_{t-1}}_{s_t}} \rho(\ket{\psi}) \ket{\phi^{u_{t-1}}_{s_t}} &= \E_{\ket{\psi}} \prod_{t=1}^{T} \left( 2^{n-1} \left|\alpha^{u_{t-1}}_{s_t}\right|^2 \left| \braket{\psi| \phi^{u_{t-1}, 0}_{s_t}} \right|^2  + \left|\beta^{u_{t-1}}_{s_t}\right|^2 \right)\\
    &= \sum_{S \subseteq \{1, \ldots, T\}} \prod_{t \not\in S} \left|\beta^{u_{t-1}}_{s_t}\right|^2  \left[\E_{\ket{\psi}} \prod_{t \in S} 2^{n-1} \left|\alpha^{u_{t-1}}_{s_t}\right|^2 \left| \braket{\psi| \phi^{u_{t-1}, 0}_{s_t}} \right|^2 \right]. \label{eq:purity-testing}
\end{align}
We need to lower bound the above quantity in order to upper bound the total variation distance.
In order to do so, we utilize the following lemma. The proof of the lemma is based on Haar integration.
For readers unfamiliar with Haar integration, we would suggest skipping the proof of this lemma.

\begin{lemma}[High moment bound for Haar-random state]
Consider any $m$-qubit pure states $\ket{\phi_1}, \ldots, \ket{\phi_K}$ and an $m$-qubit pure state $\ket{\psi}$ sampled from the Haar measure, we have
\begin{equation}
    \E_{\ket{\psi}} \prod_{k=1}^K \left| \braket{\psi| \phi_k} \right|^2 \geq \frac{1}{(2^m + K - 1) \ldots (2^m+1) (2^m)}.
\end{equation}
\end{lemma}
\begin{proof}
The Haar integration over states shows that
\begin{equation}
    \E_{\ket{\psi}} \ketbra{\psi}{\psi}^{\otimes K} = \frac{1}{(2^m + K - 1) \ldots (2^m+1) (2^m)} \sum_{\pi \in \mathcal{S}_K} \pi,
\end{equation}
where $\mathcal{S}_K$ is the permutation group of $K$ items, and $\pi$ is the permutation operator over the $K$ tensor-product space.
From Lemma~5.12 in \cite{chen2021exponential}, we have
\begin{equation}
    \sum_{\pi \in \mathcal{S}_K} \Tr\left( \pi \bigotimes_{k=1}^K \ketbra{\phi_k}{\phi_k} \right) \geq 1.
\end{equation}
Therefore, we find
\begin{equation}
    \E_{\ket{\psi}} \prod_{k=1}^K \left| \braket{\psi| \phi_k} \right|^2 \geq  \frac{1}{(2^m + K - 1) \ldots (2^m+1) (2^m)}.
\end{equation}
This concludes the proof.
\end{proof}

We apply this lemma with
\begin{equation}
    m \equiv n-1, \quad K \equiv |S|, \quad \ket{\phi_k} \equiv \ket{\phi^{u_{t-1}, 0}_{s_t}}
\end{equation}
to obtain the following lower bound,
\begin{align}
    \E_{\ket{\psi}} \prod_{t \in S} 2^{n-1} \left|\alpha^{u_{t-1}}_{s_t}\right|^2 \left| \braket{\psi| \phi^{u_{t-1}, 0}_{s_t}} \right|^2 &\geq \prod_{t \in S} \frac{\left|\alpha^{u_{t-1}}_{s_t}\right|^2}{(1 + \frac{|S|-1}{2^{n-1}}) \ldots (1 + \frac{1}{2^{n-1}})(1)}\\
    &\geq \prod_{t \in S} \frac{\left|\alpha^{u_{t-1}}_{s_t}\right|^2}{\left(1 + \frac{|S|-1}{2^{n-1}}\right)^{|S|-1}}\\
    &\geq \left(1 + \frac{T-1}{2^{n-1}}\right)^{-(T-1)} \prod_{t \in S} \left|\alpha^{u_{t-1}}_{s_t}\right|^2.
\end{align}
Combining with Eq.~\eqref{eq:purity-testing}, we have
\begin{align}
     \E_{\ket{\psi}} \prod_{t=1}^{T} 2^n \bra{\phi^{u_{t-1}}_{s_t}} \rho(\ket{\psi}) \ket{\phi^{u_{t-1}}_{s_t}}  &\geq \left(1 + \frac{T-1}{2^{n-1}}\right)^{-(T-1)} \sum_{S \subseteq \{1, \ldots, T\}} \prod_{t \not\in S} \left|\beta^{u_{t-1}}_{s_t}\right|^2 \prod_{t \in S} \left|\alpha^{u_{t-1}}_{s_t}\right|^2 \\
     &= \left(1 + \frac{T-1}{2^{n-1}}\right)^{-(T-1)} \prod_{t = 1}^T\left( \left|\beta^{u_{t-1}}_{s_t}\right|^2 + \left|\alpha^{u_{t-1}}_{s_t}\right|^2\right) \\
     &= \left(1 + \frac{T-1}{2^{n-1}}\right)^{-(T-1)}.
\end{align}
Next we leverage Eq.~\eqref{eq:TV-PCA} to obtain
\begin{align}
\mathrm{TV}\left(\E_{\ket{\psi}} p_{\rho(\ket{\psi})}, p_{I / 2^n} \right) &\leq \sum_{\ell} p_{I / 2^n}(\ell) \max\left(0, \,\, 1 - \left(1 + \frac{T-1}{2^{n-1}}\right)^{-(T-1)} \right)\\
&= 1 - \left(1 + \frac{T-1}{2^{n-1}}\right)^{-(T-1)}.
\end{align}
The second line follows from $\sum_{\ell} p_{I / 2^n}(\ell) = 1$ because $p_{I / 2^n}(\ell)$ is a probability distribution.

\subsubsection{Lower bound for the number of measurements}

We can now utilize the lower bound on the total variation distance given in Eq.~\eqref{eq:TVTV0.6} and the upper bound obtained above to find
\begin{equation}
\label{E:usefulTVboundQPCA1}
    2\left(1 - \left(1 + \frac{T-1}{2^{n-1}}\right)^{-(T-1)}\right) \geq     \mathrm{TV}\left(\E_{\ket{\psi}} p_{\rho_A(\ket{\psi})}, p_{I / 2^n} \right) + \mathrm{TV}\left(\E_{\ket{\psi}} p_{\rho_B(\ket{\psi})}, p_{I / 2^n} \right) \geq 0.6.
\end{equation}
Hence, we have the inequality
\begin{align}
    &0.7 \geq \left(1 + \frac{T-1}{2^{n-1}}\right)^{-(T-1)},\\
    &\implies (T-1) \log\left(1 + \frac{T-1}{2^{n-1}}\right) \geq \log(10 / 7).
\end{align}
Because $\log(1 + x) \leq x$ for all  $x > -1$, we have
\begin{equation}
(T-1)^2 \geq 2^{n-1} \log(10 / 7) \implies T \geq 1 + \sqrt{\frac{\log\left(\frac{10}{7}\right)}{2}}\, 2^{n/2}.
\end{equation}
Finally, we have established the lower bound $T = \Omega(2^{n/2})$ stated in Theorem~\ref{thm:qPCA}.

\subsection{An exponential lower bound using pseudorandomness}
\label{sec:pseudo}

In our exponential lower bound in the previous subsection, we relied on a state that was in part constructed using a Haar-random $(n-1)$-qubit state $|\psi\rangle$.  However, preparing a Haar-random state from a simple initial state (say, a product state) requires  circuit depth exponential in $n$.  As such, we can not prepare Haar-random states in practice.  Accordingly we cannot prepare either $\rho_A(|\psi\rangle)$ or $\rho_B(|\psi\rangle)$ in realistic circumstances.

However, we could instead efficiently construct pseudorandom states $|\psi\rangle$ which are (very plausibly) indistinguishable from Haar-random states if we probe with any POVM instantiated by a $\text{poly}(n)$-time quantum algorithm~\cite{ji2018pseudorandom}.  We will elaborate on this shortly.  The parenthetical caveat `very plausibly' is due to the fact that the construction we use relies on cryptographic assumptions which are not proven, but are widely believed.  In particular, we need to suppose the existence of quantum-secure one-way functions~\cite{ji2018pseudorandom}; these are functions which are efficient to evaluate but are hard to invert even with a quantum computer.  Making such an assumption is standard practice in computational complexity theory, and so we proceed apace.

Let us recall the definition of a pseudorandom quantum state:
\begin{definition}[Pseudorandom quantum states; paraphrased from Definition 3 of~\cite{ji2018pseudorandom}]
Given a set $\mathcal{K}$, a family of pseudorandom quantum states on $n$ qubits is a family of states $\{|\phi_k\rangle\}_{k \in \mathcal{K}}$ and a probability distribution $\cD$ over $\mathcal{K}$ such that:
\begin{itemize}
    \item There is a $\text{\rm poly}(n)$-time quantum algorithm that samples a single element $k$ from $\mathcal{K}$ according to $\cD$ and generates the corresponding state $|\phi_k\rangle$;
    \item For any polynomial $t(n)$ and any $\text{\rm poly}(n)$-time quantum algorithm $\mathcal{A}$ with outputs in $\{0,1\}$, we have
    \begin{equation}
        \label{E:compindist1}
        \left|\text{\rm Pr}_{k \leftarrow \mathcal{K}}\big[\mathcal{A}(|\phi_k\rangle^{\otimes t(n)}) = 1\big] - \text{\rm Pr}_{|\psi\rangle \leftarrow \text{\rm Haar}}\big[\mathcal{A}(|\psi\rangle^{\otimes t(n)}) = 1\big] \right| \leq \text{\rm negl}(n).
    \end{equation}
    Here $\text{\rm negl}(n)$ is a function such that for all constants $c > 0$ we have $\text{\rm negl}(n) < n^{-c}$ for $n$ sufficiently large.
\end{itemize}
\label{def:pseudostates}
\end{definition}
\noindent We can interpret the above definition as follows.  Eq.~\eqref{E:compindist1} says that if we are given a polynomial $t(n)$ number of copies of either (i) a fixed pseudorandom state, or (ii) a fixed Haar random state, any polynomial time quantum algorithm with binary outputs cannot distinguish between the two cases.  Since an exponential depth quantum algorithm can distinguish between the two cases,~\eqref{E:compindist1} describes a notion of \textit{computational indistinguishability}, i.e.~we cannot distinguish with polynomial time computational resources. 

A key question is: do pseudorandom quantum states exist?  We recall the following, contingent result.
\begin{lemma}[Existence of pseudorandom quantum states~\cite{ji2018pseudorandom}] If there exist quantum-secure one-way functions, then they can be used to construct pseudorandom quantum states.
\end{lemma}
\noindent Explicit details of the construction can be found in~\cite{ji2018pseudorandom}; there have also been refinements and generalizations in follow-up work (see e.g.~\cite{brakerski2019pseudo, brakerski2020scalable}).

Before stating the main result of this section, we require the following definition:
\begin{definition}[Polynomial-time algorithms in the conventional scenario]
A polynomial-time algorithm $\mathcal{A}$ in the conventional scenario is constructed as follows.  We consider a learning tree $\mathcal{T}$ in the \textbf{conventional scenario} with leaves $\ell$, and require that the protocol described by the learning tree can be implemented by an at most $\text{\rm poly}(n)$-time quantum algorithm.  (As such, the depth of $\mathcal{T}$ is at most polynomial in $n$.) Then let $\mathcal{D}$ be a $\text{\rm poly}(n)$-time classical algorithm which, given the transcript of measurements encoded into the leaves $\ell$ of $\mathcal{T}$, provides a binary output $0$ or $1$.  We let $\mathcal{A}$ be a map from $n$-qubit density matrices $\rho$ to $\{0,1\}$, corresponding to instantiating the learning tree $\mathcal{T}$ on copies of $\rho$, followed by using $\mathcal{D}$ on the measurement transcript to determine a binary outcome.
\end{definition}
\noindent We can now leverage the putative pseudorandom states and the above definition to establish the following result:
\begin{theorem}[Lower bound many-versus-many distinguishing task using pseudorandom states]
Let $\{|\phi_k\rangle\}_{k \in \mathcal{K}}$ be a family of pseudorandom states on $n-1$ qubits.  Then any polynomial-time algorithm $\mathcal{A}$ in the \textbf{conventional} scenario with binary output cannot distinguish $\rho_A(|\phi_k\rangle)$ from $\rho_B(|\phi_k\rangle)$ for $k$ sampled from the probability distribution over $\mathcal{K}$.  That is:
    \begin{equation}
    \label{E:compindist2}
    \left|\text{\rm Pr}_{k \leftarrow \mathcal{K}}\big[\mathcal{A}(\rho_A(|\phi_k\rangle)) = 1\big] - \text{\rm Pr}_{k \leftarrow \mathcal{K}}\big[\mathcal{A}(\rho_B(|\phi_k\rangle)) = 1\big] \right| \leq \text{\rm negl}(n)\,.
    \end{equation}
\end{theorem}
\begin{proof}
Using the triangle inequality several times we have
\begin{align}
\label{E:pseudotriangle1}
 &\left|\text{\rm Pr}_{k \leftarrow \mathcal{K}}\big[\mathcal{A}(\rho_A(|\phi_k\rangle)) = 1\big] - \text{\rm Pr}_{k \leftarrow \mathcal{K}}\big[\mathcal{A}(\rho_B(|\phi_k\rangle)) = 1\big] \right|  \\
 & \qquad \qquad \qquad \qquad \leq \left|\text{\rm Pr}_{|\psi\rangle \leftarrow \text{Haar}}\big[\mathcal{A}(\rho_A(|\psi\rangle)) = 1\big] - \text{\rm Pr}_{|\psi\rangle \leftarrow \text{Haar}}\big[\mathcal{A}(\rho_B(|\psi\rangle)) = 1\big]\right| \nonumber \\
 & \qquad \qquad \qquad \qquad \qquad + \left|\text{\rm Pr}_{k \leftarrow \mathcal{K}}\big[\mathcal{A}(\rho_A(|\phi_k\rangle)) = 1\big] - \text{\rm Pr}_{|\psi\rangle \leftarrow \text{Haar}}\big[\mathcal{A}(\rho_A(|\psi\rangle)) = 1\big]\right| \nonumber \\
 & \qquad \qquad \qquad \qquad \qquad + \left|\text{\rm Pr}_{k \leftarrow \mathcal{K}}\big[\mathcal{A}(\rho_B(|\phi_k\rangle)) = 1\big] - \text{\rm Pr}_{|\psi\rangle \leftarrow \text{Haar}}\big[\mathcal{A}(\rho_B(|\psi\rangle)) = 1\big]\right|\,. \nonumber
\end{align}
Let $\mathcal{T}$ be the learning tree corresponding to $\mathcal{A}$, and let $\mathcal{D}$ be the binary decision function mapping leaves of $\mathcal{T}$ to $\{0,1\}$.  The depth $T$ of $\mathcal{T}$ is necessarily at most polynomial in $n$; let us denote the depth by $T(n)$.  Then the first term on the right-hand side of~\eqref{E:pseudotriangle1} is upper bounded by
\begin{align}
\label{E:firstpseudoterm}
&\left|\text{\rm Pr}_{|\psi\rangle \leftarrow \text{Haar}}\big[\mathcal{A}(\rho_A(|\psi\rangle)) = 1\big] - \text{\rm Pr}_{|\psi\rangle \leftarrow \text{Haar}}\big[\mathcal{A}(\rho_B(|\psi\rangle)) = 1\big]\right|\\
&= \left|\sum_{\ell \,\in\, \text{leaf}(\mathcal{T}) \, : \, \mathcal{D}(\ell) = 1}\left( \E_{\ket{\psi}} p_{\rho_A(|\psi\rangle)}(\ell) - \E_{\ket{\psi}} p_{\rho_B(|\psi\rangle)}(\ell)\right) \right| \nonumber \\
&\leq \sum_{\ell \,\in\,\text{leaf}(\mathcal{T})} \left|\E_{\ket{\psi}} p_{\rho_A(|\psi\rangle)}(\ell) - \E_{\ket{\psi}} p_{\rho_B(|\psi\rangle)}(\ell)\right| \nonumber \\
&\leq 2\,\mathrm{TV}\left(\E_{\ket{\psi}} p_{\rho_A(\ket{\psi})}, p_{I / 2^n} \right) + 2\,\mathrm{TV}\left(\E_{\ket{\psi}} p_{\rho_B(\ket{\psi})}, p_{I / 2^n} \right) \nonumber \\
&\leq 4\left(1 - \left(1 + \frac{T(n)-1}{2^{n-1}}\right)^{-(T(n)-1)}\right)
\end{align}
where the last inequality comes from~\eqref{E:usefulTVboundQPCA1}. 

Next we turn to bounding the second term on the right-hand side of the inequality in~\eqref{E:pseudotriangle1}, namely
\begin{equation*}
\left|\text{\rm Pr}_{k \leftarrow \mathcal{K}}\big[\mathcal{A}(\rho_A(|\phi_k\rangle)) = 1\big] - \text{\rm Pr}_{|\psi\rangle \leftarrow \text{Haar}}\big[\mathcal{A}(\rho_A(|\psi\rangle)) = 1\big]\right| \,.
\end{equation*}
First, we observe that every $\mathcal{A}$ takes in $T(n)$ copies of $\rho_A$.  Accordingly, there is a polynomial-time algorithm $\widetilde{\mathcal{A}}$ such that $\widetilde{\mathcal{A}}(\rho_A^{\otimes T(n)}) = \mathcal{A}(\rho_A)$ for all inputs $\rho_A$; this just amounts to a slightly different way of notating the domain of the algorithm $\mathcal{A}$.  Then we can rewrite our term of interest as
\begin{equation}
\label{E:pseudointermediate1}
\left|\text{\rm Pr}_{k \leftarrow \mathcal{K}}\big[\widetilde{\mathcal{A}}(\rho_A(|\phi_k\rangle)^{\otimes T(n)}) = 1\big] - \text{\rm Pr}_{|\psi\rangle \leftarrow \text{Haar}}\big[\widetilde{\mathcal{A}}(\rho_A(|\psi\rangle)^{\otimes T(n)}) = 1\big]\right| \,.
\end{equation}
But now observe that for any state $|\omega\rangle$, there is a polynomial-time quantum algorithm which takes $|\omega\rangle$ as input and produces $\rho_A(|\omega\rangle)$ as output.  Accordingly, repeating this algorithm on $T(n)$ copies of $|\omega\rangle$, we produce $T(n)$ copies of $\rho_A(|\omega\rangle)$; let us denote this $T(n)$-copy algorithm by $\mathcal{B}$.  We have $\mathcal{B}(|\omega\rangle^{\otimes T(n)}) = \rho_A(|\omega\rangle)^{\otimes T(n)}$, where $\mathcal{B}$ runs in polynomial time (recalling that $T(n)$ is polynomial in $n$).  Then we can write~\eqref{E:pseudointermediate1} as
\begin{equation}
\label{E:pseudointermediate2}
\left|\text{\rm Pr}_{k \leftarrow \mathcal{K}}\big[(\widetilde{\mathcal{A}} \circ \mathcal{B})(|\phi_k\rangle^{\otimes T(n)}) = 1\big] - \text{\rm Pr}_{|\psi\rangle \leftarrow \text{Haar}}\big[(\widetilde{\mathcal{A}}\circ \mathcal{B})(|\psi\rangle^{\otimes T(n)}) = 1\big]\right| \,.
\end{equation}
Since $\widetilde{\mathcal{A}} \circ \mathcal{B}$ is itself a polynomial-time quantum algorithm, Definition~\ref{def:pseudostates} tells us that~\eqref{E:pseudointermediate2} above is upper bounded by $\text{negl}(n)$.  So in summary, we find
\begin{equation}
\label{E:secondpseudoterm}
\left|\text{\rm Pr}_{k \leftarrow \mathcal{K}}\big[\mathcal{A}(\rho_A(|\phi_k\rangle)) = 1\big] - \text{\rm Pr}_{|\psi\rangle \leftarrow \text{Haar}}\big[\mathcal{A}(\rho_A(|\psi\rangle)) = 1\big]\right| \leq \text{negl}(n)\,.
\end{equation}
In an identical manner we can show that the third term on the right-hand side of~\eqref{E:pseudotriangle1} satisfies the bound
\begin{equation}
\label{E:thirdpseudoterm}
\left|\text{\rm Pr}_{k \leftarrow \mathcal{K}}\big[\mathcal{A}(\rho_B(|\phi_k\rangle)) = 1\big] - \text{\rm Pr}_{|\psi\rangle \leftarrow \text{Haar}}\big[\mathcal{A}(\rho_B(|\psi\rangle)) = 1\big]\right| \leq \text{negl}(n)\,.
\end{equation}

Putting together the inequalities in~\eqref{E:pseudotriangle1},~\eqref{E:firstpseudoterm},~\eqref{E:secondpseudoterm} and~\eqref{E:thirdpseudoterm}, as well as observing that $T(n)$ is at most polynomially large in $n$, we achieve the desired bound.
\end{proof}

The above Theorem immediately implies that the exponential advantage for quantum principal component analysis in Theorem~\ref{thm:qPCA} has a counterpart for pseudorandom states.  We note that in the pseudorandom context the advantage is not strictly exponential; rather, the advantage holds for arbitrary \textit{polynomial}-time quantum learning algorithms in the quantum-enhanced scenario versus in the conventional scenario.

\section{Learning a polynomial-time quantum process}
\label{sec:evolution}

Here we consider the problem of learning a polynomial-time quantum process, and provide a rigorous exponential separation between the conventional and quantum-enhanced learning settings.

\subsection{Problem setting}
\label{sec:learn-process-setting}

We consider an unknown quantum process $\cE$ on $n$ qubits generated as follows.
\begin{itemize}
    \item An $n$-qubit state $\sigma$ is accompanied by an $m$ ancillary qubits initialized at $\ketbra{0^n}{0^n}$.
    \item $p$ unknown two-qubit unitary gates are applied on the $(n+m)$-qubit system $\rho \otimes \ketbra{0^n}{0^n}$.
    \item The ancillary qubits are hidden, resulting in an $n$-qubit mixed state.
\end{itemize}
When $m, p = \mathrm{poly}(n)$, we refer to $\cE$ as a polynomial-time quantum process.
Next, we consider an input probability distribution $\cD$ over $n$-qubit mixed states $\sigma$.
The goal is to learn an approximate model $\tilde{\cE}$ of $\cE$, such that we can accurately predict the output state on average: 
\begin{equation}
    \E_{\sigma \sim \cD} \norm{\tilde{\cE}(\sigma) - \cE(\sigma)}_1 \leq \epsilon.
\end{equation}
In the above, $\norm{X}_1 = \max_{O: \norm{O}_\infty \leq 1} |\Tr(O X)|$ is the trace norm.

\subsection{Rigorous statements}
\label{sec:rigor-learn-process}

We have the following theorem for quantum-enhanced experiments showing that they can efficiently learn a polynomial-time quantum process.
We will prove the theorem later in Appendix~\ref{sec:poly-proc}.

\begin{theorem}[Approximate learning of quantum processes -- polynomial upper bound] \label{thm:poly-upp-learn-proc}
For any distribution $\cD$ and any $\epsilon, \delta > 0$, there exists a learning algorithm in the quantum-enhanced setting that can learn an approximate model $\tilde{\cE}$ such that with probability at least $1 - \delta$,
\begin{equation}\label{eq:thm8}
    \E_{\sigma \sim \cD} \norm{\tilde{\cE}(\sigma) - \cE(\sigma)}_1 \leq \epsilon
\end{equation}
from at most $\wt{\mathcal{O}}(\mathrm{poly}(n)\log(1 / \delta)/\epsilon^4)$ accesses to $\cE$, where $\wt{\mathcal{O}}(\cdot)$ hides factors of $\log(1/\epsilon)$.
\end{theorem}
\noindent In contrast, our hardness results for predicting properties of physical states in the conventional setting (see Theorem~\ref{thm:obs-adv} in Appendix~\ref{sec:expo-adv-abs-one}) immediately implies the following exponential lower bound.

\begin{corollary}[Approximate learning of quantum processes -- exponential lower bound] \label{cor:learn-proc-hard}
Let $\cD$ be any distribution and $\cE$ be the quantum process that always generates a state $\rho$ considered in Def.~\ref{def:hardrhoO}. 
Any algorithm in the conventional setting that learns an approximate model $\tilde{\cE}$ such that
\begin{equation}
    \E_{\sigma \sim \cD} \norm{\tilde{\cE}(\sigma) - \cE(\sigma)}_1 \leq 0.25,
\end{equation}
must use at least $\Omega(2^n)$ accesses to $\cE$.
\end{corollary}
\begin{proof}
We consider $m = 2n$.
The two-qubit gates swap the input state $\sigma$ to the first $n$ ancillary qubits.
Then we use the rest of the $n$ ancillary qubits and the $n$ system qubits (i.e.~the qubits in the input state $\rho$) to prepare a state $\rho$ considered in Def.~\ref{def:hardrhoO}.
To prepare the maximally mixed state $I / 2^n$, we entangle each of the system qubits with the corresponding ancillary qubit to prepare a Bell state $\frac{1}{\sqrt{2}}(\ket{00} + \ket{11})$.
To prepare the alternative state $(I + 0.9 s P) / 2^n$, we perform the following procedure.
\begin{enumerate}
    \item For qubit $i$, where $P_i$ is \emph{not} the last non-identity Pauli operator (i.e.~last in terms of having the largest index $i$), we entangle the $i$-th system qubit with the corresponding ancillary qubit to prepare a Bell state $\frac{1}{\sqrt{2}}(\ket{00} + \ket{11})$.
    \item For the last remaining qubit $i$, which corresponds to the index $i$ such that $P_i$ is the last non-identity Pauli operator, we apply a sequence of two-qubit gates entangling qubit $i$ to qubit $j$ with $P_j \neq I$. The sequence of two-qubit gates stores the parity (or $1-\mathrm{parity}$) of all qubits $j$ with $P_j \neq I$.
    After this step, when we trace over the ancillary qubits, we have generated $(I+s P^{(Z)})/2^n$, where $P^{(Z)} = \bigotimes_{i=1}^n F(P_i)$ and $F(I) = I, F(X) = Z, F(Y) = I, F(Z) = I$.
    Then, we rotate the corresponding ancillary qubit for qubit $i$ from $\ket{0}$ to $\sqrt{0.95} \ket{0} + \sqrt{0.05} \ket{1}$ and apply a controlled-not gate from the ancillary qubit (control) to qubit $i$.
    The system qubits are now in the state $(I+ 0.9 s P^{(Z)})/2^n$.
    \item Finally, for each qubit $i$ with $P_i = X$, we rotate the system qubit from $|0\rangle$ to $|+\rangle = \frac{1}{\sqrt{2}}(|0\rangle + |1\rangle)$ and from $|1\rangle$ to $|-\rangle = \frac{1}{\sqrt{2}}(|0\rangle - |1\rangle)$. For each qubit $i$ with $P_i = Y$, we rotate from $0$ to $y+$ and from $1$ to $y-$.
    Tracing over the ancillary qubits, the system qubits are now in the state $(I+ 0.9 s P)/2^n$.
\end{enumerate}
We can see that the number of gates $p$ is $\mathcal{O}(n)$.
Furthermore, no matter what the input state $\sigma$ is, the above quantum process always produces a state $\rho$ considered in Def.~\ref{def:hardrhoO}.

When an algorithm in the conventional setting has learned an approximate model $\tilde{\cE}$ with
\begin{equation}
    \E_{\sigma \sim \cD} \norm{\tilde{\cE}(\sigma) - \cE(\sigma)}_1 \leq 0.25,
\end{equation}
we can use Jensen's inequality to conclude
\begin{equation}
    \norm{\E_{\sigma \sim \cD} \tilde{\cE}(\sigma) - \rho }_1 \leq \E_{\sigma \sim \cD} \norm{\tilde{\cE}(\sigma) - \cE(\sigma)}_1 \leq 0.25.
\end{equation}
Because $\norm{X}_1 = \max_{O: \norm{O}_\infty \leq 1} |\Tr(O X)|$,
the above implies that the algorithm can predict $\Tr(O \rho)$ up to an error of $0.25$.
From Theorem~\ref{thm:obs-adv}, we conclude that the learning algorithm must use at least $\Omega(2^n)$ copies of $\rho$, which corresponds to $\Omega(2^n)$ accesses to $\cE$.
\end{proof}

\subsection{Proof of polynomial upper bound in Theorem~\ref{thm:poly-upp-learn-proc}}
\label{sec:poly-proc}

\begin{figure}[t]
    \centering
    \includegraphics[width=0.75\textwidth]{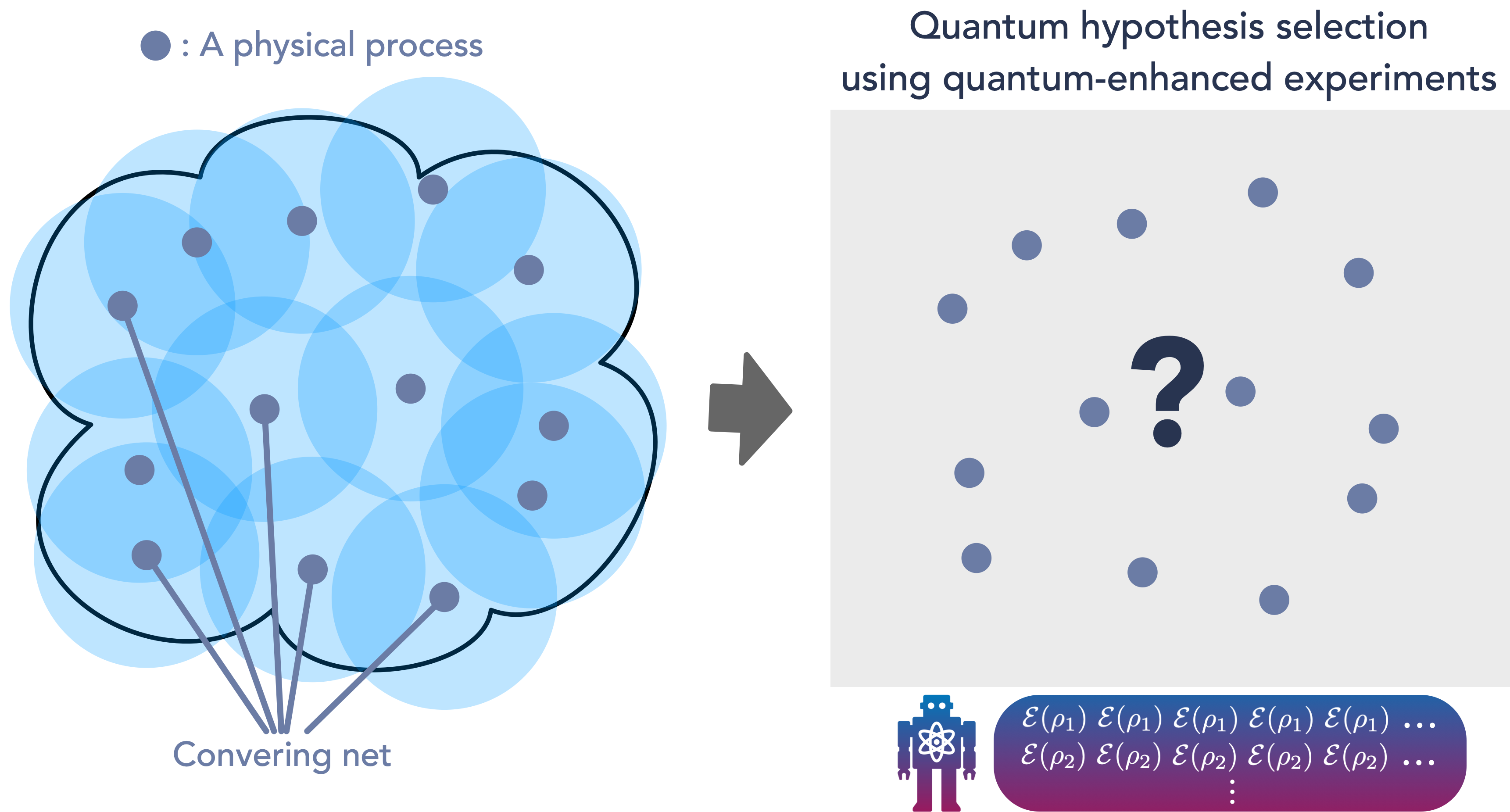}
    \caption{\emph{Illustration for the proof of Theorem~\ref{thm:poly-upp-learn-proc} on learning polynomial-time quantum processes.} We first form a covering net (all dark blue dots) for the space of all polynomial-time quantum processes (the cloud shape). Any polynomial-time quantum process is close to an element in the covering net. Then we perform quantum hypothesis selection \cite{buadescu2020improved} using a quantum dataset stored in the quantum memory to find the approximate  physical process. }
    \label{fig:learn-process}
\end{figure}

Theorem~\ref{thm:poly-upp-learn-proc} establishes an upper bound on the number of times we must access the unknown process $\mathcal{E}$ in the quantum-enhanced setting to construct an approximate model of $\mathcal{E}$. Note that the theorem only concerns the number of times we run the process $\mathcal{E}$; it does not address the computational complexity of the learning procedure.
Our strategy for proving the theorem is as follows. First we find an upper bound on the number of elements of an $\epsilon'$-covering net for the set of all quantum processes that can be constructed using up to $p$ two-qubit quantum gates, with distance defined by the diamond norm.
Next we explain how to use a quantum hypothesis testing algorithm to find a process $\mathcal{\tilde E}$ in the covering net that approximates $\mathcal{E}$ as specified in (\ref{eq:thm8}), if $\epsilon'$ is appropriately chosen.
This quantum hypothesis testing method can be carried out in the quantum-enhanced setting, but not in the conventional setting.
The number of times we must access $\mathcal{E}$ depends on the size of the covering net, and can be shown to scale polynomially with the number of gates $p$, proving the theorem.
An illustration is given in Supp.~Fig.~\ref{fig:learn-process}.

\subsubsection{Covering net}

First, we construct the covering net for the set $\cS$ of quantum processes with a fixed $n, m, p$.
An $\epsilon$-covering net of a set $\cS$ is a subset $\cN_\epsilon \subseteq \cS$ such that for every point $x \in \cS$, there exists a point $y \in \cN$ with $\norm{x - y} \leq \epsilon$ in an appropriate norm.

Recall that a unitary $U$ corresponds to a unitary channel $\cU$ defined as
\begin{equation}
    \cU(\rho) = U \rho U^\dagger.
\end{equation}
Because two-qubit unitary channels form a bounded set in a finite-dimensional space, the $\tilde{\epsilon}$-covering net for two-qubit unitary channels has a size of at most
\begin{equation}
    \left( \frac{c_1}{\tilde{\epsilon}} \right)^{c_2}, \label{eq:twoqubit}
\end{equation}
where $c_1, c_2$ are two constants (see, e.g., Section $4.2$ in \cite{vershynin2018highdimensional}).
Here, we consider the norm to be the diamond norm $\norm{\cdot}_\diamond$ (see Section~3.3 in  \cite{watrous2018book}).
The bound in \eqref{eq:twoqubit} only pertains to the covering net size when the unitary acts on a \emph{fixed} set of two qubits.
Let us now consider two-qubit unitary channels that can act on any two of the $n+m$ qubits.
Because there are $\binom{n+m}{2}$ pairs of qubits that the unitary could act on, the size of the $\tilde{\epsilon}$-covering net $\cN_{\tilde{\epsilon}, n+m}$ of all two-qubit gates on an $(n+m)$-qubit system is upper bounded as follows,
\begin{equation}
\left| \cN_{\tilde{\epsilon}, n+m} \right| \leq \binom{n+m}{2}  \left( \frac{c_1}{\tilde{\epsilon}} \right)^{c_2}.
\end{equation}
To construct an $\epsilon$-covering net for the composed quantum process $\cE$, we need to consider $\tilde{\epsilon}= \epsilon' / p$ in $\cN_{\tilde{\epsilon}, n+m}$.
Consider any sequence of two-qubit unitary channels $\cU_1, \ldots, \cU_p$ on an $(n+m)$-qubit system.
For each $U_i$ in the sequence, we find the closest unitary channel $\tilde{\cU}_i$ in $\cN_{\tilde{\epsilon}, n+m}$, hence $\norm{\cU_i - \tilde{\cU}_i}_{\diamond} \leq \tilde{\epsilon}$.
Then we can use a telescoping sum and the triangle inequality to see that
\begin{align}
    &\norm{\Tr_{n+1, \ldots, n+m}\left( \cU_p \ldots \cU_1 (\rho \otimes \ketbra{0^m}{0^m}) \right) - \Tr_{n+1, \ldots, n+m}\left( \tilde{\cU}_p \ldots \tilde{\cU}_1 (\rho \otimes \ketbra{0^m}{0^m}) \right)}_1 \\
    &\leq \norm{ \cU_p \ldots \cU_1 (\rho \otimes \ketbra{0^m}{0^m}) - \tilde{\cU}_p \ldots \tilde{\cU}_1 (\rho \otimes \ketbra{0^m}{0^m}) }_1\\
    &\leq \norm{ \cU_p \ldots \cU_1 (\rho \otimes \ketbra{0^m}{0^m}) - \cU_p \tilde{\cU}_{p-1} \ldots \tilde{\cU}_1 (\rho \otimes \ketbra{0^m}{0^m}) }_1 \nonumber \\
    &\qquad + \norm{ \cU_p \tilde{\cU}_{p-1} \ldots \cU_1 (\rho \otimes \ketbra{0^m}{0^m}) - \tilde{\cU}_{p} \tilde{\cU}_{p-1} \ldots \tilde{\cU}_1 (\rho \otimes \ketbra{0^m}{0^m}) }_1 \\
    &\leq \norm{ \cU_{p-1} \ldots \cU_1 (\rho \otimes \ketbra{0^m}{0^m}) - \tilde{\cU}_{p-1} \ldots \tilde{\cU}_1 (\rho \otimes \ketbra{0^m}{0^m}) }_1 + \norm{ \cU_p  - \tilde{\cU}_{p} }_\diamond \\
    &\leq \norm{ \cU_{p-2} \ldots \cU_1 (\rho \otimes \ketbra{0^m}{0^m}) - \tilde{\cU}_{p-2} \ldots \tilde{\cU}_1 (\rho \otimes \ketbra{0^m}{0^m}) }_1 + \norm{ \cU_{p-1}  - \tilde{\cU}_{p-1} }_\diamond + \norm{ \cU_p  - \tilde{\cU}_{p} }_\diamond \\
    &\leq \ldots\\
    &\leq \sum_{i=1}^p \norm{ \cU_i  - \tilde{\cU}_{i} }_\diamond \\
    &\leq p \tilde{\epsilon} = \epsilon'.
\end{align}
The first inequality uses the fact that taking partial trace does not increase the trace norm.
The second inequality uses $\norm{A - B} \leq \norm{A - C} + \norm{C - B}$.
The third inequality uses $\norm{\cE(X)}_1 \leq \norm{X}_1$ for any CPTP map $\cE$, and $\norm{\cE(\rho)}_1 \leq \norm{\cE}_\diamond \norm{\rho}_1 = \norm{\cE}_\diamond$.
The fourth inequality considers the same steps taken in the second and third inequality.
Then, using induction, we obtain the formula given in the second-to-last line.
The last line uses the fact that $\norm{\cU_i - \tilde{\cU}_i}_{\diamond} \leq \tilde{\epsilon}$ for all $i$ and $\tilde{\epsilon} = \epsilon' / p$.

From the above analysis, we can see that we can find an $\epsilon'$-covering net $\cN_{\epsilon', n, m, p}$ for the space of $\cE$ with an $n$-qubit input state, $m$ ancillary qubits, and $p$ two-qubit gates that satisfies
\begin{equation} \label{eq:Neps-sizeupp}
    \left| \cN_{\epsilon', n, m, p} \right| \leq \left[ \binom{n+m}{2} \left( \frac{p c_1 }{\epsilon'} \right)^{c_2} \right]^p.
\end{equation}
For any $\cE$ in the space, we can find an $\tilde{\cE} \in \cN_{\epsilon', n, m, p}$ such that for all $n$-qubit input states $\rho$ we have
\begin{equation}
    \norm{\cE(\rho) - \tilde{\cE}(\rho) }_1 \leq \epsilon'.
\end{equation}
We will then utilize the $\epsilon$-covering net $\cN_{\epsilon', n, m, p}$ in the subsequent proof.
An $\epsilon$-covering net of quantum processes have also been used in \cite{huang2021information} to establish an information-theoretic bound on quantum advantage in \cite{caro2021generalization} to analyze generalization performance of quantum neural networks.

\newcommand{\Nin}{N_{\mathrm{in}}}
\newcommand{\Nout}{N_{\mathrm{out}}}

\subsubsection{Learning via Hypothesis Selection: Protocol and Analysis}

We will sample $\Nin$ input states $\rho_1,\ldots,\rho_{\Nin}$ from the distribution $\cD$.  For each $i\in[\Nin]$ and every $\wt{\cE}_k\in\cN_{\epsilon',n,m,p}$ (for $\epsilon'$ to be tuned later), we will access the true process $\cE$ a number of times $\Nout$ using $\rho_i$ as the input state, obtaining $\Nout$ copies of $\cE(\rho_i)$. We will store these $\Nin\cdot\Nout$ states in the quantum memory and run a known algorithm for \emph{quantum hypothesis selection} \cite{buadescu2020improved} to determine for which $k$ the product state $\bigotimes^{\Nin}_{i = 1}\wt{\cE}_k(\rho_i)$ is approximately closest to $\bigotimes^{\Nin}_{i=1}\cE(\rho_i)$. We will argue that if $\epsilon'$ is sufficiently small and $\Nin$ sufficiently large, then the index $k$ that we find will satisfy \begin{equation}
    \E_{\rho\sim\cD}\left[\norm{\cE(\rho) - \wt{\cE}_k(\rho)}_1\right] \le \epsilon\,, \label{eq:desired}
\end{equation} as desired.

We now proceed to the analysis of this protocol. We begin with an estimate for the distance between product states whose components are pairwise far from each other.

\begin{lemma}\label{lem:tensordiff}
    If $\rho_1,\ldots,\rho_N$ and $\rho'_1,\ldots,\rho'_N$ satisfy $\frac{1}{N}\sum^N_{i = 1} \norm{\rho_i - \rho'_i}_1 \ge \epsilon$, then
    \begin{equation}
        \norm{\bigotimes_i \rho_i - \bigotimes_i \rho'_i}_1 \ge 2\left( 1 - (1 - \epsilon^2/4)^{N/2} \right).
    \end{equation}
\end{lemma}

\begin{proof}
    For convenience, denote $\epsilon_i := \norm{\rho_i - \rho'_i}_1$. We have that $F(\rho_i, \rho'_i) \le 1 - \epsilon^2_i/4$, so 
    \begin{align}
        \norm{\bigotimes_i \rho_i - \bigotimes_i \rho'_i}_1 &\ge 2\left(1 - \sqrt{F\left(\bigotimes_i \rho_i, \bigotimes_i \rho'_i\right)}\right) \\
        &\ge 2\left(1 - \sqrt{\prod^N_{i=1}(1 - \epsilon^2_i / 4)} \right) \\
        &\ge 2\left( 1 - \left(\frac{1}{N}\sum_i (1 - \epsilon^2_i/4)\right)^{N/2} \right) \\
        &= 2\left( 1 - (1 - \epsilon^2/4)^{N/2} \right)
    \end{align}
    where the first step follows by the standard inequality $\norm{\rho - \rho'}_1 \ge 2\sqrt{1 - F(\rho,\rho')}$\,, the second step follows by tensorization of fidelity, the third step follows by AM-GM and the fact that $\epsilon_i = \norm{\rho_i - \rho'_i}_1 \le 2$, and the last step follows from the assumption.
\end{proof}

Next, we elaborate on how to select $\Nin$. Consider any $\wt{\cE}\in\cN_{\epsilon',n,m,p}$. By Hoeffding's inequality, because $\rho_1,\ldots,\rho_{\Nin}$ are sampled independently and identically distribued from the distribution $\cD$, and $\norm{\rho - \rho'} \le 2$ for all density matrices $\rho,\rho'$ we have
\begin{equation}
    \left|\frac{1}{N_{\mathrm{in}}} \sum_{i=1}^{N_{\mathrm{in}}} \norm{ \wt{\cE}(\rho_i) - \cE(\rho_i) }_1 - \E_{\rho\sim\cD}\norm{\wt{\cE}(\rho) - \cE(\rho)}_1\right| \le \epsilon/2 \label{eq:conc}
\end{equation}
with probability at least $1 - \delta'$ provided $\Nin = \Omega(\log(1 / \delta') / \epsilon^2 )$. By a union bound over $\cN_{\epsilon',n,m,p}$, the above bound holds simultaneously for all $\wt{\cE}\in\cN_{\epsilon',n,m,p}$ with probability at least $1 - |\cN_{\epsilon',n,m,p}|\delta'$. Henceforth, we condition on this event holding.

The following shows that it suffices to find $\wt{\cE}\in\cN_{\epsilon',n,m,p}$ for which the product state $\bigotimes_{i=1}^{\Nin} \wt{\cE}(\rho_i)$ is sufficiently close to $\bigotimes_{i=1}^{N_{\mathrm{in}}} \cE(\rho_i)$.

\begin{lemma}\label{lem:suffice}
    If \eqref{eq:conc} holds for all $\wt{\cE}\in\cN_{\epsilon',n,m,p}$, then if 
    \begin{equation}
        \norm{ \bigotimes_{i=1}^{\Nin} \wt{\cE}(\rho_i) - \bigotimes_{i=1}^{N_{\mathrm{in}}} \cE(\rho_i) }_1 \le 2\left( 1 - (1 - \epsilon^2/16)^{\Nin / 2} \right) \label{eq:suffice}
    \end{equation}
    for some $\wt{\cE}\in\cN_{\epsilon',n,m,p}$, we have that $\E_{\rho\sim\cD}\norm{\wt{\cE}(\rho) - \cE(\rho)}_1 \le \epsilon$.
\end{lemma}

\begin{proof}
    We prove the contrapositive. Suppose $\E_{\rho\sim\cD}\norm{\wt{\cE}(\rho) - \cE(\rho)}_1 > \epsilon$. Then by \eqref{eq:conc}, we have \begin{equation}
        \frac{1}{N_{\mathrm{in}}} \sum_{i=1}^{N_{\mathrm{in}}} \norm{ \wt{\cE}(\rho_i) - \cE(\rho_i) }_1 \ge \epsilon/2.
    \end{equation} The lemma follows from Lemma~\ref{lem:tensordiff}.
\end{proof}

As we show in the next lemma, there exists an $\wt{\cE}\in\cN_{\epsilon',n,m,p}$, namely the process in the covering net which is closest to $\cE$, for which \eqref{eq:suffice} holds, provided $\epsilon'$ is sufficiently small.

\begin{lemma}\label{lem:existsgood}
    For any $\epsilon' > 0$, there exists an $\wt{\cE}\in\cN_{\epsilon',n,m,p}$ for which \begin{equation}
        \norm{ \bigotimes_{i=1}^{\Nin} \wt{\cE}(\rho_i) - \bigotimes_{i=1}^{N_{\mathrm{in}}} \cE(\rho_i) }_1 \le \Nin\epsilon'.
    \end{equation}
\end{lemma}

\begin{proof}
    Take $\wt{\cE}\in\cN_{\epsilon',n,m,p}$ satisfying $\norm{\wt{\cE}(\rho) - \cE(\rho)}_1\le \epsilon'$ for all $\rho$. For convenience, let $\sigma_i := \cE(\rho_i)$, $\sigma'_i := \wt{\cE}(\rho_i)$, and $\delta_i := \sigma'_i - \sigma_i$ for all $i\in \{1, \ldots, \Nin\}$. Then
    \begin{equation}
        \bigotimes^{\Nin}_{i=1} \sigma'_i - \bigotimes^{\Nin}_{i=1} \sigma_i = \sum^{\Nin}_{i = 1} \left[\left(\bigotimes^{i-1}_{j = 1}\sigma'_j\right)\otimes (\sigma'_i - \sigma_i) \otimes \left(\bigotimes^{\Nin}_{j = {i+1}} \sigma_j\right)\right],
    \end{equation}
    so by the triangle inequality we conclude that $\norm{\bigotimes_i \sigma'_i - \bigotimes_i \sigma_i}_1 \le \Nin \epsilon'$ as claimed.
\end{proof}

Lemma~\ref{lem:suffice} and Lemma~\ref{lem:existsgood} guarantee the existence of a process in $\cN_{\epsilon',n,m,p}$ satisfying the desired bound of \eqref{eq:desired}. To complete the proof of Theorem~\ref{thm:poly-upp-learn-proc}, we will use the following special case of a result from \cite{buadescu2020improved} to find a process in the covering net which performs comparably.

\begin{theorem}[\cite{buadescu2020improved}, Theorem 1.5]\label{thm:data_analysis}
    Suppose we are given $m$ fixed hypothesis states $\sigma_1,\ldots,\sigma_M\in\mathbb{C}^{d\times d}$, parameters $0 < \epsilon,\delta < 1/2$, and access to copies of a state $\rho\in\mathbb{C}^{d\times d}$.  Then there is an algorithm that uses
    \begin{equation}
        N = \mathcal{O}\left(\frac{1}{\epsilon^2}\left(\log^3 M + \alpha\log M\right)\cdot \alpha\right)
    \end{equation} copies of $\rho$ for $\alpha:= \log(\log(1/\eta)/\delta)$ and $\eta:= \min_i \norm{\rho - \sigma_i}_1$ such that with probability at least $1 - \delta$ the algorithm outputs a $k\in \{1, \ldots, M\}$ for which $\norm{\rho - \sigma_k}_1 \le 4\eta$.
\end{theorem}

We can now put together the ingredients assembled in this section to complete the proof of Theorem~\ref{thm:poly-upp-learn-proc}.

\begin{proof}[Proof of Theorem~\ref{thm:poly-upp-learn-proc}]
    We first prove the theorem for constant $\delta$. Take $\epsilon' = c/\Nin$ for a sufficiently small constant $c > 0$, $\delta' = 1/(10|\cN_{\epsilon',n,m,p}|)$, and $\Nin = \wt{\mathcal{O}}(p/\epsilon^2)$, where $\wt{\mathcal{O}}(\cdot)$ hides factors of $\log n, \log m, \log p, \log 1/\epsilon$, so that $\Nin \ge \Omega(\log(1/\delta')/\epsilon^2)$ and \eqref{eq:conc} holds for all $\wt{\cE}\in\cN_{\epsilon',n,m,p}$ with probability at least $4/5$. Note that for some absolute constant $C > 0$,
    \begin{equation}
        1 - (1 - \epsilon^2/16)^{\Nin / 2} \ge 1 - e^{-\Nin\epsilon^2/32} = 1 - {\delta'}^C \ge \Omega(1)\,. \label{eq:big}
    \end{equation}
    In contrast, by Lemma~\ref{lem:existsgood} and our choice of $\epsilon'$, there exists some $\wt{\cE}\in\cN_{\epsilon',n,m,p}$ for which 
    \begin{equation}
        \norm{ \bigotimes_{i=1}^{\Nin} \wt{\cE}(\rho_i) - \bigotimes_{i=1}^{N_{\mathrm{in}}} \cE(\rho_i) }_1 \le c\,. \label{eq:small}
    \end{equation} Applying Theorem~\ref{thm:data_analysis} to $\sigma_k = \bigotimes_i \wt{\cE}_k(\rho_i)$ where $\wt{\cE}_k$ is the $k$-th element of $\cN_{\epsilon',n,m,p}$, using $\Nout$ copies of $\rho = \bigotimes_i \cE(\rho_i)$ where \begin{equation}
        \Nout = \mathcal{O}\left(\frac{1}{\epsilon^2}\left(\log^3|\cN_{\epsilon',n,m,p}| + \alpha\log|\cN_{\epsilon',n,m,p}|\right)\cdot \alpha\right), \qquad \alpha := \mathcal{O}(\log\log(1/c)), \label{eq:Ndef}
    \end{equation} we can output a $k$ for which $\norm{\sigma_k - \rho}_1 \le 4c$ with probability $4/5$. By taking the constant $c$ sufficiently small, we can leverage \eqref{eq:big} and Lemma~\ref{lem:suffice} to conclude that \begin{equation}
        \E_{\rho\sim\cD}\left[\norm{\cE(\rho) - \wt{\cE}_k(\rho)}_1\right] \le \epsilon\,.
    \end{equation}
    Note that $\Nout$ in \eqref{eq:Ndef} is dominated by the $\epsilon^{-2}\log|\cN_{\epsilon',n,m,p}|$ term since $\alpha = \mathcal{O}(1)$, and so recalling \eqref{eq:Neps-sizeupp} and our choice of $\epsilon' = c/\Nin$ we obtain
    \begin{equation}
        \Nout = \mathcal{O}\left(\frac{p^3}{\epsilon^2}\log^3\left((n + m)p\Nin\right)\right) = \mathcal{O}\left(\frac{p^3}{\epsilon^2}\log^3\left((n + m)p\Nin\right)\right) = \wt{\mathcal{O}}\left(\frac{p^3}{\epsilon^2}\right),
    \end{equation} where again $\wt{\mathcal{O}}(\cdot)$ hides logarithmic factors in $n, m, p, 1/\epsilon$.

    As we require $\Nout$ copies of $\bigotimes_{i=1}^{N_{\mathrm{in}}} \cE(\rho_i)$, we must make $\Nout\cdot \Nin = \wt{\mathcal{O}}(p^4/\epsilon^4)$ accesses to $\cE$. By union bounding over \eqref{eq:conc} holding for all $\wt{\cE}\in\cN_{\epsilon',n,m,p}$ and over the success of the algorithm in Theorem~\ref{thm:data_analysis}, we obtain Theorem~\ref{thm:poly-upp-learn-proc} for $\delta = 2/5$ from the assumption that $p = \mathrm{poly}(n)$.
    
    We now describe how to extend this result to general $\delta$ by a standard clustering argument. We can run $r := \Theta(\log(1/\delta))$ independent copies of the above protocol, resulting in indices $k_1,\ldots,k_r$ into $\cN_{\epsilon',n,m,p}$ such that for any fixed $i\in[r]$, $\norm{\sigma_{k_i} - \rho}_1 \le 4c$ with probability at least $3/5$. If $S\subset[r]$ denotes the set of $i\in[r]$ for which $\norm{\sigma_{k_i} - \rho}_1 \le 4c$, then by a Chernoff bound, $|S| \ge r/2$ with probability at least $1 - \delta$ provided the constant factor in the definition of $r$ is sufficiently large. Condition on the event that $|S| \ge r/2$.
    
    Let $k$ be an index into $\cN_{\epsilon',n,m,p}$ for which there are at least $r/2$ indices $i\in[r]$ for which $\norm{\sigma_k - \sigma_{k_i}}_1 \le 8c$, and output the channel $\wt{\cE}_k$. Such a $k$ certainly exists: take any $i\in S$ and note that by triangle inequality, for any other $j\in S$ we have \begin{equation}
        \norm{\sigma_{k_i} - \sigma_{k_j}}_1 \le \norm{\sigma_{k_i} - \rho}_1 + \norm{\sigma_{k_j} - \rho}_1 \le 8c.
    \end{equation}
    Now observe that regardless of which $k$ we choose that meets the criterion that at least $r/2$ indices $i\in[r]$ satisfy $\norm{\sigma_k - \sigma_{k_i}}_1 \le 8c$, we must have \begin{equation}
        \norm{\sigma_k - \rho}_1 \le 12c. \label{eq:cluster_close}
    \end{equation} Indeed, suppose to the contrary. Then for any $i\in S$, \begin{equation}
        \norm{\sigma_k - \sigma_{k_i}}_1 \ge \norm{\sigma_k - \rho}_1 - \norm{\sigma_{k_i} - \rho}_1 > 12c - 4c = 8c,
    \end{equation}
    where the second step is by the definition of $S$ and the assumption that \eqref{eq:cluster_close} does not hold. As $|S| \ge r/2$, this yields a contradiction of the fact that there are at least $r/2$ indices $i\in[r]$ for which $\norm{\sigma_k - \sigma_{k_i}}_1 \le 8c$.
    
    If we take the constant $c$ sufficiently small, then \eqref{eq:cluster_close} together with \eqref{eq:big} and Lemma~\ref{lem:suffice} allow us to conclude that $\E_{\rho\sim\cD}\left[\norm{\cE(\rho) - \wt{\cE}_k(\rho)}_1\right] \le \epsilon$. As we ran $r := \Theta(\log(1/\delta))$ independent copies of the protocol that we used for constant failure probability $\delta$, we merely incur an additional $\Theta(\log(1/\delta))$ multiplicative overhead in the number of access to $\cE$ we must make for general failure probability $\delta$.
\end{proof}

\section{Predicting highly-incompatible observables with bounded quantum memory}
\label{sec:hard-bdd-mem}

Here we will substantively generalize our results for predicting highly-incompatible observables, given in Appendix~\ref{sec:shadow}.
We show an exponential lower bound on the number of experiments when the size of the additional quantum memory is not large enough.

\subsection{Background and statement of results}
\label{sec:background1}

Let us first recapitulate the setting of our previous results, so as to draw a contrast with our generalization.  We have so far considered an experimentalist who is given sequential access to copies of an unknown state $\rho$.  In each measurement round, the experimentalist receives a copy of $\rho$, and can measure with a POVM.  The residual post-measurement state is then discarded, and only the classical data of the POVM outcome is kept.  This classical data can be used to inform the choice of POVM measurement employed in subsequent rounds, i.e.~the protocol can be adaptive.  This kind of protocol is emblematic of most contemporary and historical experiments in physics.

Note that the information maintained and processed from round to round in the protocols described above is solely classical.  With the advent of quantum computers and more flexible quantum memory architectures, a new possibility emerges.  Suppose that the unknown state $\rho$ in question is an $n$-qubit state.  Moreover, suppose we have $n+k$ qubit registers under our control.  Then we can use those registers however we please, including performing arbitrary quantum information processing.  Our only constraint is that each time we receive a new copy of $\rho$, we must necessarily use $n$ qubits of our registers to hold it.  It is thus appropriate to say that we have $k$ qubits of quantum memory, since even when we receive a new state $\rho$ we can still maintain $k$ qubits worth of quantum data.  We further allow ourselves to maintain and process an arbitrary amount of classical data, thought of as being stored in a classical device external to our quantum system.

A new question immediately presents itself: are there experimental tasks which are exponentially hard with only $k$ qubits of quantum memory, but easy with $k' > k$ qubits of quantum memory?  In~\cite{chen2021exponential} this question was answered in the affirmative, but only in the sense of query complexity.  That is, it was shown that there is an experimental task which requires $\Omega(2^{(n-k)/3})$ copies of $\rho$ if there are only $k$ qubits of quantum memory; however, the gate complexity of achieving the task is \textit{always} exponential regardless of the size of $k$.  Let us unpack this result.  Note that if $k = n$, the bound $\Omega(2^{(n-k)/3})$ becomes trivial; indeed, it can be shown that one only requires a modest number of copies of $\rho$ to achieve the specified task if $n = k$.  But even in this case, the total number of quantum operations required is exponentially large in $n$.  Nonetheless, the result is an interesting one: it means that unless $k$ goes as $n - c\, \log(n)$ (i.e.~unless $k$ is logarithmically close to $n$), the task is exponentially hard.  When $k \sim n - c \,\log(n)$, the task is only polynomially hard in the sense of query complexity.

While the aforementioned result is theoretically interesting, it does not correspond to a quantum memory advantage that could be realized by a quantum device on account of the exponential gate complexity required to achieve the task for a quantum memory of \textit{any} size.  Here we ameliorate this issue, and present the first example of a quantum memory advantage in the sense of both query \textit{and} gate complexity, and as such it can be realized on a quantum device.  Moreover, we have realized this quantum advantage in our experiments on the Sycamore quantum computer.

Our experimental task has the form of a partially-revealed many-versus-one distinguishing task, closely related to the one in Appendix~\ref{sec:shadow}.  A statement of the new task is as follows:
\begin{task}[Expectation value with bounded quantum memory] \label{task:bounded}
There is an unknown state $\rho$ which is either
\begin{enumerate}
    \item A maximally mixed state $I/2^n$ on $n$ qubits, or
    \item The state $(I + P)/2^n$ where $P \in \{I, X, Y, Z\}^{\otimes n} \setminus \{I^{\otimes n}\}$ is a random but fixed Pauli string.
\end{enumerate}
The choice of whether case 1 or case 2 is instantiated is made with equal probability at the outset, and is not revealed.  The experimentalist is given access to $T$ copies of the unknown state $\rho$ for a $T$ decided by the experimentalist, and after this an observable $O$ is revealed.  The task of the experimentalist is to determine the value of $|\text{\rm tr}(O \rho)|$ using the final state of the $n+k$ qubit registers, along with any classical information that has been stored or processed along the way.  In case 1 the operator $O$ is chosen uniformly at random from the non-identity Pauli strings; in case 2 the $O$ is chosen to be the Pauli operator $P$ if the state $\rho$ is $(I + P)/2^n$. 
\end{task}
\noindent Note that if $k = n$ so that the total number of registers is $2n$, then the task can be readily solved using the algorithm given in Appendix~\ref{sec:uppbd-obs}.  This algorithm is both query and gate efficient: we only require a constant number of copies of $\rho$ (i.e.~the number of copies is independent of $n$) and $O(n)$ gate complexity.

What is difficult is to show that if $k < n$, then the number of copies of $\rho$ we require to determine $|\text{tr}(O \rho)|$ as per the task above is $\Theta(2^{(n-k)/3})$.  We will establish this in the subsections which follow below.

\subsection{Review of learning tree framework for bounded quantum memories}
\label{sec:revew-bdd-qmem}

\begin{figure}[t]
    \centering
    \includegraphics[width=0.7\textwidth]{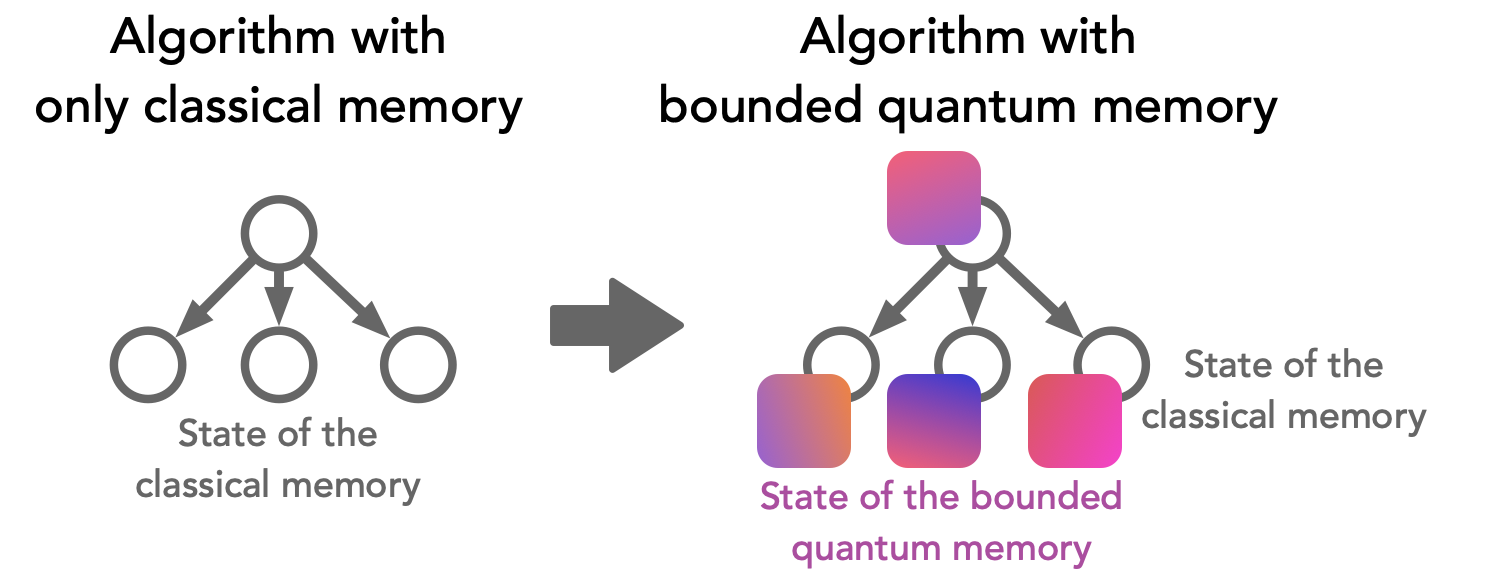}
    \caption{\emph{Illustration of learning tree representation for algorithms with bounded quantum memory.} We consider algorithms with unlimited classical memory and a bounded quantum memory consisting of $k$ qubits. To each node in the tree (corresponding to the state of the classical memory), we associate the $k$-qubit state of the bounded quantum memory. }
    \label{fig:bdd-qmem}
\end{figure}

Here we provide an exposition of the learning tree framework for bounded quantum memories in~\cite{chen2021exponential}.  As explained above, suppose we have $n+k$ qubit registers, where we designate $k$ as the quantum memory.  Suppose for the moment that $k < n$.  At each round in the protocol, we receive an $n$-qubit state $\rho$, which we must hold on the $n$ non-memory registers.  (Note that we cannot receive more than one copy of $\rho$, since we do not have enough registers to hold additional copies on account of $k < n$.)  Then upon receiving and holding $\rho$, the state of all of our registers can be written as $\rho \otimes \boldsymbol{\Sigma}$, where $\boldsymbol{\Sigma}$ is the density matrix of the $k$-qubit quantum memory.  The most general operation we can perform on the joint system is a quantum process, followed by a POVM measurement; we can then apply another quantum process followed by another POVM measurement, and so on.  However, we can rewrite an alternating sequence of quantum processes and POVM measurements as a single POVM measurement, which we denote by $\{F_i\}_{i=1}^{N-1}$.  Writing $F_i = M_i^\dagger M_i$, suppose our measurement outputs the $i$th POVM element.  Defining
\begin{equation}
A_{M_i}^\rho(\boldsymbol{\Sigma}) := \text{tr}_{1,...,n}\!\left(M_i (\rho \otimes \boldsymbol{\Sigma}) M_i^\dagger\right)\,,
\end{equation}
the reduced density matrix of our quantum memory is
\begin{equation}
\label{E:newmemory1}
\frac{A_{M_i}^\rho(\boldsymbol{\Sigma})}{ \text{tr}(A_{M_i}^\rho(\boldsymbol{\Sigma}))}
\end{equation}
with probability $\text{tr}(A_{M_i}^\rho(\boldsymbol{\Sigma}))$.  In the next round, we can leverage our measurement output $i$ to adaptively inform our choice of POVM on
\begin{equation}
\rho \otimes \frac{A_{M_i}^\rho(\boldsymbol{\Sigma})}{ \text{tr}(A_{M_i}^\rho(\boldsymbol{\Sigma}))}\,.
\end{equation}
Indeed, we can use the information of \textit{all} of our previous POVM outcomes to inform the choice of our next POVM.
An illustration is given in Supp.~Fig.~\ref{fig:bdd-qmem}.

The above description is ripe for being cast in the learning tree framework, which we presently articulate.  The definition below is the same as Definition 6.1 of~\cite{chen2021exponential}, albeit with slightly different notation.
\begin{definition21}[Tree representation of learning states with bounded quantum memory]
Let $\rho$ be a fixed, unknown quantum density matrix on $n$ qubits.  Suppose we have access to $n+k$ qubit registers.  A learning algorithm with a quantum memory of size $k$ can be expressed as a rooted tree $\mathcal{T}$ of depth $T$, where each node encodes the current state of the quantum memory in addition to the transcript of measurement outcomes the algorithm has seen so far.  In particular, the tree satisfies the following properties:
\begin{enumerate}
    \item Each note $u$ is associated with a $k$-qubit unnormalized density matrix $\boldsymbol{\Sigma}^\rho(u)$ corresponding to the current state of the quantum memory.
    \item For the root $r$ of the tree, $\boldsymbol{\Sigma}^\rho(r)$ is an initial state denoted by $\boldsymbol{\Sigma}_0$.
    \item At each note $u$, we apply a POVM measurement $\{F_s^u\}_s$ on $\rho \otimes \boldsymbol{\Sigma}^\rho(u)$ to obtain a classical outcome $s$.  Each child node $v$ of $u$ is connected through the edge $e_{u,s}$.
    \item If $v$ is the child note of $u$ connected through the edge $e_{u,s}$, then letting $F_s^u = M_s^{u\,\dagger} M_s^u$ we have
    \begin{equation}
    \boldsymbol{\Sigma}^\rho(v) := A_{M_s^u}(\boldsymbol{\Sigma}^\rho(u))\,.
    \end{equation}
    \item Note that for any node $u$ we have that $p^\rho(u) := \text{\rm tr}(\boldsymbol{\Sigma}^\rho(u))$ is the probability that the transcript of measurement outcomes observed by the learning algorithm is described by $u$.  Moreover, $\boldsymbol{\Sigma}^\rho(u)/p^\rho(u)$ is the (normalized) state of the $k$-qubit memory at the node $u$.
\end{enumerate}
\end{definition21}

Let us unpack the ingredients of this definition.  The initial state of our quantum memory is $\boldsymbol{\Sigma}_0$, and we apply some initial choice of POVM $\{F_s^r\}_s$ (where $r$ denotes the `root' of the tree).  If we measure the $s$th POVM outcome, then the quantum memory is in the unnormalized state $A_{M_s^r}(\boldsymbol{\Sigma}_0)$ with probability $\text{tr}(A_{M_s^r}(\boldsymbol{\Sigma}_0))$.  Each outcome $s$ of the POVM corresponds to a child note of the root; thus at the next level of the tree, each node is labelled by the POVM outcome $s$ and the corresponding state of the quantum memory $A_{M_s^r}(\boldsymbol{\Sigma}_0) := \boldsymbol{\Sigma}^\rho(s)$.  For the next measurement, we can leverage our knowledge of the previous POVM to craft a new POVM to be applied to the present state of the quantum memory.  This type of procedure is repeated for many rounds.

To be explicit, suppose that the present state of the quantum memory is $\Sigma^\rho(u)$, where the node $\rho$ reflects a sequence or \textit{transcript} of POVM outcomes which brought us to the present state by an adaptive protocol.  We can use this transcript of previous outcomes to choose a new POVM $\{F_s^u\}_s$ that we use to measure $\rho \otimes \boldsymbol{\Sigma}^u(\rho)$, which will result in the output $A_{M_s^u}(\boldsymbol{\Sigma}^\rho(u))$ with probability $\text{tr}(A_{M_s^u}(\boldsymbol{\Sigma}^\rho(u)))$.  The nodes $v$ in the next layer encode the data of the previous measurement outcomes and the latest outcome (i.e., determined by the location of $v$ in the tree), as well as the new (conveniently unnormalized) state of the quantum memory $A_{M_s^u}(\boldsymbol{\Sigma}^\rho(u)) := \boldsymbol{\Sigma}^\rho(v)$, where here $v$ is connected to $u$ by an edge $e_{u,s}$ (designating that $v$ is the consequence of the $s$th measurement outcome starting from the configuration in $u$).

\subsection{Hardness result for small quantum memories}
\label{sec:hardness-smallqmem}

We will prove the following result using the learning tree framework:
\begin{theorem}[Shadow tomography with partial reveal using a bounded quantum memory] Consider Task~\ref{task:bounded} for learning an expectation value with a bounded quantum memory.  Any learning algorithm with $n + k$ qubit registers needs $T \geq \Omega(2^{(n-k)/3})$ copies of $\rho$ to determine $|\text{\rm tr}(O \rho)|$ with probability at least $2/3$.\label{thm:memory1}
\end{theorem}
\noindent On account of~\eqref{eq:TVbound-rev}, it suffices to upper bound $\mathbb{E}_P\!\left[\text{TV}(p_{I/2^n}, p_{(I + P)/2^n})\right]$.  To do so, we will leverage a key technical result coming from Theorem 1.4 of~\cite{chen2021exponential}.  But in order to state this technical result, we first need to introduce the notion of \textit{good Paulis} and \textit{bad Paulis}.  While details are provided in Definition 6.4 of~\cite{chen2021exponential}, here we explain the essential intuition and key properties.

In the learning protocol, we are trying to distinguish between the maximally mixed state $I/2^n$ and states of the form $(I + P)/2^n$.  The intuition is that if the size of our quantum memory $k$ is small relative to $n$, then it is hard to tell the two kinds of states apart.  If this was the case, then any round of the protocol should only reveal a very small amount of distinguishing information.  In particular, suppose we are at node $u$ in the learning tree, and so the state of the quantum memory at that node is either $\boldsymbol{\Sigma}_{I/2^n}(u)$ or $\boldsymbol{\Sigma}_{(I+P)/2^n}(u)$ for some Pauli string $P$.  If we consider a POVM $\{F_s^u\}_s$ where $F_{s}^u = M_s^{u\,\dagger} M_s^u$, then if we measure some fixed outcome $s$ the new state of the quantum memory will be either $A_{M_s^u}^{I/2^n}\!(\boldsymbol{\Sigma}_{I/2^n}(u))$ or $A_{M_s^u}^{(I+P)/2^n}\!(\boldsymbol{\Sigma}_{(I+P)/2^n}(u))$.  We would like for $\left\|A_{M_s^u}^{I/2^n}\!(\boldsymbol{\Sigma}_{I/2^n}(u)) - A_{M_s^u}^{(I+P)/2^n}\!(\boldsymbol{\Sigma}_{(I+P)/2^n}(u))\right\|_1$ to be exponentially small in $n-k$, in particular relative to some distinguishing measure between $\boldsymbol{\Sigma}_{I/2^n}(u)$ and $\boldsymbol{\Sigma}_{(I+P)/2^n}(u)$.  This would mean that starting from node $u$, the next POVM measurement will not significantly change our ability to distinguish the two possibilities for the resulting memory registers.  While we cannot guarantee that such a property holds for all states $(I+P)/2^n$, such a property will hold for some $P$'s.  Given a node $u$, the set of good Paulis $P[u]$ is the set of all Pauli operators satisfying a particular version of the above property for \textit{all} edges from the root of the tree to $u$ (see Definition 6.4 of~\cite{chen2021exponential} for details).  The residual Paulis are called the set of bad Paulis.  In other words, the good Paulis $P[u]$ designate the states $(I+P)/2^n$ which are hard to distinguish from $I/2^n$ for a particular instantiation of the learning tree, specifically for the sequence of POVMs that get us from the root of said learning tree to the node $u$.  By contrast, the bad Paulis reveal too much information.

We have the following useful Lemma about bad Paulis, which we will soon leverage in the proof of Theorem~\ref{thm:memory1}:
\begin{lemma}[Fact 6.5 of~\cite{chen2021exponential}]
For any edge $e_{u,s}$, there are at most $2^{-(n-k)/3}\cdot(4^n - 1)$ bad Paulis $P$.  In particular, along any root-to-leaf path of the learning tree, there are at most $T \cdot 2^{-(n-k)/3}\cdot (4^n-1)$ Paulis which are bad for some edge along the path.
\label{lemm:fact6p5}
\end{lemma}

Equipped with our discussion of good and bad Paulis, we can now state the following technical result from~\cite{chen2021exponential}:
\begin{lemma}[Following from the proof of Theorem 1.4 of~\cite{chen2021exponential}] We have the inequality
\begin{equation}
\frac{1}{4^n - 1} \sum_{\ell \in \text{\rm leaf}(\mathcal{T})} \sum_{P \in P[\ell]} \left\| \boldsymbol{\Sigma}_{I/2^n}(\ell) - \boldsymbol{\Sigma}_{(I+P)/2^n}(\ell)\right\|_1 \leq T \cdot 2^{-(n-k)/3} \cdot \sqrt{\frac{2^{2n}}{2^{2n}-1}}\,.
\end{equation}
\label{lemm:following}
\end{lemma}
\noindent We are finally ready to prove Theorem~\ref{thm:memory1}.
\begin{proof}[Proof of Theorem~\ref{thm:memory1}]
Let us upper bound $\mathbb{E}_P\!\left[\text{TV}(p_{I/2^n}, p_{(I + P)/2^n})\right]$.  We have the inequalities
\begin{align}
\label{E:Ineq1}\mathbb{E}_P\!\left[\text{TV}(p_{I/2^n}, p_{(I + P)/2^n})\right] &\leq \mathbb{E}_{P}\!\left[\sum_{\ell}\max\!\left(0, p_{I/2^n}(\ell) - p_{(I + P)/2^n}(\ell)\right)\right] \\
\label{E:Ineq2}
&\leq \mathbb{E}_{P}\!\left[\sum_{\ell}\min\!\left(p_{I/2^n}(\ell),|p_{I/2^n}(\ell) - p_{(I + P)/2^n}(\ell)|\right)\right] \\
\label{E:Ineq3}
&\leq \mathbb{E}_{P}\!\left[\sum_{\ell}\min\!\left(p_{I/2^n}(\ell),\left\|\boldsymbol{\Sigma}_{I/2^n}(\ell) - \boldsymbol{\Sigma}_{(I+P)/2^n}(\ell)  \right\|_1\right)\right] \\
\label{E:Ineq4}
&\leq \sum_\ell \text{Pr}[P \not\in P[\ell]] \, p_{I/2^n}(\ell) + \frac{1}{4^n - 1}\sum_{P \in P[\ell]} \left\|\boldsymbol{\Sigma}_{I/2^n}(\ell) - \boldsymbol{\Sigma}_{(I+P)/2^n}(\ell)  \right\|_1\,.
\end{align}
In the first line, we have used that $\text{TV}(p,q) = \frac{1}{2}\sum_i |p_i - q_i| = \sum_{i \, : \, p_i \geq q_i}(p_i - q_i) = \sum_i \max(0, p_i - q_i)$.  To go from~\eqref{E:Ineq1} to~\eqref{E:Ineq2} we used $\max(0,a-b) \leq \min(a, |a-b|)$.  In going from~\eqref{E:Ineq2} to~\eqref{E:Ineq3} we leveraged that $|p_{I/2^n}(\ell) - p_{(I+P)/2^n}(\ell)| = |\text{tr}(\boldsymbol{\Sigma}_{I/2^n}(\ell) - \boldsymbol{\Sigma}_{(I+P)/2^n}(\ell))| \leq \| \boldsymbol{\Sigma}_{I/2^n}(\ell) - \boldsymbol{\Sigma}_{(I+P)/2^n}(\ell) \|_1$.  Finally, to go from~\eqref{E:Ineq3} to~\eqref{E:Ineq4} we used $\sum_{i \in S} \min(a_i, b_i) \leq \sum_{i \in S\setminus R} a_i + \sum_{i \in R} b_i$ for any $R \subseteq S$.

By Lemma~\ref{lemm:fact6p5} and the fact that $\sum_\ell p_{I/2^n}(\ell) = 1$, we have the simple bound
\begin{equation}
\sum_\ell \text{Pr}[P \not\in P[\ell]]\,p_{I/2^n}(\ell) \leq T \cdot 2^{-(n-k)/3}
\end{equation}
and Lemma~\ref{lemm:following} gives us
\begin{equation}
\frac{1}{4^n - 1}\sum_{P \in P[\ell]} \left\|\boldsymbol{\Sigma}_{I/2^n}(\ell) - \boldsymbol{\Sigma}_{(I+P)/2^n}(\ell)  \right\|_1  \leq T\cdot 2^{-(n-k)/3}\cdot \sqrt{\frac{2^{2n}}{2^{2n}-1}}\,.
\end{equation}
Then in total, we have
\begin{equation}
    \mathbb{E}_P\!\left[\text{TV}(p_{I/2^n}, p_{(I + P)/2^n})\right] \leq T\cdot 2^{-(n-k)/3}\left(1 + \sqrt{\frac{2^{2n}}{2^{2n}-1}} \right)\,.
\end{equation}
If the left-hand side is $\Omega(1)$, then we must thus have $T \geq \Omega(2^{(n-k)/3})$, as claimed.
\end{proof}

\end{document}